\documentclass[a4paper,UKenglish,cleveref, numberwithinsect, thm-restate]{lipics-v2021}
\hideLIPIcs  %uncomment to remove references to LIPIcs series (logo, DOI, ...), e.g. when preparing a pre-final version to be uploaded to arXiv or another public repository
\usepackage[dvipsnames,svgnames,table]{xcolor}
\usepackage{url}

\usepackage{enumitem}

\usepackage{mathtools}
\usepackage{tikz}
\usetikzlibrary{automata, positioning,tikzmark,fit,backgrounds,calc}

\usepackage{listings}
\usepackage[textsize=small]{todonotes}

\usepackage[notion, quotation, electronic]{knowledge}
\definecolor{Dark Ruby Red}{HTML}{a51c1f}
\definecolor{Dark Blue Sapphire}{HTML}{0e6072}
\definecolor{Dark Gamboge}{HTML}{be7c00}

\definecolor{Blue Sapphire}{HTML}{005f73} 
\definecolor{Blue Matt}{HTML}{094c7b} 
\definecolor{Blue Dark}{HTML}{1B4C6E} 
\definecolor{Golden}{HTML}{E27400}
\definecolor{Brown}{HTML}{ae740e}
\definecolor{Gamboge}{HTML}{ee9b00}
\definecolor{Ruby Red}{HTML}{9b2226}
\definecolor{Blue Marine}{HTML}{022687}

\IfKnowledgePaperModeTF{
	}{
		\knowledgestyle{intro notion}{color={Dark Ruby Red}, emphasize}
		\knowledgestyle{notion}{color={Dark Blue Sapphire}}
		\hypersetup{
				colorlinks=true,
				breaklinks=true,
				linkcolor={}, % Links to sections, pages, etc.
				citecolor={}, % Links to bibliography
				filecolor={Blue Marine}, % Links to local file
				urlcolor={Blue Marine},
			}
	}

\knowledge{notion}
| notion
| definition

\knowledge{notion}
| hyperlink
| Hyperlinks
| hyperlinks

\knowledge{notion}
| consistent
 | Consistent
 | consistency
 | Consistency

\knowledge{notion}
  | consistent with

\knowledge{notion}
 | $x$-SCC
 | $x$-SCCs

\knowledge{notion}
| final SCC

\knowledge{notion}
| accepting@run
| accepting run
| is accepting@run

\knowledge{notion}
| recognising@alt
| recognise@alt
| recognises@alt
| recognised@alt
| language recognised@alt
| languages recognised@alt

\knowledge{notion}
| $\omega $-regular 
| $\omega $-regular languages
 | $\omega $-regular language

\knowledge{notion}
| HD automaton
| history determinism
| history deterministic
| HD
| HD automata
 | History deterministic
 | History determinism
 | HDness

\knowledge{notion}
| resolver
| resolvers
| Eve-resolver

\knowledge{notion}
| valid resolver
 | valid resolvers
 | valid

\knowledge{notion}
| normalised
| normal form
| normalisation
| normality
| normal
| normalising
| Normalisation

\knowledge{notion}
 | layered automata
 | layered automaton
 | layered
 | layered automaton
 | Layered automata
 | automaton@lay

\knowledge{notion}
 | $x$-component
 | $x$-components
 | $(x+1)$-component

\knowledge{notion}
 | longest suffix resolver
 | longest suffix
 | longest suffix resolvers

 \knowledge{notion}
 | longest suffix resolver initialised
 | non-initialised resolver

 \knowledge{notion}
 | longest $x$-suffix of $u$ to $q$
 | longest $(x+1)$-suffix

\knowledge{notion}
 | centralised
 | centrality
 | central

\knowledge{notion}
 | $x$-central sequence
 | $x$-central sequences
 | $x$-central 
 | $x$-central for $p$

 \knowledge{notion}
 | central sequences
 | central@seq
 | central sequence

\knowledge{notion}
 | $[1,d]$-flower

\knowledge{notion}
| strongly accepted
 | strongly accepted by $p$
 | strongly accepts
 | strongly accept
 | strong acceptance

\knowledge{notion}
 | strongly rejected
 | strongly rejected by $p$
 | strongly reject
 | rejected@str
 | strongly rejects
 | rejected@strong

\knowledge{notion}
| pointed
| Pointed
| pointed classes

\knowledge{notion}
 | alternating $\omega $-automata
 | alternating automaton
 | alternating parity automata
 | alternating
 | alternating parity automaton

\knowledge{notion}
 | 0-1 probabilistic automata
 | 0-1 probabilistic
 | 0-1 probabilistic automaton

\knowledge{notion}
 | coB\"uchi automata
 | coB\"uchi automaton
 | coB\"uchi

\knowledge{notion}
 | subautomaton
 | sub-layered automaton

\knowledge{notion}
 | morphism@lay
 | morphisms@lay
 | morphisms of layered automata
 | morphism of layered automata
| morphism@layered

\knowledge{notion}
 | safe minimal
 | safe minimality
 | Safe minimality
 | safe minimise
 | safe minimising
 | safe minimisation

\knowledge{notion}
 | priority
 | priorities

 \knowledge{notion}
  | priority of
  | priority of action $a$
  | priority of an action

\knowledge{notion}
 | parity condition

\knowledge{notion}
 | transition system
 | transition systems

\knowledge{notion}
 | complete
 | completeness

\knowledge{notion}
 | run

 \knowledge{notion}
 | $\sigma $-run
 | $\tau $-run

\knowledge{notion}
 | morphism
 | morphisms
 | morphism of transition systems
 | morphism@TS
 | isomorphisms of transition systems
 | isomorphism of transition systems
 | isomorphism@TS
 | isomorphisms@TS
 | ismorphism@TS

\knowledge{notion}
 | initial state

\knowledge{notion}
 | accepted
 | accepts
 | accepting

\knowledge{notion}
 | language of a state
 | language@state

\knowledge{notion}
 | simple by priorities automaton
 | simple by priorities parity automaton
 | simple by priorities automata
 | simple by priorities

\knowledge{notion}
 | language accepted by an Eve-resolver

\knowledge{notion}
 | resolvers for Adam
 | Adam-resolver

\knowledge{notion}
 | semantically deterministic@alt
 | semantically deterministic
 | semantic determinism
 | SD
 | Semantic determinism

\knowledge{notion}
 | $x$-th layer

\knowledge{notion}
 | $x$-states
 | $x$-state
 | $(x+1)$-state
 | $(x+1)$-states

\knowledge{notion}
 | leaf state
 | leaf states

\knowledge{notion}
 | internal nodes
 | internal node
 | internal states

\knowledge{notion}
 | forest@representation
 | forest
 | tree of a layered automaton
 | tree
 | tree representation
 | tree structure

\knowledge{notion}
 | $x$-safe language
 | $x$-safe languages
 | safe language
 | safe languages
 | $x$-safe

 \knowledge{notion}
 | safe language@congr
 | safe languages@congr

\knowledge{notion}
 | $x$-determinism

\knowledge{notion}
| determinism
| deterministic
| deterministic@TS

 \knowledge{notion}
| Deterministic automata
| deterministic parity automaton
| deterministic parity automata
| DPA
| deterministic automaton

\knowledge{notion}
 | semantics of a
 | semantics of
 | semantics
 | semantics automaton

\knowledge{notion}
 | ${x}$-safe paths

\knowledge{notion}
 | automaton of residuals

\knowledge{notion}
 | safe deterministic

\knowledge{notion}
 | $1$-saturated

\knowledge{notion} 
 | $x$-saturation

\knowledge{notion}
 | $x$-coherence

\knowledge{notion}
 | run@lay

\knowledge{notion} 
 | rejected

\knowledge{notion}
 | equivalent

\knowledge{notion}
 | isomorphic@lay
 | isomorphism@lay
 | isomorphic
 | isomorphism

\knowledge{notion}
 | N1

\knowledge{notion}
 | N2

\knowledge{notion}
 | language equivalent
 | language equivalence

 \knowledge{notion}
 | covered

\knowledge{notion}
 | centralisation operation
 | operation@centralisation
 | Centralisation
 | centralisation
 | centralising

\knowledge{notion}
 | ancestor
 | ancestor relation

\knowledge{notion}
 | uniformly semantically deterministic
 | uniformly semantic deterministic
 | uniformly SD
 | uniform semantic determinism

\knowledge{notion}
 | acceptance@layered
 | accepted@layered
 | accepts@layered
 | recognising@lay
| recognise@lay
| recognises@lay
| recognised@lay
| language recognised@lay
| languages recognised@lay

\knowledge{notion}
 | ultimately safe

\knowledge{notion}
 | rejected@layered
 
 \knowledge{notion}
  | merging
  | Merging
  | merge

\knowledge{notion}
 | concatenation
 | Concatenation
 | concatenate
 | concatenating
 | concatenated

\knowledge{notion}
 | preserves strong acceptance
 | morphism that preserves strong acceptance

\knowledge{notion}
 | strong morphism@lay
| strong morphism
 | strong morphisms
 | Strong morphisms

\newrobustcmd{\rankex}[1]{\kl[\rankex]{\mathit{rank}_{\exists}^{#1}}}
\knowledge{\rankex}{notion}
\newrobustcmd{\rankax}[1]{\kl[\rankax]{\mathit{rank}_{\forall}^{#1}}}
\knowledge{\rankax}{notion}

\newrobustcmd{\Lsgeq}[1]{\kl[\Lsgeq]L_{\geq #1}}
\newrobustcmd{\LAsgeq}[2]{\kl[\LAsgeq]L^{#2}_{\geq #1}}
\knowledge{\Lsgeq}[\LAsgeq]{notion}

\newrobustcmd{\Lsem}[1]{\kl[\Lsem]L^{\mathrm{sem}}_{\geq #1}}
\knowledge{\Lsem}{notion}

\newrobustcmd{\seq}{\mathrel{\kl[\seq]\approx}}
\knowledge{\seq}{notion}

\newrobustcmd{\lang}{\kl[\lang]L}
\newrobustcmd{\langA}[1]{\kl[\langA]L^{#1}}
\knowledge{\lang}[\langA]{notion}

\newrobustcmd{\lsim}{\mathrel{\kl[\lsim]\sim}}
\knowledge{\lsim}{notion}

\newrobustcmd{\Lpgeq}[1]{\kl[\Lpgeq]L^+_{\geq #1}}
\knowledge{\Lpgeq}{notion}
\newrobustcmd{\Lmgeq}[1]{\kl[\Lmgeq]L^-_{\geq #1}}
\knowledge{\Lmgeq}{notion}
\newrobustcmd{\Lp}{\kl[\Lp]L^+}
\knowledge{\Lp}{notion}
\newrobustcmd{\Lm}{\kl[\Lm]L^-}
\knowledge{\Lm}{notion}

\newrobustcmd{\Aageq}[1]{\kl[\Aageq]\Aa{\geq #1}}
\knowledge{\Aageq}{notion}

\newcommand{\bool}{\mathsf{Bool}}
\newcommand{\boolP}{\mathsf{Bool}^+}
\newcommand{\init}{\mathsf{init}}

\newrobustcmd{\wordGame}[2]{\kl[\wordGame]{\mathcal{G}(#1,#2)}}
\newrobustcmd{\wordGameA}[1]{\kl[\wordGameA]{\mathcal{G}(#1,\Aa)}}
\knowledge{\wordGame}[\wordGameA]{notion}

\newrobustcmd\parity{\kl[\parity]{\mathsf{Parity}}}
\knowledge{\parity}{notion}

\newrobustcmd\sem[1]{\kl[\sem]{{[\![ #1 ]\!]}}}
\knowledge{\sem}{notion}

\newrobustcmd{\wmu}{\kl[\wmu]{\wh{\mu}}}
\knowledge{\wmu}{notion}

\knowledge{\ancestor}[\ancestorA|\ancestorAA|\descendantA|\descendantAA]{notion}

\newrobustcmd{\tree}{\kl[\tree]{\mathsf{Tree}}}
\knowledge{\tree}{notion}

\newrobustcmd{\classx}[1]{\kl[\classx]{[#1]_x}} % states with a common ancestor at level x
\newrobustcmd{\classxx}[2]{\kl[\classxx]{[#1]_{#2}}}
\knowledge{\classx}[\classxx|\scl]{notion}

\newrobustcmd{\eqx}{\mathrel{\kl[\eqx]{\sim_x}}} %% states with a common ancestor at level x
\newrobustcmd{\eqxx}[1]{\mathrel{\kl[\eqxx]{\sim_{#1}}}}
\knowledge{\eqx}[\eqxx]{notion}

\newrobustcmd{\transxy}{\kl[\transxy]{\d_x^y}} % transition at level x and projection at level y
\newrobustcmd{\transxxy}[1]{\kl[\transxxy]{\d_{#1}^y}}
\newrobustcmd{\transxyy}[1]{\kl[\transxyy]{\d_{x}^{#1}}}
\newrobustcmd{\transxxyy}[2]{\kl[\transxxyy]{\d_{#1}^{#2}}}
\knowledge{\transxy}[\transxxy|\transxyy|\transxxyy]{notion}

\newrobustcmd{\SA}{\kl[\SA]{\mathit{SA}}} %% strong acceptance
\knowledge{\SA}{notion}

\newrobustcmd{\SR}{\kl[\SA]{\mathit{SR}}} %% strong rejection
\knowledge{\SR}{notion}

\newrobustcmd{\Lsafex}{\kl[\Lsafex]{L_{x}}} % safe language of a tuple
\newrobustcmd{\Lsafexx}[1]{\kl[\Lsafexx]{L_{#1}}}
\newrobustcmd{\LsafexxA}[2]{\kl[\LsafexxA]{L^{#2}_{#1}}}
\knowledge{\Lsafex}[\Lsafexx|\LsafexxA]{notion}

\newrobustcmd{\fleq}{\mathrel{\kl[\fleq]\preccurlyeq}}
\newrobustcmd{\fleqx}{\mathrel{\kl[\fleqx]{\preccurlyeq_x}}}
\newrobustcmd{\fleqxx}[1]{\mathrel{\kl[\fleqxx]{\preccurlyeq_{#1}}}}

\knowledge{\fleq}[\fleqx|fleqxx]{notion}

\newrobustcmd{\dual}[1]{\kl[\dual]{\bar{#1}}}
\knowledge{\dual}{notion}

\newrobustcmd{\graph}{\kl[\graph]G}
\knowledge{\graph}{notion}
\newrobustcmd{\graphgeq}[1]{\kl[\graphgeq]G_{\geq #1}}
\knowledge{\graphgeq}{notion}

\newcommand{\littletaller}{\mathchoice{\vphantom{\big|}}{}{}{}}

\newcommand\restr[2]{{% we make the whole thing an ordinary symbol
		\left.\kern-\nulldelimiterspace % automatically resize the bar with \right
		#1 % the function
		\littletaller % pretend it's a little taller at normal size
		\right|_{#2} % this is the delimiter
}}

\newcommand\downpriority[1]{{% we make the whole thing an ordinary symbol
		\left.\kern-\nulldelimiterspace % automatically resize the bar with \right
		\littletaller % pretend it's a little taller at normal size
		\right\downarrow_{#1} % this is the delimiter
}}

\usetikzlibrary{calc,decorations.pathmorphing,shapes} %%
\newcounter{sarrow}
\newcommand\run[1]{%
	\stepcounter{sarrow}%
	\mathrel{\begin{tikzpicture}[baseline= {( $ (current bounding box.south) + (0,-0.5ex) $ )}]
			\node[inner sep=.5ex] (\thesarrow) {$\scriptstyle #1$};
			\path[draw,<-,decorate,
			decoration={zigzag,amplitude=0.7pt,segment length=1.2mm,pre=lineto,pre length=4pt}] 
			(\thesarrow.south east) -- (\thesarrow.south west);
	\end{tikzpicture}}%
}
\renewcommand{\run}[1]{\xrightarrow{#1}}
\newcommand{\lrun}[1]{\run{\phantom{.}#1\phantom{.}}}

\newrobustcmd{\tc}{\mathrel{\kl[\tc]{\equiv_L}}}
\knowledge{\tc}{notion}% tuple congruence
\newrobustcmd{\tccls}[1]{\kl[\tccls]{[#1]_{\tc}}} 
\knowledge{\tccls}{notion}% tuple congruence class
\newrobustcmd{\Aatcd}{\kl[\Aatcd]{\Aa_{d,\tc}}}% automaton for tuple congruence with parameter d
\newrobustcmd{\Aatc}{\kl[\Aatc]{\Aa_{\tc}}}% automaton for tuple congruence without parameter
\knowledge{\Aatc}[\Aatcd]{notion}
\newcommand{\Ttc}{\T^{\tc}} % transition systems in congruence automaton
\newcommand{\Qtc}{Q^{\tc}} % state sets of transition systems in congruence automaton
\newcommand{\deltc}{\delta^{\tc}} % transition function of transition systems in congruence automaton
\newcommand{\phitc}{\mu^{\tc}} % morphisms in congruence automaton
\newcommand{\qintc}{q_0^{\tc}} % initial state of congruence automaton
\newrobustcmd{\Ls}{\kl[\Ls]{L^s}}
\knowledge{\Ls}{notion} % safe language of a tuple
\newcommand{\bu}{\bar{u}} 
\newcommand{\bv}{\bar{v}}

\newcommand{\Aar}{\safeminA}
\newcommand{\Aal}{\lowA}
\newrobustcmd{\llayer}{\kl[\llayer]{\mathit{layer}}}
\knowledge{\llayer}{notion}

\newrobustcmd{\Qgeq}[1]{\kl[\Qgeq]{Q_{\geq #1}}}
\newrobustcmd{\Qgeqx}{\kl[\Qgeqx]{Q_{\geq x}}}
\newrobustcmd{\QgeqOne}{\kl[\QgeqOne]{Q_{\geq 1}}}
\knowledge{\Qgeq}[\Qgeqx|\QgeqOne]{notion}

\newrobustcmd{\actx}[1]{\mathrel{\kl[\actx]{\act{{#1}}_x}}}
\newrobustcmd{\actxx}[2]{\mathrel{\kl[\actxx]{\act{{#1}}_{#2}}}}
\newcommand{\runx}[1]{\run{#1}_x}
\newcommand{\runxx}[2]{\run{#1}_{#2}}
\knowledge{\actx}[\actxx]{notion}
\newcommand{\scl}[2]{\kl[\scl]{[#1]_{\seq_{#2}}}}

\newrobustcmd{\safemin}[1]{\kl[\safemin]{#1^{\mathrm{smin}}}}
\newrobustcmd{\safeminA}{\kl[\safeminA]{\Aa^{\mathrm{smin}}}}
\knowledge{\safemin}[\safeminA]{notion}

\newrobustcmd{\sep}[1]{\kl[\sep]{#1^{\mathrm{sep}}}}
\newrobustcmd{\sepA}{\kl[\sepA]{\Aa^{\mathit{sep}}}}
\knowledge{\sep}[\sepA]{notion}

\newrobustcmd{\low}[1]{\kl[\low]{#1^{\mathrm{low}}}}
\newrobustcmd{\lowA}{\kl[\lowA]{\Aa^{\mathit{low}}}}
\newcommand{\Tlow}{\Tt^{\mathit{low}}}
\newcommand{\mlow}{\mu^{\mathit{low}}}

\knowledge{\low}[\lowA]{notion}

\newrobustcmd{\central}[1]{\kl[\central]{#1^{\mathrm{ctl}}}}
\newrobustcmd{\centralA}{\kl[\centralA]{\Aa^{\mathit{ctl}}}}
\knowledge{\central}[\centralA]{notion}

\mathchardef\hyphen=45 %Decimal
\usepackage[super]{nth}

\newrobustcmd\inv[1]{#1^{-1}}

\newcommand{\tand}{\text{ and }}
\newcommand{\tor}{\text{ or }}
\newcommand{\tfor}{\text{ for }}
\newcommand{\tin}{\text{ in }}

\newcommand{\Nat}{\ensuremath{\mathbb{N}}}

\newcommand{\es}{\emptyset}

\newcommand{\incl}{\subseteq}
\newcommand{\sq}[1]{[#1]}
\newcommand{\di}[1]{\langle #1 \rangle}
\newcommand{\sqq}[1]{\sq{\cdot }}
\newcommand{\ddi}[1]{\di{\cdot }}
\newcommand{\set}[1]{\{#1\}}

\renewcommand\bar[1]{\overline{#1}}

\newcommand{\act}[1]{\mathrel{\stackrel{#1}{\longrightarrow}}}

\renewcommand{\d}{\delta}
\newcommand{\e}{\varepsilon}
\newcommand{\f}{\varphi}
\newcommand{\g}{\gamma}

\newcommand{\m}{\mu}

\newcommand{\s}{\sigma}

\renewcommand{\t}{\tau}
\newcommand{\w}{\omega}

\newcommand{\D}{\Delta}

\renewcommand{\S}{\Sigma}

\newcommand{\Aa}{\mathcal{A}}
\newcommand{\A}{\mathcal{A}}
\newcommand{\Bb}{\mathcal{B}}

\newcommand{\Gg}{\mathcal{G}}

\newcommand{\Tt}{\mathcal{T}}
\newcommand{\T}{\mathcal{T}}

\newcommand{\MSO}{\mathrm{MSO}}

\newcommand{\LOGSPACE}{\text{{\sc Logspace}}}

\def\sqr#1#2{\vbox%{\vcenter{\vbox%
 {\hrule height#2
  \mbox{\vrule width#2 height#1 \kern#1 \vrule width#2}%
  \hrule height#2}}%}}

\newcommand{\wh}[1]{\widehat{#1}}

\title{Layered automata: A canonical model for automata over infinite words}

\author{Antonio Casares}{University Kaiserslautern-Landau, Germany \and \url{https://antonio-casares.github.io/}}{antonio.casares@rptu.de}{https://orcid.org/0000-0002-6539-2020}{Partially supported by Deutsche Forschungsgemeinschaft (grant number 522843867) and European Research Council (grant number 101089343).
Part of this work was done while Casares was at the University of Warsaw, Poland, supported by the Polish National Science Centre (NCN) grant ``Polynomial finite state computation'' (2022/46/A/ST6/00072).}
\author{Christof L\"oding}{RWTH Aachen University, Germany \url{}}{}{https://orcid.org/0000-0002-1529-2806}{}
\author{Igor Walukiewicz}{CNRS, Bordeaux University, France\url{}}{}{https://orcid.org/000-0001-8952-7201}{}

\authorrunning{A. Casares and C. L\"oding and I. Walukiewicz} %TODO mandatory. First: Use abbreviated first/middle names. Second (only in severe cases): Use first author plus 'et al.'

\Copyright{Antonio Casares and Christof L\"oding and Igor Walukiewicz} %TODO mandatory, please use full first names. LIPIcs license is "CC-BY";  http://creativecommons.org/licenses/by/3.0/


 %TODO mandatory; please add comma-separated list of keywords

\category{} %optional, e.g. invited paper

\relatedversion{} %optional, e.g. full version hosted on arXiv, HAL, or other respository/website
\acknowledgements{We thank Thomas Colcombet, Aditya Prakash and Pierre Ohlmann for interesting discussions on the subject.}

\nolinenumbers %uncomment to disable line numbering

\begin{document}

\maketitle

\maketitle 
\begin{abstract}
  We introduce layered automata, a subclass of alternating parity automata that
  generalises deterministic automata. %By imposing a consistency property, 
  Assuming a consistency property, these
  automata are history deterministic and 0-1 probabilistic. We show that every
  omega-regular language is recognised by a unique minimal consistent layered
  automaton, and that this canonical form can be computed in polynomial
  time from every layered or deterministic automaton. 
  We further establish that for layered automata both consistency checking and inclusion testing can
  be performed in polynomial time. Much like deterministic finite automata,
  minimal consistent layered automata admit a characterisation based on
  congruences.

% We introduce layered automata, a formalism to describe a subclass of alternating parity automata, extending deterministic automata.
% We focus on layered automata that are semantically deterministic (SD). 
% We show that all SD layered automata are history deterministic.
% Each omega-regular language can be recognised by a unique minimal SD layered automaton, computable in polynomial time from any given SD layered automaton.
% We provide a congruence-based characterisation of this minimal automaton.
% %We conjecture that this automaton is also minimal amongst all alternating history deterministic automata.
\end{abstract}

\paragraph*{}
This document contains hyperlinks.
\AP Each occurrence of a "notion" is linked to its ""definition"".
On an electronic device, the reader can click on words or symbols (or just hover over them on some PDF readers) to see their definition.

\newpage

\tableofcontents

\newpage

\newpage
\section{Introduction}\label{sec:intro}
Automata over infinite words ($\omega$-automata in the following) were first introduced by B\"uchi to show the decidability of monadic second-order logic over $(\Nat,<)$~\cite{Buchi.Decision1962}. 
Since then, they have found numerous applications in the verification and synthesis of reactive systems~\cite{Cla.Eme.Sis.Automatic1986,Pnu.Ros.Synthesis1989,Kupferman.Handbook2018}.
%Languages of infinite words recognised by finite automata are called \emph{"$\omega$-regular"}.

A key limitation of current $\w$-automata theory is the lack of minimal canonical
models. Minimisation is instrumental for practical applications, particularly
when automata are derived from specifications, like LTL formulas, through
hierarchical translations. 
In this process, each constructor in the
specification corresponds to an operation on the automaton — some of which
necessitate a powerset construction. 
We know from the case of finite words and trees that minimisation after each operation
is key to efficient translations~\cite{KlarlundMS02}.
Minimisation is also important for theory.
One immediate example is learning algorithms for $\omega$-regular languages, but
more broadly, it indicates that we have gained insights into the structural
properties of acceptors. 
Such insights may eventually allow us to
revisit long-standing open problems, such as the index problem for tree automata \cite{ColcombetL08,IdirL25}. 

Unlike to the case of automata over finite words, deterministic
$\omega$-automata do not have minimal canonical models, as a language may admit several
non-isomorphic deterministic automata with a minimal number of states. Moreover,
the minimisation problem is NP-hard for all kinds of deterministic
$\omega$-automata expressive enough to recognise all $\omega$-regular
languages~\cite{Schewe.MinimisingDet2010,Casares.Chromatic2022,Abu.Ehl.NPhard2025}.
While there are some canonical ways to represent $\omega$-regular languages,
such as $\w$-semigroups or congruences, these representations can be potentially exponentially
larger than "deterministic parity automata", and furthermore, they cannot be used directly in applications from verification and synthesis. 

"History deterministic" (HD) automata are a subclass of non-deterministic automata for which non-determinism can be resolved on the fly. They were originally introduced under the name of \emph{good for games}~\cite{Hen.Pit.GFG2006}\footnote{The related notion of \emph{good for trees automata} was previously introduced in~\cite{Kup.Saf.Var.Relating1996Relating}.}, as they are exactly those automata that can be composed with any game while preserving the winner. This property allows their use in applications in verification and synthesis instead of deterministic automata. Every "deterministic automaton" is "history deterministic", but a "history deterministic" automaton can be exponentially smaller than the smallest deterministic one~\cite{Kup.Skr.Determinisation2015}. "History determinism" has been generalised to "alternating $\omega$-automata"~\cite{Colcombet.HDR2013,Bok.Leh.Alternating19}. History deterministic alternating automata have the same good compositionality properties, making them suitable for applications in verification and synthesis, while being even smaller than non-deterministic HD automata~\cite[Lemma~6]{BKLS.AlternatingGFG2020}.

Other types of automata that have been actively studied recently are those in which the non-determinism can be resolved randomly~\cite{AvniK15,Hen.Pra.The.Resolving2025,PPSTTY.Resolving2025}, or automata that compose faithfully with MDPs~\cite{HPSS0W.GoodForMdp2020}. In particular, \emph{"0-1 probabilistic automata"} are automata in which a random walk over an infinite word is accepting with probability $1$ if and only if the word belongs to the semantics of the automaton; otherwise, it is rejecting with probability $1$.
Such automata are in particular suitable for the study of MDPs~\cite{Hen.Pra.The.Resolving2025}.
\smallskip

In this work, we propose a new model of acceptors for "$\omega$-regular languages" that we call "layered automata".
Our starting point for this model  is the result of Abu Radi and
Kupferman~\cite{Abu.Kup.Minimization2022} stating that (transition-based) "history deterministic" "coB\"uchi automata" can be minimised in polynomial time and admit a canonical minimal representation. In \cite{Lod.Wal.Congruences2025} it has furthermore been shown that this canonical automaton can be described in a natural way using a congruence over pairs of finite words.
However, "HD" "coB\"uchi automata" can only recognise a fragment of "$\omega$-regular languages": those that are in the $\Sigma_2^0$ level in the Borel hierarchy, or equivalently, persistent properties in the Manna-Pnueli hierarchy~\cite{Man.Pnu.Hierarchy1989}.

The model of "layered automata" that we propose in this work can be seen as an extension of the minimisation result of Abu Radi and
Kupferman to all "$\omega$-regular languages". 
This is in line with some recent works such as \emph{chains of coB\"uchi
automata}~\cite{Ehl.Sch.Natural2022} and \emph{rerailing
automata}~\cite{Ehlers.Rerailing2025} to which we compare in Section~\ref{subsec:comparison}. 
Orthogonally, "layered automata" can also be seen as a generalisation of the Zielonka tree of a Muller language~\cite{Zielonka.Infinite1998} and of signature automata for positional languages~\cite{Cas.Ohl.Positional2024}.
We provide ample evidence that layered automata have good structural and
algorithmic properties.

\subsection*{Contributions}
We introduce \emph{"layered automata"}, a formalism to represent all
$\omega$-regular languages. 
One can think of a "layered automaton" as a representation of an "alternating
parity automaton" with some good structural properties. 
If a layered automaton satisfies an additional property of \emph{"consistency"},
then the associated alternating automaton is "history deterministic".

"Consistent" "layered automata" combine the structure of both minimal "history deterministic" "coB\"uchi automata" (Proposition~\ref{prop:minimalCoBuchi-are-layered}) and Zielonka trees. %~\cite{Zielonka.Infinite1998,CCFL.ACD2024}.
They also encompass "deterministic" parity automata (Proposition~\ref{prop:deterministic-are-layered}).
They enjoy remarkably good properties. Most notably, they admit a canonical minimal form that is computable in polynomial time.

This canonicity result can be stated as follows:

\begin{theorem}
    Every "$\omega$-regular language" $L$ can be recognised by a "consistent" "layered automaton" $\Aa_L$ such that every equivalent "consistent" "layered automaton" contains a "subautomaton" admitting a surjective "morphism@@lay" to $\Aa_L$.
    Moreover, $\Aa_L$ can be computed in polynomial time from any given "consistent" "layered automaton" recognising $L$.
\end{theorem}

We detail further properties of "layered automata".

\subparagraph{Applicability.} 
"Consistent" "layered automata" are representations of alternating parity
automata that are guaranteed to be both "history deterministic" and "0-1 probabilistic" (Propositions~\ref{prop:layered-are-HD} and~\ref{prop:layered-are-probabilistic}).
Therefore, "consistent" "layered automata" have the key properties that make them suitable for applications in verification and  synthesis.
Moreover, they provide an easy-to-check sufficient condition for "HDness" and being "0-1 probabilistic", as "layered automata" are defined syntactically and checking "consistency" can be done in polynomial time (Proposition~\ref{prop:consistent-PTIME}).

As "layered automata" generalise minimal coB\"uchi automata, they can be exponentially smaller than minimal "deterministic parity automata"~\cite{Kup.Skr.Determinisation2015}.

We also show that emptiness and inclusion of "consistent" "layered automata" can be done in polynomial time (Propositions~\ref{prop:emptiness-PTIME} and~\ref{prop:inclusion-PTIME}).

\subparagraph{Canonicity, morphisms, and congruences.} We show that every "$\omega$-regular language" can be recognised by a unique minimal "consistent" "layered automaton".
%To state this canonicity result, we develop a theory of "morphisms of layered automata".
More precisely, we show that for every "$\omega$-regular language" $L$ there is a "consistent" "layered automaton" $\Aa_L$ that is "normal", "centralised" and "safe minimal" and any "layered automaton" with these three properties is "isomorphic" to $\Aa_L$ (Theorem~\ref{thm:concise-are-unique}).
Moreover, every "consistent" "layered automaton" "recognising@@lay" $L$ contains a "subautomaton" that admits a surjective "morphism@@lay" to $\Aa_L$ (Corollary~\ref{cor:canonicity}). Furthermore, we provide a characterisation in terms of congruences for the minimal "layered automaton" $\Aa_L$ (Theorem~\ref{thm:Aatc-correct}).
This characterisation is purely based on the language $L$, and does not depend on a given representation as an automaton. 
To the best of our knowledge, these are the first results of this type for an
automaton model capable of recognising all  $\omega$-regular languages and having the
minimal form never bigger than "deterministic parity automata".

\subparagraph{Minimisation in PTIME.} 
We show that given a "consistent" "layered automaton" $\Aa$, we can construct in polynomial time the minimal "layered automaton" for $L(\Aa)$ (Theorem~\ref{thm:minimisation}). Since every "deterministic parity automaton" is trivially a "consistent" "layered automaton", this provides an efficient minimisation procedure for "deterministic parity automata" into "layered automata". This does not contradict the above-mentioned lower bounds on minimisation, as the result of the procedure is an alternating automaton.

\subparagraph{Relation to other models.}
We provide examples and results comparing and relating "layered automata" to
other models, such as chains of coB\"uchi automata and Zielonka trees. We also
discuss how "layered automata" can be used to characterise important properties
of $\omega$-regular languages, such as positionality,
following~\cite{Cas.Ohl.Positional2024}. While we leave a more detailed study of
all these relations as future work, our discussion shows that "layered automata"
have many interesting connections to various other well-established and more
recent formalisms from the theory of $\omega$-regular languages. 

\subsection*{Related work}
The starting point of our work is the minimisation procedure for "HD" "coB\"uchi automata" from Abu
Radi and Kupferman~\cite{Abu.Kup.Minimization2022}. By duality, this entails a minimisation procedure for universal "HD" B\"uchi automata, and the notion of "layered automata" can be seen as 
a natural extension to arbitrary parity indices that mixes nondeterministic and universal branching.
"Layered automata" are  also inspired by
signature and $\e$-complete automata
from~\cite[Sect~3.3]{Cas.Ohl.Positional2024}, as well as the minimisation
procedure for them. 
The proof that all "layered automata" are "history deterministic" generalises ideas from~\cite[Lemma~3]{Kup.Skr.Determinisation2015}.
The congruence-based characterisation is an extension of the congruences for "coB\"uchi automata" from~\cite[Sect~3]{Lod.Wal.Congruences2025}, using conditions closely related to the priority mapping on finite words defined in~\cite[Def~3.1,3.5]{Boh.Lod.Constructing2024}.

Various representations of "$\omega$-regular languages"  admitting--in some sense--canonical models exist in the literature. 
%Traditionally, these representations had the drawback of being exponentially larger than a minimal "deterministic parity automaton".
These include some algebraic approaches based on congruences, such as Arnold's
congruence~\cite{Arnold.Congruence1885}, Wilke-algebras and
$\omega$-semigroups~\cite{Wilke.Eilenberg1991,Per.Pin.OmegaSemigroup1995},
families of right congruences~\cite[Def.~5]{Mal.Sta.Syntactic1997} and families
of DFAs~\cite{Klarlund94,Ang.Fis.Learning2016}. These objects have the
disadvantage that they can be exponentially larger than deterministic parity
automata (it is folklore that already for regular languages of finite words, the
syntactic monoid can be of size $n^n$ for the minimal DFA being of size $n$, see
\cite[Thm.~3]{HolzerK02} for a concrete example; this exponential gap transfers
to all the aforementioned formalisms).
While their size is incomparable to the size of a smallest "DPA" for a language, 
%in comparison to deterministic automata). 
%On the other hand, these algebraic representations 
they cannot directly be used in applications like synthesis.

Another representation, taking the form of a specific "deterministic parity automaton" (DPA), is proposed in~\cite[Sect.~3]{Boh.Lod.Constructing2024}.
There, the ``precise DPA'' of a language is defined as a product of deterministic Mealy
machines that are extracted directly from the language in terms of a coloring
function. However, the precise DPA %is a specific DPA for the language that 
can be exponentially larger than a smallest "DPA".
%\acmargin{I would include the precise DPA in the previous paragraph}\chin{I separated it from the previous part because there are some essential differences. The precise DPA is an automaton, while the others are not, so it is different. For the others, the size w.r.t.\ DPAs is incomparable, while the precise DPA is always at least as big and can be exponentially bigger. The part that says that the previous formalisms require exponential constructions for transformation into DPA has been removed. I think that this is relevant in combination with the comment that they cannot be used directly in synthesis.\\ So I agree that the precise DPA belongs to the previous paragraph if we only make the distinction ``can be exponential in DPA'' and the opposite. But I find this too coarse. Then we can rather omit the precise DPA.}

More recently, some canonical representations based on "HD" automata have been introduced.
In 2022, Ehlers and Schewe proposed to represent an "$\omega$-regular language" as a decreasing chain of coB\"uchi-recognisable languages~\cite{Ehl.Sch.Natural2022}. 
Each of these languages can be given by the minimal "HD" "coB\"uchi automaton" from Abu Radi and Kupferman, leading to a canonical representation. This formalism has subsequently been studied under the name \emph{COCOA}~\cite{Ehl.Kha.Fully2024,Ehl.Kha.NaturallyColored2024}.
Minimal "layered automata" are clearly related to COCOA, but there are some important differences. 
%As COCOA and rerailing automata, minimal "layered automata" are smaller than "deterministic parity automata" (potentially exponentially smaller).
A key difference is that "layered automata" bring a direct connection with "alternating" "history deterministic" automata--which are widely accepted as the class of $\omega$-automata suitable for verification and synthesis.
Notably, we conjecture that canonical "layered automata" are minimal amongst all "history deterministic" "alternating parity automata" (Conjecture~\ref{conj:minimality-HD}). The size comparison between both models is a bit subtle. 
While COCOA can be smaller than "layered automata", in order to obtain an "HD" automaton from a COCOA, one needs to perform a product of all its components \cite{Ehl.Kha.Fully2024}.
We give a more detailed comparison with COCOA in Section~\ref{subsec:comparison}.

In 2025, Ehlers introduced \emph{rerailing automata}~\cite[Def.~1]{Ehlers.Rerailing2025} and showed that they can be minimised efficiently.
These are automata where each run is associated with a "priority" (the minimal appearing infinitely often), and a word is accepted if the maximal priority of its runs is even. Moreover, they must satisfy a semantic ``rerailing property'', inspired by the properties of minimal "HD" "coB\"uchi automata".
In~\cite[Sect.~VI]{Ehlers.Rerailing2025}, rerailing automata are shown to be usable in reactive synthesis, but their exact relation with "HD" automata is not yet established.

\section{Preliminaries}\label{sec:preliminaries}
\AP We let $[d] = \{1,\dots,d\}$, and we call the elements of $[d]$ ""priorities"". We use the variables $x,y$ to denote "priorities".
\AP The ""parity condition"" is the language over $[d]^\omega$ given by:
\[ \intro*\parity = \{ x_1x_2\dots \in [d]^\omega \mid \liminf x_i \text{ is even}\}. \;
\text{ (Note the use of min-parity in this paper.)}\]

\subparagraph{Transition systems.}
\AP A ""transition system"" over a set $\Sigma$ is an edge-labelled directed
graph $\T = (Q,\Delta)$, where $\Delta \subseteq Q\times \Sigma \times Q$.
If $(q,a,q')\in \D$, we also write $q\act{a}q'$.
\AP It is ""complete"" if for all $q\in Q$ and  $a\in \Sigma$, there is some $(q,a,q')\in \Delta$.
\AP It is ""deterministic@@TS"" if for every $q\in Q$ and $a\in \Sigma$ there is at
most one $q'$ such that $(q,a,q')\in \Delta$.
We prefer to represent the transitions of deterministic "transition systems"
by a function $\delta\colon Q\times \Sigma \to Q\cup \{\bot\}$, where
$\delta(q,a) = q'$ if $(q,a,q')\in \Delta$, and $\delta(q,a) = \bot$ if no such
transition exists.

\AP For a word $w\in \S^*$ we write $q\run{w}q'$ if there is a
path from $q$ to $q'$ whose label is $w$, which we call a ""run"".
We sometimes omit $q'$, writing just $q\run{w}$, if $q'$ is irrelevant.
We write $q\act{w}\bot$ it there is no path with label $w$
from $q$. 

If $\T$ is a "transition system" over a product alphabet $\S\times [d]$, by an abuse of notation we call a \emph{run over $w\in \Sigma^\omega$} any path having $w$ as labels in the $\S$-component.
\AP In this case, we also say that a run ""is accepting@@run""
if its sequence of $[d]$-labels satisfies the "parity condition".
We write $q\act{a:x}q'$ if $(q,(a,x),q')\in \T$, and $q\act{a: \geq x}q'$ if $(q,(a,y),q')\in \T$ for some $y\geq x$, similarly for $> x$, etc.
We write $q\run{w: \geq x} q'$ if the minimal priority over the path is $\geq x$, similarly for $> x$, etc.

\AP A ""morphism"" between two "transition systems" $\T = (Q,\D)$,
$\T'=(Q',\D')$ over the same set~$\Sigma$ is a function $\mu\colon Q \to Q'$
such that if $(p,a,q)\in \D$ then $(\mu(p),a,\mu(q))\in \D'$.
Note that any "run" over $w$ in $\T$ from $q$ maps to a "run" over $w$ in $\T'$ from $\mu(q)$.

\subparagraph{Deterministic parity automata.} \AP A ""deterministic parity automaton"" is a "complete" "deterministic" "transition system" $\Aa = (Q,\D)$ over the
set $\Sigma\times [d]$ with a designated initial state $q_\init\in Q$.

A word $w\in \S^\omega$ is accepted by the automaton if the unique run over $w$ satisfies the "parity condition".
The language of the automaton is the set of words it accepts.
\AP A language is ""$\omega$-regular"" if it can be recognised by a "deterministic parity automaton". 
The class of "$\omega$-regular" languages admits many equivalent definitions,
for instance, through definability in $\MSO$ logic, or recognition by
non-deterministic B\"uchi automata, see for instance~\cite{Gra.Tho.Wil.Book2002}.

\subparagraph{Alternating parity automata.}
\AP In full generality, an ""alternating parity automaton"" is given by a set of states together with a transition function of the form $ \delta \colon Q\times \Sigma \to \boolP(Q\times[d])$. 
In this work, all alternating automata are of a simpler form, which we call \emph{"simple by priorities" alternating automata}. These correspond to automata where all transitions are either
\[\delta(q,a) = \bigwedge_{p\in Q'} (p,x) \; \text{ for some even } x, \;\; \tor \;\; \delta(q,a) = \bigvee_{p\in Q'} (p,x) \; \text{ for some odd } x.\]
We refer to~\cite{Bok.Leh.Alternating19} for definitions in the general case.
We give next a formal, alternative presentation of the subclass of "simple by priorities" alternating automata.

\AP A ""simple by priorities"" automaton is a "complete" "transition system" $\Aa = (Q,\D)$
over the set $\Sigma\times [d]$ with a designated ""initial state"" $q_\init\in
Q$,
satisfying that for every $q\in Q$ and $a\in \Sigma$ 
all $a$-transitions from $q$ have the same "priority".
\AP So there is $x$ a ""priority of action $a$"" in $q$, meaning that for all
$(q,(a,y),q')\in \Delta$ we have $y=x$.

% We write $q\act{a:x}q'$ if $(q,(a,x),q')\in \D$ and $q\act{a}q'$ if such transition exists for some $x$.
% We also write $q\act{a: \geq x}q'$ if $(q,(a,y),q')\in \T$ for some $y\geq x$, similarly for $> x$, etc.

The semantics of a "simple by priorities" automaton is defined via game where
the "priority of an action" determines who makes the move. Given
an input word $w=a_1a_2\dots\in \Sigma^\omega$, we define the game $\Gg(\Aa,w)$
where two players, called Eve and Adam, build a "run" over $w$ in $\Aa$ as follows:
\begin{itemize}
    \item The initial position is the state $q_1 = q_\init$.
    \item Suppose the current run constructed in a play ends in a state $q_i$,
    and $x_i$ is the priority of $a_i$ in $q_i$.   
    If $x_i$ is odd then Eve picks a transition $q_i \act{a_i:x_i} q_{i+1}$,
    otherwise Adam makes the choice.
    \item Eve wins if the constructed path is "accepting".
\end{itemize}

\AP We say that $w$ is ""accepted"" by $\Aa$ if Eve has a winning strategy in the game $\Gg(\Aa,w)$.
Note that, by determinacy of parity games~\cite{Martin.Determinacy1975}, on the contrary Adam has a winning strategy.
\AP The ""language recognised@@alt"" by $\Aa$, written $\intro*\lang(\Aa)$, is the set of words it "accepts".
\AP The ""language of a state"" $q\in Q$ is the language recognised by the
automaton obtained by fixing $q$ as the "initial state". We denote it
$L(\Aa,q)$, or just $\lang(q)$ if $\Aa$ is clear from the context.

\subparagraph{Semantic determinism.}
\AP An automaton $\Aa$ is ""semantically deterministic@@alt"" if whenever we have
$p\act{a}q_1$ and $p\act{a}q_2$ then $L(\Aa,q_1)=L(\Aa,q_2)$. Equivalently, if
$p\act{a}q$ then $L(\Aa,q)=a^{-1}L(\Aa,p)$.
%, where $ a^{-1}L(p) = \{w\in \Sigma^\omega \mid aw\in L(p)\}$. 

\begin{remark}
    Let $\Aa$ be a "semantically deterministic" automaton "recognising@@alt" a language~$L$. Let $u$ be a word labelling a path from the initial state to $q$. Then, the "language@@state" of $q$ is:
    \[ u^{-1}L = \{w\in \Sigma^\omega \mid uw\in L\}.\]
\end{remark}

\subparagraph{Resolvers and history determinism.}
Consider a "simple by priorities automaton" $\Aa$.
\AP A ""resolver"" for Eve in $\Aa$ is a function $\sigma\colon (\set{q_\init}\cup\Delta^+) \times
\Sigma \to Q$, satisfying that whenever a run $\rho\in (\set{q_\init}\cup\Delta^+)$ ends in a state
$q$, and the "priority of" $a$ in $q$ is odd then $q'=\sigma(\rho,a)$ is an
$a$-successor of $q$, namely $q\act{a}q'$. 
\AP A "run" is ""consistent with"" $\sigma$ (or a ""$\sigma$-run"")
if for every prefix $\rho\cdot (q,(a,x),q')$ with $x$ odd, we have
$\sigma(\rho,a)=q'$.
The ""language accepted by an Eve-resolver"" $\sigma$ is
\[ L(\Aa,\sigma) = \{w\in \Sigma^\omega \mid \text{ every "$\sigma$-run" over $w$ is accepting}\}. \]
Note that we always have $L(\Aa,\sigma) \subseteq L(\Aa)$.
\AP We say that the "Eve-resolver" $\sigma$ is ""valid"" if $L(\Aa,\sigma) = L(\Aa)$.
We define symmetrically ""resolvers for Adam"" and their languages (note that if $\tau$ is an "Adam-resolver", then $L(\Aa,\tau) \supseteq L(\Aa)$).
For two fixed resolvers $\sigma$ and $\tau$ for Eve and Adam, respectively, every word admits a unique run that is "consistent with" both $\sigma$ and $\tau$. We write $L(\Aa,\sigma,\tau)$ for the set of words for which this run is accepting.

\AP An "alternating automaton" is ""history deterministic"" (HD) if both Eve and Adam have "valid" "resolvers". That is, if there are "resolvers" $\sigma$ and $\tau$ for Eve and Adam such that
\[L(\Aa)=L(\Aa,\sigma)=L(\Aa,\tau)=L(\Aa,\sigma,\tau).\]

% More formally, we say that $f:S\to \RR_{\geq 0}$ is a ""fair probability distribution"" if $\sum_{s\in S}f(s)\leq 1$ and $f(s)>0$ for all $s\in S$. We let $\Dd_{\mathrm{fair}}(S)$ be the set of "fair probability distributions" over $S$
% Let $\bar{\delta}\colon Q\times \Sigma \to \Dd_{\mathrm{fair}}(\delta(q,a))$, where $\delta(q,a) = \{(q,a,x,q')\in \D\}$.

\section{Layered automata}\label{sec:layered}
In this section we introduce "layered automata", the central notion of this paper. 
We develop their definition in two stages: the syntactic structure of "layered
automata" (Section~\ref{subsec:syntactic}), and the language they recognise via an associated "alternating
parity automaton" (Section~\ref{subsec:semantics-alternating}).
Thanks to the structure of a "layered automaton" the associated automaton is
guaranteed to be "simple by priorities".

Next, we introduce "uniform semantic determinism", a slightly stronger notion
than "semantic determinism".
"Semantic determinism" is a necessary condition for an
automaton to be HD.
For "layered automata", "uniform semantic determinism" yields two significant
consequences: the associated alternating automaton is both "history
deterministic" and "0-1 probabilistic" (Propositions~\ref{prop:layered-are-HD}
and~\ref{prop:layered-are-probabilistic}).

Our main technical tool, presented in Section~\ref{sec:layered acceptance}, is a
characterisation of acceptance by "layered automata" that does not rely on
alternation.
For this we introduce a notion of "consistency", which characterises "uniform semantic determinism". For "layered automata", we can test "consistency" in polynomial time (Proposition~\ref{prop:consistent-PTIME}).
Furthermore, for "consistent" "layered automata", emptiness and language inclusion are decidable in polynomial time (Propositions~\ref{prop:emptiness-PTIME} and~\ref{prop:inclusion-PTIME}).

\subsection{Syntactic definition of layered automata}\label{subsec:syntactic}

 \AP A ""layered automaton"" $\Aa$ is a tuple of "deterministic" transition systems
$\T_1, \dots, \T_d$ together with "morphisms" $\mu_x:
\T_{x+1} \to \T_{x}$, for $1\leq x <d$. 
Additionally, we require that $\T_1$ is "complete", has a distinguished
"initial state" $q_\init$ in $\T_1$, and all states of $\T_1$ are reachable from
$q_\init$.\footnote{The choice of starting at index $1$
(instead of for instance $0$) is arbitrary. It has been chosen so that layered
automata with two layers correspond to the well understood class of  HD "coB\"uchi automata".}
%These transition systems may contain isolated vertices.

\AP We say that $\T_x$ is the ""$x$-th layer"" of the automaton.
We use $Q_x$ to denote the set of states of $\T_x$, 
which we refer to as ""$x$-states"".
We assume $Q_1,\dots,Q_d$ pairwise disjoint. 
We let $\d_x:Q_x\times\S \to Q_x\cup\{\bot\}$ be 
the (partial) transition function of $\T_x$ (which is required to be total for $\T_1$).
\AP We write $p\intro*\actx{a} q$ to emphasize that a transition is taken in the "$x$-th layer" of the automaton.
Sometimes we write $Q_{x+1}$ for $x$ some layer; we assume that this is the empty set when $x=d$.

\AP We write $\intro*\Qgeqx$ for the union of $Q_x,\dots,Q_d$, 
so $\QgeqOne$ is the union of all the sets of states.

\AP For every $x$, we define $\wmu_x\colon Q_{\geq x} \to Q_x$ by
\[\intro*\wmu_x(q)=\mu_x(\mu_{x+1}(\dots
\mu_{y-1}(q)\dots)), \, \tfor q\in Q_y, \; y>x \;\ ; \;\; \tand \;\; \wmu_x(q)= q,  \,\tfor q\in Q_x.\]

\AP For a state $q$, let $\intro*\llayer(q)$ be the layer $x$ such that $q\in
Q_x$.
For a pair $(q,a)$ of a state and a letter, let $\llayer(q,a)$ be the biggest
layer $x$ such that $\d_x(\wmu_x(q),a)$ is defined.
Observe that $\llayer(q,a)$ is always defined as $\T_1$ is complete.

\AP We say that $p$ is an ""ancestor"" of $q$ if $\wmu_x(q) = p$, for $x = \llayer(p)$.
\AP Given $q,p\in Q_{\geq x}$, we write $q\intro*\eqx p$ if $\wmu_x(q) = \wmu_x(p)$ and we let $ \intro*\classx{q} = \wmu_x^{-1}(\wmu_x(q))$.
\AP A state $q\in Q_{\geq 1}$ is a ""leaf state"" if it is not an image of any $\mu_x$.
%\AP We sometimes call ""internal nodes"" arbitrary nodes in $Q_{\geq 1}$.

\subparagraph{The forest of a layered automaton.}
\AP We can view a "layered automaton" as a ""forest@@representation"" with nodes
$Q_{\geq 1}$. Each state in $\T_1$ is the root of a tree, a node $q\in Q_{x+1}$
has $\mu_x(q)\in Q_{x}$ as the parent.
The levels of this "forest" are labelled from $1$ (top level) to $d$.

In this way, $\llayer(q)$ is the level of $q$ in this forest.
%Note that there may be branches of different lenghts, i.e. leaves are not necessarily at level $d$.
The function $\wmu_x(q)$ designates the "ancestor" at level $x$ of a node $q\in Q_{\geq x}$. We have that $q\eqx p$ if $p$ and $q$ have the same "ancestor" at level $x$.
Equivalently, $\classx{q}$ is the set of nodes of the subtree rooted at the "ancestor" of $q$ at level $x$.
As expected, "leaf states" are the states that are leaves in this forest.
Observe that more than one level can have leave states.

\begin{example}
  In Figure~\ref{fig:layered-example-trees}, we show two layered automata, represented as trees. 
  Each transition system constitutes a level of the tree.
  The morphisms $\mu_x$, corresponding to the parent relation, are represented as dotted arrows.
  
  In these two examples, all "leaf states" are on the same layer. This needs not to be the case. Figure~\ref{fig:congruence-aba} (Section~\ref{sec:congruences}, Example~\ref{ex:congruence-aba}) shows a "layered automaton" with "leaf states" at different layers.
\end{example}
\begin{figure}[h!]
  \centering
  \includegraphics[width=0.56\textwidth]{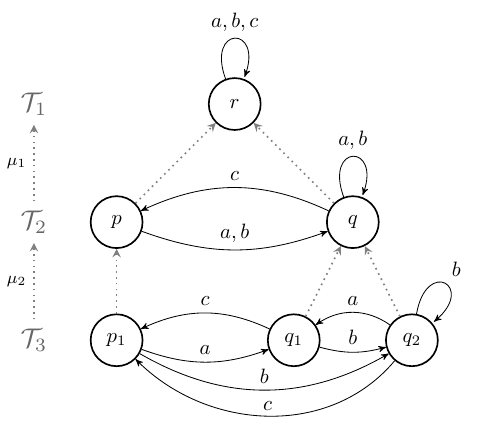}
  \hspace{2mm}
  \includegraphics[width=0.40\textwidth]{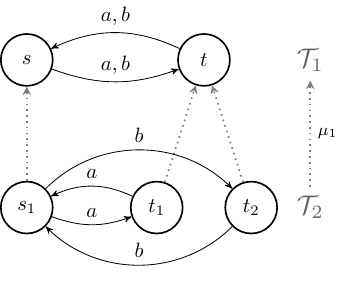}
  \caption{Two "layered automata", represented as trees.}\label{fig:layered-example-trees}
\end{figure}

\subsection{Semantics via alternating parity automata}\label{subsec:semantics-alternating}
We now define acceptance for "layered automata" by associating to each "layered
automaton" $\Aa$ an "alternating automaton" $\sem{\Aa}$, called the "semantics
of" $\Aa$; the language of $\Aa$ is defined as the language of $\sem{\Aa}$. The
tree structure of $\Aa$ constrains the transitions of $\sem{\Aa}$ to be "simple
by priorities". With this semantics in place, we introduce "uniform semantic
determinism", a slight strengthening of "semantic determinism". This
property ensures, that all "leaf states" with the same ancestor in
layer~$1$ accept the same language — a constraint we later characterise via "consistency" (Section~\ref{sec:layered acceptance}).

\subsubsection{Definition}
\AP We associate with a "layered automaton" $\Aa$ its ""semantics automaton"" $\intro*\sem{\Aa}$, which is a "simple by priorities" "alternating parity automaton" defined as follows.
The states $Q$ are the "leaf states" of $\Aa$.
The "initial state" is any "leaf state" $q$ with $\wmu_1(q)=q_\init$ the "initial state" in $\T_1$.\footnote{Formally, this defines a family of automata. 
 In general, these automata are not necessarily equivalent.
 However, for the automata considered in this paper (see notion of "uniformly SD" and "consistency" below), the choice of initial state is irrelevant.}

The transitions are
\[ \Delta = \left\{(q,a,x,q') \in Q\times \Sigma \times [d] \times Q  \;\; \mid \;\; 
\begin{array}{l}
  x=\llayer(q,a), \tand\\
   \wmu_x(q) \act{a} \wmu_x(q') \tin \T_x
\end{array} 
 \right\}.\]
Intuitively, when in a leaf $q$ and given a letter $a$ we consider
$\llayer(q,a)=x$, that is the biggest $x$ such that the transition $\wmu_x(q)\act{a}_x p$ is defined.
There is a transition from $q$ to every leaf $p'\in\wmu^{-1}_x(p)$.
All these transitions have priority $x$.
If $x$ is odd, Eve chooses a leaf under $p$, otherwise it is Adam.

\begin{remark}
  Whenever there is a transition $q\act{a:x}q'$ in $\sem{\Aa}$, the there is
  also a transition from $q$ to every "leaf state" $q''\in
  \classx{q'}$, and to no other states. This property justifies the notation
  $q\act{a:x}\classx{q'}$.
\end{remark}

% \AP We say that a path in $\sem{\Aa}$ is ""${x}$-safe@@path"" if it contains only "priorities" $\geq x$.
% We note that there is
% an exact correspondence between the paths in $\Tt_x$ and "${x}$-safe" paths in
% $\sem{\Aa}$.
% In particular, in the "semantics of a" "layered automaton", whether a priority $<x$ is seen from $p$ after reading $u$ is a deterministic property, independent of the strategies used by the players.

% \begin{remark}\label{rmk:x-safe-paths}
%   Let $u\in \S^*\cup \S^\omega$ and $p$ a state in $\sem{\Aa}$.
%   Then, there is a run $p\run{u : \geq x}$ that is "${x}$-safe" if and
%   only if there is a run $\wmu_x(p)\runx{u}$ in $\Tt_x$. 
%   This implies that if there is an "${x}$-safe" run from $p$ over $u$  then \emph{every} run $p\lrun{u}$ in $\sem{\Aa}$ is "${x}$-safe".
% \end{remark}

In the following, we focus almost exclusively on "layered automata" such
that $\sem{\Aa}$ is "semantically deterministic". 
We actually work with a slightly stronger property.
\begin{definition}
  \AP A layered automaton $\Aa$ is ""uniformly semantically deterministic"" if
  for every pair of "leaf states" $p\sim_1 q$ we have $L(\sem{\Aa},p) =
  L(\sem{\Aa},q)$.   
\end{definition}

As the name suggests we have the following.
\begin{lemma}\label{lem:SD-1-equivalent}
  If a "layered automaton" $\Aa$ is "uniformly semantically deterministic" then it is  "semantically deterministic".
\end{lemma}
\begin{proof}
  If $q\act{a} p_1$ and $q\act{a}p_2$ in $\sem{\Aa}$, then $\wmu_1(p_1)=\wmu_1(p_2)$, hence the result.
\end{proof}

In Proposition~\ref{prop:consistent-PTIME} we show that we can check 
in polynomial time if a given layered automaton is uniformly semantically deterministic.

\begin{remark}
    \AP If $\Aa$ is "uniformly semantically deterministic" then $\T_1$ admits a
    "morphism" into the ""automaton of residuals"" of $L = L(\sem{\Aa})$: the
    "transition system" with states the sets $u^{-1}L$ and transitions $u^{-1}L
    \act{a} (ua)^{-1}L$. This is thanks to our assumption that all states in
    $\T_1$ are reachable from $q_\init$.
\end{remark}

\subsubsection{Examples}\label{subsec:examples}

\begin{example}
  In Figure~\ref{fig:alternating-example-3letters}, we show the alternating parity automaton corresponding to the "layered automaton" on the left of Figure~\ref{fig:layered-example-trees}.
  It recognises the language
  \[ L = \text{Words with finitely many } cc \text{ and infinitely many } aa.\]

  The states are the leaves of the tree.
  Whenever a letter $c$ is seen from state $p_1$, priority $1$ is produced, and
  Eve can choose to go to any state.
  Whenever a letter $a$ is seen from state $q_1$, priority $2$ is produced, and
  Adam can choose to go to $q_1$ or $q_2$.
  The rest of transitions produce priority $3$, and correspond to the transitions at level $3$ in the "layered automaton".
\end{example}
\begin{figure}[h!]
  \centering
  \includegraphics[]{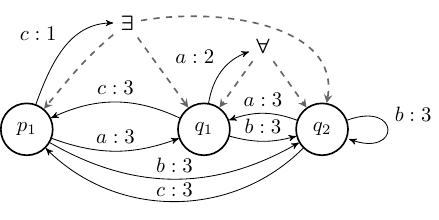}
  \caption{Alternating parity automaton for the "layered automaton" on the left of Figure~\ref{fig:layered-example-trees}.}\label{fig:alternating-example-3letters}
\end{figure}

\begin{example}
  In Figure~\ref{fig:alternating-example-mod2}, we show the alternating parity automaton corresponding to the "layered automaton" on the right of Figure~\ref{fig:layered-example-trees}.
  It recognises the language
  \[ L = (\Sigma^2)^*(aa+bb)^\omega.\]
  In this case, the automaton turns out to be "deterministic".
  The "transition system" $\T_1$ is the automaton of residuals of this language,
  which counts the parity (even or odd) of positions in the word.
\end{example}
\begin{figure}[h!]
  \centering
  \includegraphics[]{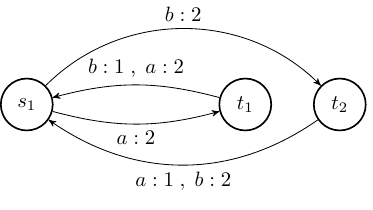}
  \caption{Alternating parity automaton for the "layered automaton" on the right of Figure~\ref{fig:layered-example-trees}.}\label{fig:alternating-example-mod2}
\end{figure}

\subparagraph{Deterministic parity automata as layered automata.} 
"Deterministic automata" can be seen as a special case of "layered automata", as shown in the next proposition.

\begin{proposition}\label{prop:deterministic-are-layered}
  Let $\Bb$ be a "deterministic parity automaton". 
  There exists a "layered automaton" $\Aa$ such that $\sem{\Aa} = \Bb$.
\end{proposition}
\begin{proof}
  Let $\Bb_{\geq x}$ be the "transition system" obtained by restricting $\Bb$ to
  transitions with priority $\geq x$. Let $\Aa$ be the "layered automaton" with
  "transition systems" $\T_x = \set{x}\times\Bb_{\geq x}$ and "morphisms"
  $\mu_x((x+1,q)) = (x,q)$ for all priorities $x$ and states $q$. (Note that the $x$-component is just added for obtaining disjointness of the sets $Q_x$.)
  It is direct to check that $\sem{\Aa} = \Bb$.
\end{proof}

\subparagraph{Minimal HD coB\"uchi automata.} 
\AP A (nondeterministic) ""coB\"uchi automaton"" is a ($\Sigma\times [1,2]$)-"transition system".
Kuperberg and Skrzypczak showed that "history deterministic" "coB\"uchi
automata" can be pruned not only into "semantically deterministic", but also
into ""safe deterministic"" automata at the same time.
The later property meaning that the non-deterministic choices
only appear in transitions carrying priority
$1$~\cite{Kup.Skr.Determinisation2015}.
Such automata are "simple by priorities automata" with "priorities" $[1,2]$ as for every $q\in Q$ and $a\in \S$ there is at most one transition $q\act{a:2}q'$.

Building on this work, Abu Radi and Kupferman showed that "HD" "coB\"uchi" automata can be minimised in polynomial time, and that the minimal automaton is unique if it is assumed "$1$-saturated"~\cite{Abu.Kup.Minimization2022}. 
\AP An automaton is ""$1$-saturated"" if whenever $q\act{a:1}p$, then
$q\act{a:1}p'$ for all $p'$ with $L(p') = L(p)$.
This minimisation is a special of the minimisation of "layered
automata" presented here. 

\begin{proposition}\label{prop:minimalCoBuchi-are-layered}
  Let $\Bb$ be a "semantically deterministic", "safe deterministic", "$1$-saturated" "coB\"uchi automaton".
  There exists a "layered automaton" $\Aa$ such that $\sem{\Aa} = \Bb$.
\end{proposition}
\begin{proof}
  Let $\T_1$ be the automaton of residuals of $\Bb$ and $\T_2 = \Bb_{\geq 2}$, with $\mu_1(q) = L(q)$.
  It is immediate to check that $\sem{\Aa} = \Bb$.
\end{proof}

\subsubsection{From alternating parity automata to layered automata}\label{sec:alternating-to-layered}
We have introduced layered automata as a tree shaped structure that is 
interpreted as an alternating parity automaton. 
Here we offer another perspective on this class.
We characterise directly alternating parity automata coming from "layered
automata" in terms of equivalence relations on states that should
be well-behaved.

% As we have introduced them, "layered automata" are syntactically different from "alternating parity automata".
% However, we will see now that automata that appear as the "semantics of" some "layered automaton" can be characterised directly in terms of "alternating parity automata" equipped with well-behaved equivalent relations, offering another perspective on this class. %of automata.

Consider a "simple by priorities" alternating parity automaton together with a family of $d$ nested equivalence relations on its states, that is:
\begin{equation*}
   Q\times Q\supseteq \; \intro*\lsim_1 \;\supseteq \; \lsim_2 \;\supseteq \dots \supseteq  \; \lsim_{d}.
\end{equation*}

\AP We consider the following properties. For all $a\in \S$, $x\in [d]$ and $p,q\in Q$:
\begin{itemize}
 \item \textbf{(""$x$-determinism"")} If $p\act{a:x}q$ and $p\act{a:x}q'$, then   $q'\sim_x q$.\\
 \item \textbf{(""$x$-coherence"")} If $p\lsim_x p'$ and $p\act{a:\geq x}q$  then $p'\act{a : \geq x}q'$ for some $q'\sim_x q$.\\
 \item \textbf{(""$x$-saturation"")} If $p\act{a:x}q$ then   for all $q'\sim_x q$ there is a transition $p\act{a:x}q'$. 
\end{itemize}

\begin{proposition}
  Let $\Bb$ be a "simple by priorities" automaton.
  Then, there is a "layered automaton" $\Aa$ such that $\Bb = \sem{\Aa}$ if and only if $\Bb$ can be equipped with $d$ nested equivalence relations $\lsim_1,\dots,\lsim_d$ satisfying  "$x$-determinism", "$x$-coherence" and "$x$-saturation".
\end{proposition}
\begin{proof}
  Assume that $\Bb = \sem{\Aa}$ for a "layered automaton" $\Aa$.
  Then, we can define $q\lsim_x p$ if $\wmu_x(q) =
  \wmu_x(p)$ and $q,p\in Q_{\geq x}$.
  By definition of $\sem{\Aa}$, it is immediate that the three properties hold.

  Assume that $\Bb$ can be equipped with $d$ nested equivalence relations $\lsim_1,\dots,\lsim_d$ satisfying the three properties above. 
  Let $\T_x$ be the "transition system" having as states the $\lsim_x$-classes of $\Bb$, and $[p]_x \actx{a} [q]_x$ if $p\act{a:\geq x} q'$ in $\Bb$ for some $q'\in [q]_x$ (which, by "$x$-saturation", implies that $p\act{a:\geq x} q'$ for all $q'\in [q]_x$).
  These transitions are well-defined by "$x$-coherence".
  The "transition system" $\T_x$ is "deterministic" by "$x$-determinism".
  We define $\mu_{x+1}([q]_{x+1}) = [q]_x$, which is a "morphism" by definition of the transitions of $\T_x$.
  The fact that $\Bb = \sem{\Aa}$ is a routine check.
\end{proof}

We note, however, that it may be possible to equip $\Bb$ with different equivalence relations satisfying the three properties above.
In that case, the proof above provides two different "layered automata" $\Aa, \Aa'$ such that $\sem{\Aa} = \sem{\Aa'} = \Bb$.

%As we will see in Section~\ref{subsec:safe-languages}, there is a unique family of coarsest relations making $\Bb$ a "layered automaton", given by the safe languages" of $\Bb$. \ac{Not exactly true, because we modify B}

\subsection{Semantics via safe languages}\label{sec:layered acceptance}
This section offers an alternative characterization of layered automaton
languages that avoids explicit reference to alternation. While our earlier
definition used the classical concept of alternating automata and allowed us to
introduce "uniform semantic determinism", we now present a more direct approach.
We introduce a restriction called "consistency" directly on "layered automata", which we prove equivalent to
"uniform semantic determinism". For consistent automata, we can characterize the
language of $\sem{\Aa}$ through an "acceptance@@layered" condition that does not
rely on alternation. This "acceptance@@layered" notion closely parallels the
natural color of words~\cite{Ehl.Sch.Natural2022} and the acceptance mechanisms
of rerailing automata~\cite{Ehlers.Rerailing2025}, as discussed further in
Section~\ref{subsec:comparison}.

\subparagraph{Safe languages.}  The ""$x$-safe language"" of a state $q\in Q_{\geq x}$ is defined
as
\[   \intro*\Lsafex(q) =  \set{w\in\S^*\cup\S^\w \mid \text{ the run over $w$ from } \wmu_x(q) \text{ is defined in $\T_x$}}. \]
That is, $\Lsafex(q)$ are the words on which there is a path from the "ancestor"
of $q$ at level $x$. Note that by "completeness" of $\T_1$ we have $\Lsafexx{1}(q) =
\S^*\cup\S^\w$ for every state $q$.
When we need to specify the automaton $\Aa$, we write $\LsafexxA{x}{\Aa}(q)$.

\begin{example}
  For the "layered automaton" on the left of Figure~\ref{fig:layered-example-trees}, we have
  \begin{align*}
  \Lsafexx{2}(q) =& \text{ Words without the factor } cc,\\
  \Lsafexx{3}(q_2) =& \text{ Words without the factors } cc \tand aa.
  \end{align*}
\end{example}

We have the following equalities relating the words
labelling paths with no "priority" $<x$ in $\sem{\Aa}$ and "$x$-safe languages". 
Recall that we write $\run{w:\geq x}$ if the minimal priority on this path is $\geq x$.

\begin{align*}
  \Lsafex(q) & =  \set{w\in\S^*\cup\S^\w \mid \text{there is a path } \wmu_x(q)\run{w}_x \text{ in $\T_x$}} \\
  & =  \set{ w\in\S^*\cup\S^\w \mid \text{there is a path } \classx{q} \run{w : \geq x} \text{ in } \sem{\Aa}} \\
  & =  \set{ w\in\S^*\cup\S^\w \mid \text{no path } \classx{q} \lrun{w} \text{ in } \sem{\Aa} \text{ sees a priority $<x$}} .
\end{align*}

\begin{definition}
  Consider a "layered automaton" $\Aa$.
  A word $w\in\S^\w$ is ""ultimately safe"" in layer $x$ if there is a decomposition $uw'=w$ and a state
  $p$ on layer $x$ such that $q_\init\act{u}_1\wmu_1(p)$ and $w'\in \Lsafex(p)$.
  We say that $w$ is ""accepted@@layered"" by $\Aa$ if the maximal layer in which $w$ is "ultimately safe" is even.
  If it is odd then we say that $w$ is ""rejected@@layered"" by $\Aa$.
\end{definition}

%In our proof of the characterisation, stated formally below in Theorem~\ref{thm:equivSemantics}, 
We also introduce a related notion of "strong acceptance" that allows us to
define "consistency", and eventually show the equivalence between acceptance by
$\Aa$ and the language of $\sem{\Aa}$

\subparagraph{Strong acceptance.}
We say that $w\in\S^\w$ is ""strongly accepted by $p$"" if $w\in\Lsafex(p)$, $p$
is on some even layer $x$, and for every decomposition $w = uw'$, there is no $p'\in Q_{x+1}$ with
$\wmu_x(p)\runx{u}\wmu_x(p')$ and $w'\in \Lsafexx{x+1}(p')$. 
The notion of $w$ being ""strongly rejected"" is defined similarly but for $x$ odd.

Intuitively, "strong acceptance" implies that on every run from $\classx{p}$ over
$w$ in $\sem{\Aa}$, only priorities $\ge x$ occur and $x$ occurs infinitely
often, regardless the choices of the players.
As $x$ is even, all these runs are accepting.
%In the notion of "strong acceptance" we can see quite clearly the role of the structure of layers in the acceptance process.

A priori, a word can be both "strongly accepted" by some state, and "strongly rejected" by a different state.
We forbid this by imposing an extra condition called "consistency", defined below.
We show that under the "consistency" assumption, all the acceptance mechanisms we introduce are equivalent: acceptance by the "alternating automaton" $\sem{\Aa}$, "acceptance@@layered" in $\Aa$, or existence (resp. non-existence) of a suffix that is "strongly accepted" (resp. "strongly rejected") by a state in a suitable residual (see Corollary~\ref{cor:equivalences-acceptance}).

\begin{example}
  In the "layered automaton" on the left of Figure~\ref{fig:layered-example-trees}, any word containing infinitely many factors $cc$'s is "strongly rejected" by the state $r\in Q_1$.
  The state $q\in Q_2$ "strongly accepts" the words that contain no factor $cc$ and infinitely many factors $aa$.
\end{example}

The next two lemmas link  "strong acceptance" in $\Aa$ and the
acceptance by $\sem{\Aa}$. 
\begin{lemma}\label{lem:language-p-strong-acc}
  Let $\Aa$ be a "layered automaton", $p$ a "leaf state".
  If a word $w$ is "strongly accepted"  by $\wmu_x(p)$, for some even $x$,   then $w\in  L(\sem{\Aa},p)$. 
  Analogously, if $w$ is "strongly rejected"  by $\wmu_x(p)$, for some odd $x$, then $w\not\in
  L(\sem{\Aa},p)$
\end{lemma}
\begin{proof}
  We show that any "run@@lay" over $w$ from $p$ produces the "priority" $x$ infinitely often (regardless the strategies of Eve and Adam). Since $w\in \Lsafex(p)$, only priorities $\geq x$ are produced by any run over $w$ from $p$. 
  If such a run does not produce priority $x$ infinitely often, then it is of the form $p \lrun{u} q \run{w': >x}$, a contradiction.
\end{proof}

\begin{lemma}\label{lem:layered-acceptance-strong}
  If $\Aa$ "accepts@@layered" $w$ then there is a decomposition $uw'=w$
  and a state $p$ on some even layer $x$ such that $q_\init\act{u}_1\wmu_1(p)$ and $w'$ is "strongly accepted" from $p$.
\end{lemma}
\begin{proof}
  Consider the maximal layer $x$ where $w$ is "ultimately safe", which is even as $w$ is "accepted@@layered".  
  By definition, we have a decomposition $uw'=w$, with a state $p$, a path $q_\init\act{u}_1\wmu_1(p)$, and  $w'\in
  \Lsafex(p)$. 
  %Observe that $x$ is even as $w$ is "accepted@@layered".
  It is direct to check that $w'$ is "strongly accepted" from $p$, because
  otherwise $x$ would not be the maximal layer where $w$ is "ultimately safe". 
\end{proof}

To get the desired characterisation we need one more condition that corresponds to "uniform semantic determinism".
% \AP We say that a word $w\in \S^\omega$ is ""finally strongly accepted"" (resp. ""finally strongly rejected"") by $\Aa$ if there is a decomposition $w = uw'$ and a state $p$ in an even (resp. odd) layer such that 
% \[ q_\init \runxx{u}{1} \wmu_1(p) \quad \tand \quad w' \text{ is "strongly accepted" (resp. "strongly rejected") by } p. \]

% We note that every word is "finally strongly accepted" or  "finally strongly rejected" by a given "layered automaton".
% But, a priori, it could be both.

\subparagraph{Consistency} \AP We say that a "layered automaton" $\Aa$ is
""consistent"" if there is no pair of states $p,p'\in
Q_{\geq 1}$ with $\wmu_1(p) = \wmu_1(p')$, and a word $w\in \S^\omega$ such that $w$ is "strongly
accepted" by $p$ and $w$ is "strongly rejected" by $p'$.

%The advantage of this notion over semantic determinism is that it is.
As stated below, this notion is equivalent to "uniform semantic determinism" of $\sem{\Aa}$.
We show in Proposition~\ref{prop:consistent-PTIME} that this property can be checked in polynomial time.

We state the key result saying that the semantics via "alternating parity automata" and via "strong acceptance" coincide for "consistent" "layered automata".
The proof, which requires introducing "longest suffix resolvers" in $\sem{\Aa}$, is relegated to Section~\ref{subsec:layered-HD} below.

\begin{restatable}{theorem}{equivSemantics}
  \label{thm:equivSemantics}
  Let $\Aa$ be a "consistent" "layered automaton".
  A word $w\in\S^\w$ is "accepted@@layered" by $\Aa$ if and only if $w\in
  L(\sem{\Aa})$.
  Moreover, $\sem{\Aa}$ is "uniformly semantically deterministic". 
\end{restatable}

\begin{corollary}\label{cor:consistent-iff-unifSD}
  A "layered automaton" is "consistent" if and only if it is "uniformly semantic deterministic".
\end{corollary}
\begin{proof}
  The left to right implication is given by the theorem above.
  The other direction follows directly from the definitions and Lemma~\ref{lem:language-p-strong-acc}.
\end{proof}

\begin{corollary}\label{cor:equivalences-acceptance}
  Let $\Aa$ be a "consistent" "layered automaton" and let $w\in \Sigma^\omega$. The following are equivalent:
  \begin{itemize}
  \item $w\in L(\sem{\Aa})$,
  \item $w$ is "accepted@@layered" by $\Aa$,
  \item There is a decomposition $uw' = w$ and a state
  $p$ such that $q_\init\act{u}_1\wmu_1(p)$ and $p$ "strongly accepts" $w'$, and
  \item There is no decomposition $uw' = w$ and a state
  $p$ such that $q_\init\act{u}_1\wmu_1(p)$ and $p$ "strongly rejects" $w'$.
\end{itemize}
\end{corollary}

% \begin{restatable}{theorem}{equivSemantics}
%     \label{thm:equiv-semantics}
%   Let $\Aa$ be a "consistent" "layered automaton".
%   Then, $L(\Aa) = L(\sem{\Aa})$.
%   Moreover, $\sem{\Aa}$ is "semantically deterministic". 
% \end{restatable}

% We note that, once language equality is established, "semantic determinism" of $\sem{\Aa}$ directly follows from Lemmas~\ref{lem:SD-1-equivalent} and~\ref{lem:SD-iff-consistent}.

\subsection{Layered automata are history deterministic and 0-1 probabilistic}\label{subsec:layered-HD}
In this section we prove Theorem~\ref{thm:equivSemantics} and the following propositions.

\begin{proposition}\label{prop:layered-are-HD}
  For every "consistent" "layered automaton" $\Aa$, the alternating parity automaton $\sem{\Aa}$ is "history deterministic".
\end{proposition}

%\subparagraph{0-1 probabilistic automaton.} 
%\AP Let $\Aa$ be a "simple by priorities parity automaton". 
\AP We say that  a "simple by priorities parity automaton" $\Aa$ is
a ""0-1 probabilistic automaton"" if for every random walk $\rho$
over a word $w\in \Sigma^\omega$ starting in $q_\init$, it holds:
\begin{itemize}
    \item $\mathrm{Pr}(\rho \text{ "is accepting@@run" }) = 1 \; \Longleftrightarrow \; w\in L(\Aa)$, and 
    \item $\mathrm{Pr}(\rho \text{ "is accepting@@run" }) = 0 \; \Longleftrightarrow \; w\notin L(\Aa)$.
\end{itemize}
We note that "0-1 probabilistic automata" are \emph{good for MDPs} in the sense of~\cite{HPSS0W.GoodForMdp2020}
See~\cite{Bai.Gro.Ber.Probabilistic2012,Hen.Pra.The.Resolving2025} for related notions.

\begin{proposition}\label{prop:layered-are-probabilistic}
  For every "consistent" "layered automaton" $\Aa$, the alternating parity automaton $\sem{\Aa}$ is a "0-1 probabilistic automaton".
\end{proposition}

We start by introducing \emph{"longest suffix resolvers"}.
This is a class of "resolvers" witnessing "HDness" of "layered automata" and enjoying some useful properties.

\subparagraph{Longest suffix resolvers.}
For a "priority" $x$, a finite word $u$, and an "$x$-state" $q$, we define the ""longest $x$-suffix of $u$ to
$q$"" to be the longest suffix $v$ of $u$ such that $p\runx{v} q$ in $\Tt_x$, for some $p$.
Observe that this suffix can be empty.

We let $Q$ and $\Delta$ be the sets of states and transitions of $\sem{\Aa}$, respectively.

\AP The ""longest suffix resolver""\footnote{Formally, we define a family of resolvers, as there are some non-deterministic choices in the definition. By a slight abuse of notation, we call \emph{the} longest suffix resolver any of them.} for Eve is a "resolver" $\s:(\{q_\init\}\cup \D^+)\times\S\to Q$ such that whenever we have a play $p_0\run{v} p \act{a:x}[q]_x$ in $\sem{\Aa}$  with $x$ odd, Eve chooses any state $q'\in \classx{q}$ admitting the "longest $(x+1)$-suffix" of $va$ to $\wh\m_{x+1}(q')$.
The "longest suffix resolver" for Adam is defined symmetrically.

The next lemma states the main property of "longest suffix resolvers". 
It guarantees that a play leaves an odd layer whenever there is a run with this property. 
The statement uses a more general notion of a ""longest suffix resolver
initialised"" to $u\in \S^*$.
It is a function $\s_{u}:(\{q_\init\}\cup \D^+)\times\S\to Q$ such that whenever we have a play
$p_0\run{v} p \act{a:x}[q]_x$ in $\sem{\Aa}$ with $x$ odd, Eve chooses any state
$q'\in \classx{q}$ admitting the "longest $(x+1)$-suffix" of $uva$ to
$\wh\m_{x+1}(q')$; instead of the longest suffix of $va$, as a "non-initialised
resolver" would. 
So the previous notion amounts to initialisation with the empty word.

\begin{lemma}\label{lem:longest-suffix-ensures}
  Let $x$ be an odd "priority" and $p$ a "leaf state" in a "layered automaton" $\Aa$. 
  Suppose $w = u\tilde{w}\in\S^\w$  admits a run $p \run{u : \geq x}q\run{\tilde{w}:\geq x+1}$ in $\sem{\Aa}$ (equivalently, $w\in \Lsafex(p)$ but not "strongly rejected" by $\wmu_x(p)$). 
  Then, for every $u_0$, any "longest suffix resolver" $\sigma$ for Eve initialised to $u_0$
  ensures that on any "$\sigma$-run" on the word $w$ from $p$, eventually all priorities are  $\geq x+1$.
\end{lemma}

\begin{proof}
  Consider a play from $p$ according a "longest suffix resolver" for Eve, initialised to some $u_0$, that  sees "priority" $x$ after some prefix $uw'$ of $u\tilde{w}$:
  \begin{equation*}
    p\run{u : \geq x}q\run{w' : \geq x}q'\act{a:x}q''\ .
  \end{equation*}
  In this case, the "longest suffix resolver" initialised to $u_0$ finds some suffix $v$ of $u_0uw'$ with
  $p_{va}\run{va}_{x+1}\wmu_{x+1}(q'')$ in $\T_{x+1}$ (with $p_{va}$ some state in this automaton).
  Observe that by assumption of the lemma, $w'$ is a suffix of $v$ since $\wmu_{x+1}(q)\run{w'a}_{x+1} \wmu_{x+1}(q'')$.
  So this path on $v$ can be presented as
  $p_{va}\run{v'}_{x+1}q_{w'a}\run{w'a}_{x+1}\wmu_{x+1}(q'')$, for some $v'$ and $q_{w'a}\in \T_{x+1}$.
  We call $q_{w'a}$ a support of $w'a$.
  Now consider the next time when the longest suffix run sees "priority" $x$, namely 
  \begin{equation*}
    p\run{u : \geq x}q\run{w' : \geq x}q'\act{a:x}q''\run{w'' : >x}q^{(3)}\act{b:x}\ .
  \end{equation*}
  This means that we have
  $q_{w'a}\run{w'a}_{x+1}\wmu_{x+1}(q'')\run{w''}_{x+1}\wmu_{x+1}(q^{(3)})$ with no
  $b$-transition from $\wmu_{x+1}(q^{(3)})$ in $\T_{x+1}$.
  Hence, $q_{w'a}$ cannot be a support for $w'aw''b$ and a new support is found
  by the reasoning above. 
  This argument shows that each time a play following the "longest suffix
  resolver" initialised to $u_0$ meets "priority" $x$, one "$(x+1)$-state" is eliminated
  as a potential support.
  Hence, the number of times such a play can
  see "priority" $x$ is bounded by the number of "$(x+1)$-states".
\end{proof}

% \begin{lemma}\label{lem:longest-suffix-ensures}
%   Consider a "layered automaton" $\Aa$ and odd "priority" $x$ and  $p\in Q_x$.
%   Suppose that for a word $uw\in\S^\w$ there is a state $q\in Q_{x+1}$ with two
%   paths $p\runx{u}\mu_x(q)$, and  $q\run{w}_{x+1}$.
%   Then, any "longest suffix resolver" for Eve from any $p_0\in[p]_x$ on the word $uw$ ensures that eventually  all priorities produced are $\geq x+1$. 
% \end{lemma}

For the statement of the next lemma recall that $(\Aa,q)$ stands for $\Aa$ where
$q$ is taken to be the initial state. 
For layered automata, the initial state should come from layer $1$, while the
initial state of $\sem{\Aa}$ is a leaf of $\Aa$.
\begin{lemma}\label{lem:lss-are-reslovers}
  Consider $\Aa$ a "consistent" "layered automaton", $q$ a state of $\sem{\Aa}$, and $\sigma$ and $\tau$
  "longest suffix resolvers" for Eve and Adam in $\sem{\Aa}$, respectively.
  If a word $w$ is "accepted@@layered" by $(\Aa,\wmu_1(q))$ then every $\s$-run of
  $(\sem{\Aa},q)$ over $w$ is accepting. 
  Symmetrically, if a word $w$ is "rejected@@layered" by $(\Aa,\wmu_1(q))$ then every
  $\t$-run of  $(\sem{\Aa},q)$ over $w$ is rejecting.
\end{lemma}
\begin{proof}
  We start by recalling Lemma~\ref{lem:layered-acceptance-strong} saying that if $w$ is "accepted@@layered" by $(\Aa,q)$ then
  there is a decomposition $u'w'=w$ and a state $p'$ on some even layer $y$ such
  that  $q\act{u'}_1\wmu_1(p')$, and $w'$ is "strongly accepted" from $p'$.

  Suppose the "longest suffix resolver" of Eve rejects when Eve plays on $w$ in $(\sem{\Aa},q)$. 
   Take a non-accepting run consistent with this "resolver",
  \[q\run{u''}p''\run{w'' : \geq x} \; \text{ in } \sem{\Aa},\]
  where for some odd $x$, eventually only priorities $\geq x$ appear and $x$
  appears infinitely often.
  We can assume that $u'$ is a prefix of $u''$. %(and therefore $w''$ a suffix of $w'$).

  First, we show that $w''$ is "strongly rejected" from $\wmu_x(p'')$ in $\Aa$.
  Since the run in $\sem{\Aa}$ avoids priorities $<x$, we have $w''\in \Lsafex(p'')$.  
  If $w''$ is not "strongly rejected", then by Lemma~\ref{lem:longest-suffix-ensures} the "longest suffix resolver" for Eve
  initialised to $u''$ eventually sees only "priorities" $\geq x+1$ from $p''$; a
  contradiction with the fact that $x$ appears infinitely often in the run in $\sem{\Aa}$.

  Now let $u''=u'v$ ($u'$ is assumed to be a prefix of $u''$). We have a path $p'\actxx{v}{y} p_y$ in the even layer $\T_y$, and $w''$ is "strongly accepted" from $p_y$.
  Moreover, $\wmu_1(p'') = \wmu_1(p_y)$.
  This contradicts the "consistency" of $\Aa$.
\end{proof}

The above lemma implies the first part of Theorem~\ref{thm:equivSemantics},
stating that for any "consistent" "layered automaton"
$\Aa$, "acceptance@@layered" by $\Aa$ is the same as acceptance by $\sem{\Aa}$.
It also proves HDness of $\sem{\Aa}$ (Proposition~\ref{prop:layered-are-HD}).

\begin{corollary}\label{cor:layered-is-HD}
  Let $\Aa$ be a "consistent" "layered automaton".
  Then, $\sem{\Aa}$ is "history deterministic", and "longest suffix resolvers" are "valid resolvers".
\end{corollary}

It remains to see that a consistent $\sem{\Aa}$ is uniformly semantically
deterministic.

\begin{lemma}
  If $\Aa$ is a "consistent" automaton then it is "uniformly semantically deterministic".
\end{lemma}
\begin{proof}
  We note that the definition of "acceptance@@layered" (by safe languages) in $(\Aa,q)$ only depends on $\wmu_1(q)$. 
  %That is, if $p \eqxx{1} q$ are "leaf states" in $\Aa$, $L(\Aa,p) = L(\Aa,q)$.
  Since Lemma~\ref{lem:lss-are-reslovers} holds for any  choice of initial state in $\Aa$, we have that for any $p \eqxx{1} q$
  \[ L(\sem{\Aa},p) = L(\Aa,p) = L(\Aa,q) = L(\sem{\Aa},q).\qedhere\]
\end{proof}

\subparagraph{Random strategies.}
Proposition~\ref{prop:layered-are-probabilistic} admits an almost identical proof.

\begin{lemma}\label{lem:random-walk-ensures}
  Let $p$ be a "leaf state" in a "layered automaton" $\Aa$.
  Let $uw\in\S^\w$ admitting a run $p\run{u : \geq x}q\run{w : \geq x+1}$ in $\sem{\Aa}$.
Then, any random walk on the word $uw$ from $p$ eventually produces only priorities $\geq x+1$. 
\end{lemma}
\begin{proof}
  Let $n$ be the number of "leaf states".
  Let $w = uw'a\tilde{w}$ and let 
  \begin{equation*}
    p\run{u : \geq x}q'\run{w' : \geq x}q''\act{a:x} \classx{\tilde{q}}
  \end{equation*}
  be a random walk from $p$ on a prefix of $uw$, producing priority $x$.
  By assumption, there is at least one state $\tilde{q}'$ in $\classx{\tilde{q}}$ such that $\tilde{w}\in \Lsafexx{x+1}(\tilde{q}')$.
  Therefore, the random walk goes to this state with probability at least $1/n$. The probability of seeing $k$ transitions with priority $x$ after the prefix $u$ is therefore less than $(1-1/n)^k \xrightarrow[k \to \infty]{} 0$.
\end{proof}

\begin{proof}[Proof of Proposition~\ref{prop:layered-are-probabilistic}]
  Let $w\notin L(\Aa)$ (the case $w\in L(\Aa)$ is symmetric), and let $x$ be an even "priority". 
  We show that the probability that a random walk over $w$ produces $x$ as minimal priority infinitely often is $0$.
  The probability of this event is the probability of obtaining a path of the form:
  \[q_\init \run{u} p \run{w':\geq x} \; \text{ producing $x$ infinitely often, with } w = uw'.\]
  For every such path $q_\init \run{u} p$, since $w'\in \Lsafex(p)$ and is not "strongly rejected" by $\wmu_x(p)$, there is a run $p\run{u' : \geq x}q\run{w'' : \geq x+1}$, with $w' = u'w''$.
  Therefore, by Lemma~\ref{lem:random-walk-ensures}, the probability of obtaining a path $p \run{w':\geq x}$ producing $x$ infinitely often is $0$.
\end{proof}

\subsection{Some decision procedures in PTIME}\label{subsec:decision-procedures}

We show that we can check whether a given layered automata $\Aa$ is "consistent" in polynomial time. 
We can also check emptiness and language inclusion of "consistent" "layered automata" in polynomial time.

\begin{proposition}\label{prop:consistent-PTIME}
It is decidable in polynomial time whether a given layered automaton is "consistent".
\end{proposition}
\begin{proof}
  We check for a pair $p,p'\in Q_{\geq 1}$ conflicting with the definition of "consistency", that is, $\wmu_1(p)=\wmu_1(p')$ and there is a word $w$ "strongly accepted" by $p$ and "strongly rejected" by $p'$.
  There are polynomially many candidates, %$p,p'\in Q_{\geq 1}$ and $x,x'\in [d]$, 
  so we can iterate through all of them. Given $p\in Q_x$ and $p'\in Q_{x'}$, with $x$
  even and $x'$ odd, we can check whether there exists a conflicting word $w \in
  \Sigma^\omega$ by a game. In this game, a player \emph{Prover} tries to
  find a word witnessing an inconsistency, and \emph{Refuter} tries to show
  that the constructed word is not a witness.
  
  The game starts in $(p,p')$. In each round, Prover plays the next letter, and
  Refuter chooses the next pair of states according to the transition function
  on the chosen letters. Prover wins if on the run from $p$ only
  priorities $\ge x$ occur and $x$ occurs infinitely often, and on every run
  from $p'$ on $w$, only priorities $\ge x'$ occur and $x'$ occurs infinitely
  often. This is a generalised Büchi game which can be solved in polynomial
  time. We claim that Prover has a winning strategy if and only if $w$ is "strongly
  accepted" by $p$ and 
  "strongly rejected" by $p'$.
  Clearly, if this condition is satisfied, then Prover can play the word $w \in \Sigma^\omega$, and all the runs that Refuter can play are such that Prover wins the game.

  For the other direction, assume that Prover has a winning strategy in the game and let her play according to the strategy. Any word produced by this strategy has the property that all priorities visited from $p$ are $\ge x$, and all priorities visited from $p'$ are $\ge x'$. Let Refuter use a "longest suffix resolver" in both components for selecting the runs, and let $w \in \Sigma^\omega$ be the word produced by Prover's strategy against this strategy of Refuter. If there is a run on $w$ from $p$ that does not produce $x$ infinitely often, then this run produces only priorities $\ge x+1$ from some point onwards. The "longest suffix resolver" used by Refuter for selecting the run would then result in a run that sees only priorities $\ge x+1$ from some point onwards (see Lemma~\ref{lem:longest-suffix-ensures}). Similarly for $p'$.
\end{proof}

\begin{proposition}\label{prop:emptiness-PTIME}
Given a "consistent" "layered automaton" $\Aa$, we can check in polynomial time whether $L(\Aa) = \emptyset$.
\end{proposition}
\begin{proof}
  Since we assume that all states in $\T_1$ are reachable, $L(\Aa) \not= \emptyset$ if and only if there is some state $p_0\in Q_{\geq 1}$ in an even layer $x$ that "strongly accepts" some word.

  We can check if this is the case via a B\"uchi game similar to the one considered above.
  Two players play over pairs $(p,q)\in Q_x \times Q_{x+1}$.
  The starting position is $(p_0,q_0)$, for $q_0$ any state in $\mu_x^{-1}(p_0)$.
  At each step, Prover gives the next letter such that the obtained word is safe from $p_0$. The next state $p\actx{a} p'$ is given deterministically.
  If transition $q\actxx{a}{x+1}q'$ exists, this transition is taken. If not,
  Refuter can chose to reset to a state in $\mu_x^{-1}(p')$.
  Prover wins if there are infinitely many resets at level $x+1$.
  It is clear that Prover wins this game if and only if $p_0$ "strongly accepts"
  some word. 
  Indeed, Refuter can use the "longest suffix resolver" to win whenever the
  produced word is not "strongly accepted".
\end{proof}

\begin{proposition}\label{prop:inclusion-PTIME}
Given two "consistent" "layered automaton" $\Aa$ and $\Aa'$, we can check in polynomial time whether $L(\Aa) \subseteq L(\Aa')$.
\end{proposition}
\begin{proof}
  We show that we can reduce this problem to solving a Rabin game with $3$ Rabin pairs of polynomial size.
  This allows to conclude, as solving Rabin games with a fix number of Rabin pairs can be done in polynomial time (see e.g.~\cite[Thm.~3.1]{BookGamesOnGraphs2025}).
  The winning condition is in fact slightly simpler: 
  if the game stays in the first phase Prover loses, otherwise Prover wins in
  phase 2 if two B\"uchi conditions are satisfied.

In this game, Prover tries to show that $L(\Aa) \nsubseteq L(\Aa')$ by giving a word in $L(\Aa)\setminus L(\Aa')$.
Prover gives letters one by one.
In a first phase, these letters determine two runs:
\[q_\init \actxx{a_1}{1}q_2 \actxx{a_2}{1} q_3 \actxx{a_3}{1} \dots \text{ in } \T_1 \quad \tand \quad q'_\init \actxx{a_1}{1}q'_2 \actxx{a_2}{1} q'_3 \actxx{a_3}{1} \dots \text{ in } \T'_1. \]

At any point, Prover can decide to enter a phase 2 in $\Aa$: he goes to some pair $(p_i,s_i)$, with $p_i\in \wmu_1^{-1}(q_i)$ in some even layer $\T_x$ and $s_i\in \mu_x^{-1}(p_i)$.
Similar, at any point Prover can enter a phase 2 in $\Aa'$, going to $(p'_i,s'_i)$, with $p'_i\in \wmu_1^{-1}(q'_i)$ in some odd layer $\T'_{x'}$ and $s'_i\in \mu_{x'}^{-1}(p'_i)$.

If for $\Aa$ or $\Aa'$, Prover never enters phase 2, then he loses.

When automaton $\Aa$ is in phase 2 and in a pair $(p_i,q_i)$, the game is as in the previous proposition.
The transition $p_i\actx{a_i}p_{i+1}$ must exist (else, Prover loses automatically).
If the transition $q_i \actxx{a_i}{x+1} q_{i+1}$ exists in $\T_{x+1}$, then we go to $q_{i+1}$ in the second component.
If this transition does not exist, then this is an accepting edge for the first B\"uchi condition, and Refuter resets the second component to any $q_{i+1}\in \mu_x^{-1}(p_{i+1})$.

Phase $2$ in $\Aa'$ is identical. We follow the run in $\T'_{x'}$ in the first component, which must exist (else, Prover loses). Refuter can choose how to reset the second component in $\T'_{x'+1}$ whenever necessary. 
Each time a reset happens, we see an accepting edge for the second B\"uchi condition.

It is clear that there exists a word in $L(\Aa)\setminus L(\Aa')$ if and only if Prover can see infinitely often edges in both B\"uchi conditions.
\end{proof}

\subsection{Connections and comparison with other models}\label{subsec:comparison}

In this section, we discuss how "layered automata" relate to
other models from the literature.
Sometimes we refer to the minimal layered automaton for a language that
is introduced in Section~\ref{sec:concise}.
We leave for future work a thorough analysis on the size of "layered automata" compared to other representations of "$\omega$-regular languages"; we include here just some remarks on this matter.

\subparagraph{Deterministic and alternating automata.} First, we notice that "consistent" "layered automata" can be exponentially smaller than "deterministic" automata, as this is already the case for "HD" "coB\"uchi automata"~\cite[Thm.~1]{Kup.Skr.Determinisation2015}.
On the other hand, they can be doubly exponentially larger than alternating
parity automata. The double exponential gap already occurs for automata over
finite words~\cite[Thm.~5.3]{CKSL.Altenation1981}. It transfers to "consistent"
"layered automata"  because on the first layer they need to have at least as
many states as there are residuals in the language. 

\subparagraph{Minimal "HD" "coB\"uchi" automata.} As discussed in Proposition~\ref{prop:minimalCoBuchi-are-layered}, minimal "HD" "coB\"uchi" automata can be seen as a special case of "layered automata".
The properties characterising minimal "layered automata" that we introduce in Section~\ref{sec:concise} are a generalisation of the notions of safe minimality and safe centrality introduced by Abu Radi and Kupferman~\cite{Abu.Kup.Minimization2022}.

\subparagraph{Muller languages.}
    The reader familiar with Zielonka trees~\cite{Zielonka.Infinite1998} and
    their parity automata~\cite{DJW.Memory1997,CCFL.ACD2024} may have noticed
    close resemblances between "layered automata" and these notions.
    Indeed, it is not difficult to see that the Zielonka tree of a Muller
    language $L$ can be taken as the "tree" of a "layered automaton"
    "recognising@@lay" $L$. In this case, the associated automaton $\sem{\Aa}$
    is deterministic by pruning. Any such pruning produces the deterministic
    parity automaton of the Zielonka tree for the language $L$.
    We refer to~\cite[Section~4]{CCFL.ACD2024} for definitions on the parity automaton of a Zielonka tree.

\subparagraph{Positional languages.} In 2024, Casares and Ohlmann characterised which "$\omega$-regular languages" are \emph{positional}, that is, languages $L$ such that the existential player can always play optimally using positional strategies in games with winning condition $L$~\cite{Cas.Ohl.Positional2024}.
In this characterisation, they use so-called \emph{signature automata}. 
It turns out that signature automata are just "layered automata" with extra properties that ensure positionality.

More precisely, we can restate their characterisation (point (2) in~\cite[Thm.~3.1]{Cas.Ohl.Positional2024}) as follows: An "$\omega$-regular language" $L$ is positional if and only if the minimal "layered automaton" $\Aa_L$ is such that:
\begin{enumerate}
    \item Nodes at even layers have at most one child,
    \item\label{it:total-order} For each $x$, the states in every SCC of $\T_x$ are totally ordered by inclusion of "$x$-safe languages",
    \item The residuals are totally ordered by inclusion (this can be seen as a special case of the previous condition), and
    \item Each even layer $x$ is \emph{progress consistent}. That is, if $p\runx{u} q$ and $\Lsafex(p)\subsetneq \Lsafex(q)$, then $u^\w$ is "strongly accepted" by $p$.
\end{enumerate}

\subparagraph{Chains of coB\"uchi automata (COCOA).} 
In~\cite{Ehl.Sch.Natural2022}, Ehlers and Schewe propose to represent an "$\omega$-regular language" by a decreasing sequence of "coB\"uchi" languages  $L_1\supseteq L_2\supseteq \dots \supseteq L_d$,\footnote{The definitions in \cite{Ehl.Sch.Natural2022} use $0$ as minimal color, whereas we use $1$ as minimal layer. As mentioned in \cite{Ehl.Sch.Natural2022}, the definitions can easily be adapted to the setting with $1$ as minimal color.} where $L_x$ are the words that are \emph{not at home at level $x-1$} (we refer to~\cite[Def.~1, Thm.~7]{Ehl.Sch.Natural2022} for definitions\footnote{We believe that in~\cite[Thm.~7]{Ehl.Sch.Natural2022} it should read ``i.e., $w$ is \emph{not} at home in color $i-1$''.}). A word $w$ is in the language $L$ represented by the chain if the maximal index $x$ such that $w \in L_x$ is even. %, that is, $L = \Sigma^\omega \setminus (L_1 \setminus (L_2 \setminus (L_3 \setminus (\cdots \setminus L_d)\cdots )$.

It is straightforward to obtain a COCOA representation from a "consistent" "layered automaton".
For every layer $x$ define an "HD" "coB\"uchi automata" from the transition
system $\T_x$ by interpreting the existing transitions as
$2$-transitions of the "coB\"uchi automaton", and for all missing transitions
adding $1$-transitions to all states of the corresponding residual class.
Therefore, the size of a COCOA representation of a language is never bigger than the size of the minimal "layered automaton" for that language. But the minimal layered automaton may be exponentially larger than the minimal COCOA for a language, as illustrated by Example~\ref{ex:COCOA-to-layered-exponential} below.

Since COCOA represent languages as Boolean combinations of "coB\"uchi
automata", it is therefore not surprising that they are smaller than a
representation by a single automaton. In order to use COCOA in applications like
synthesis, one needs to obtain a "history deterministic" automaton from it. This
requires to make a product of all the "coB\"uchi automata" that constitute it~\cite{Ehl.Kha.Fully2024}.
Therefore, the number of states of the final "alternating parity automaton" coming from a COCOA may be exponentially larger than the "alternating parity automaton" given by a "layered automaton". 
%This is the case, for instance, if the language can be recognised by the "automaton of residuals" (sometimes called \emph{languages with informative right congruence}).
%For these languages, all $k$ levels of a COCOA will have size $n$, the number of residuals of the language, so the product will have $n^k$ states, while a minimal deterministic automaton only requires $n$.
%However, this is not the only ingredient that may lead to an exponential blow-up.
In Example~\ref{ex:COCOA-exponential-product} we give such an example.\footnote{We consider here the naive product of \cite{Ehl.Kha.Fully2024}. They also define a reduced and an optimised product that ignores some product states. A more fine-grained comparison of these products with layered automata is left as future work.} 
Moreover, the languages of this example is prefix-independent (it has a single residual), so this blow-up is not due exclusively to taking a product of all residuals of the language.

One of the central definitions in~\cite{Ehl.Sch.Natural2022} is the \emph{natural color} of a word $w$ with respect to a language $L$.
We believe that the natural color is visible in the minimal "layered automaton"
for $L$, namely it coincides with the maximal layer $x$ in which $w$ is
"ultimately safe".

\begin{example}\label{ex:COCOA-to-layered-exponential}
The language in this example is a minor modification of a language used in
\cite{AngluinBF18} for showing an exponential gap between saturated families of DFAs and
DPAs.
Let $\S = \{1,\dots, k\}$ and consider the language $L$ of all $w \in \S^\omega$ satisfying the following two conditions:
\begin{itemize}
\item $w$ contains an even number of different letters infinitely often,
\item from some point onwards, for every two-letter infix $ij$, we have $j \in \{1, \ldots, i+1\}$.
\end{itemize}
The second condition ensures that for all words in $L$, the set of symbols that occur infinitely often forms an interval inside $\{1, \ldots, k\}$. 
For $k=4$, the minimal "layered automaton" for this language is shown in Figure~\ref{fig:COCOA-to-layered-exponential}. The automaton has $k+1$ layers. Since the language is prefix independent, the first layer has only one state looping on all letters. For $x \ge 1$, the SCCs on layer $x$ correspond to the intervals of size $k-x+2$. However, an interval might appear several times. For example, on layer $4$, the interval $[2,3]$ appears twice because it is contained in $[1,3]$ and in $[2,4]$. And on layer $5$ the singleton interval $[3,3]$ appears three times because it is contained in $[3,4]$ and in $[2,3]$ which itself appears already twice on the previous layer. Because of the binary tree structure starting from layer $2$, there are $2^{k-1}$ states on layer $k+1$. These are the states used by the corresponding alternating "HD" automaton. If one builds the COCOA for this language, then each interval appears only once on each level of the COCOA, so the number of states of each coB\"uchi automaton in the COCOA is at most quadratic in $k$.

As a further remark, for those readers familiar with $\omega$-semigroups, the size of the syntactic $\omega$-semigroup of this language is of order $k^4$ because an element of the semigroup only needs to store the first, last, minimal, and maximal number occurring in a word. The first and the last number in a word are needed to keep track of the second condition when concatenating two words. And assuming that the second condition is satisfied, the number of different letters occurring infinitely often can be derived from the minimal and maximal number occurring infinitely often, because letters can increase by at most one in each step. 
\end{example}
\begin{figure}
\centering
  \includegraphics[width=0.9\textwidth]{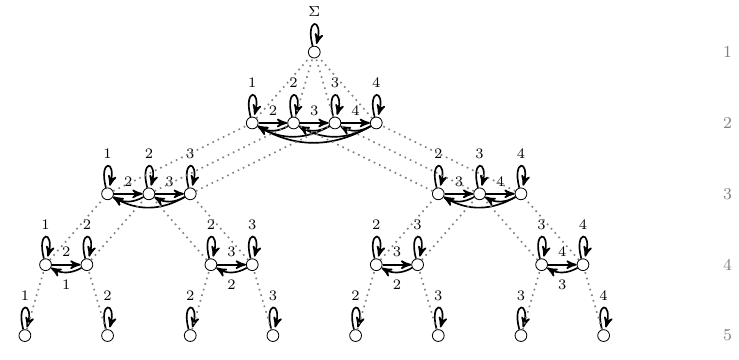}
    \caption{The layered automaton for $k=4$ in Example~\ref{ex:COCOA-to-layered-exponential}. In each SCC on each layer, all edges going to the same state have the same label. For readability, we have omitted the labels for the edges from right to left on layers~2 and~3. \label{fig:COCOA-to-layered-exponential}}
\end{figure}

% \begin{example}\label{ex:COCOA-residuals}
%     Fix $k,n\in \Nat$.
%     Let $\S = \{1,\dots, k\}$ and consider the parity language restricted to positions that are divisible by $n$, that is:
%     \[L_n = \{w\in \S^\w\mid \min_x \{x \text{ appears infinitely often in positions that are multiples of } n \} \; \text{ is even}\}.\]
%     This language has $n$ residuals, but can be recognised by the "automaton of residuals".
%     The minimal "layered automaton" is a tree of height $k$, where each level has $n$ states to count divisibility by $n$ . Therefore, $\sem{\Aa}$ has just $n$ states.
%     The minimal "coB\"uchi automaton" for each language $L_x$ also has $n$ states, as divisibility by $n$ needs to be counted. Therefore, the automaton obtained as the product of all the levels of the minimal COCOA has $n^k$ states.
% \end{example}

\begin{example}\label{ex:COCOA-exponential-product}
    Let $\S = \{a_1,\dots, a_k,b_1,\dots, b_k,c\}$ and consider the language containing all words satisfying one of the following two conditions:
    \begin{itemize}
    \item $\{b_1, \ldots, b_k\}$ occur finitely often, and $\min(\{k+1\} \cup \{x \mid a_x \text{ occurs infinitely often}\}$ is even, or
    \item $\{a_1, \ldots, a_k\}$ occur finitely often, and $\min(\{k+1\} \cup \{x \mid b_x \text{ occurs infinitely often}\}$ is even. 
    \end{itemize}
%So basically the "layered automaton" has to check if 
That is, this is the language of words where one sort of the numbered symbols occurs finitely often, while the parity condition is satisfied on the other sort of symbols (the $k+1$ in the minimum takes care of the case that only $c$ occurs infinitely often). 
The minimal "layered automaton" needs two strands for checking the two parity conditions, as shown in Figure~\ref{fig:COCOA-exponential-product} for $k=4$. The minimal COCOA is obtained by taking each layer $x$, interpreting the loops on the states as $2$-transitions, and saturating it with $1$-transitions for the missing letters on that layer. Taking the (naive) product of these coB\"uchi automata, results in exponentially many states, most of which only have a $c$-loop.
\end{example}
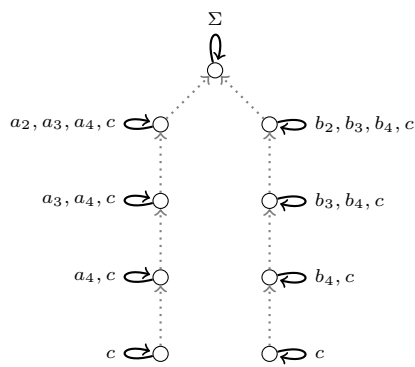
\begin{figure}
\def\xd{.7cm}  % horizontal distance between nodes
\def\yd{-1.2cm} % vertical distance between layers
\begin{center}
\begin{tikzpicture}[node distance=.8cm,
      parent/.style={
        dotted, ->, thick, draw=black!50, font=\footnotesize
      },
      layer/.style={
        ->,thick, draw=black, font=\scriptsize
      },
      ln/.style={
        font=\footnotesize,text=black!50
      },
      state/.style={circle, draw=black, minimum size=2mm, inner sep=0pt}
  ]

  \node[state] (e) at (0,0) {};
  \node[state,below left=of e] (a2) {};
  \node[state,below right=of e] (b2) {};
  \node[state,below=of a2] (a3) {};
  \node[state,below=of b2] (b3) {};
  \node[state,below=of a3] (a4) {};
  \node[state,below=of b3] (b4) {};
  \node[state,below=of a4] (a5) {};
  \node[state,below=of b4] (b5) {};
  
 %   % layer numbers
 %   \path (4444) -- ++(3*\xd,0) node[ln] {$5$}
 %                -- ++(0,-\yd)  node[ln] {$4$}
 %                -- ++(0,-\yd)  node[ln] {$3$}
 %                -- ++(0,-\yd)  node[ln] {$2$}
 %                -- ++(0,-\yd)  node[ln] {$1$}
 %                ;
                 
    % edges 
    \path[layer] (e) edge[loop above] node{$\S$} ();

    \path[layer] (a2) edge[loop left] node[left]{$a_2,a_3,a_4,c$} ();
    \path[layer] (b2) edge[loop right] node[right]{$b_2, b_3,b_4,c$} ();
    \path[layer] (a3) edge[loop left] node[left]{$a_3,a_4,c$} ();
    \path[layer] (b3) edge[loop right] node[right]{$b_3,b_4,c$} ();
    \path[layer] (a4) edge[loop left] node[left]{$a_4,c$} ();
    \path[layer] (b4) edge[loop right] node[right]{$b_4,c$} ();
    \path[layer] (a5) edge[loop left] node[left]{$c$} ();
    \path[layer] (b5) edge[loop right] node[right]{$c$} ();

    % morphism
    \path[parent] (a2) -- (e);
    \path[parent] (b2) -- (e);
    \path[parent] (a3) -- (a2);
    \path[parent] (b3) -- (b2);
    \path[parent] (a4) -- (a3);
    \path[parent] (b4) -- (b3);
    \path[parent] (a5) -- (a4);
    \path[parent] (b5) -- (b4);
  \end{tikzpicture}
    \caption{The layered automaton for $k=4$ in Example~\ref{ex:COCOA-exponential-product}. \label{fig:COCOA-exponential-product}}
    \end{center}
\end{figure}

\subparagraph{Rerailing automata.} %~\cite{Ehlers.Rerailing2025}.}
In 2025, Ehlers introduced \emph{rerailing automata}~\cite{Ehlers.Rerailing2025}, a kind of automata with specific semantics:
a word is accepted if the maximal priority of its runs is even. %Moreover, they must satisfy a semantic ``rerailing property'', inspired from the properties of minimal "HD" "coB\"uchi automata".

We believe that a  "consistent" "layered automaton" $\Aa$ can be transformed into a rerailing automaton by further saturating $\sem{\Aa}$. More precisely, whenever there is a transition $p\act{a:x} \classx{q}$, we can add transitions $p\act{a:y} \classxx{q}{y}$ for all $y<x$ without changing the language recognised by the automaton. The automaton obtained in this way has the rerailing property.
Moreover, a word is "accepted@@layered" by $\Aa$ if and only if it is accepted with the semantics of rerailing automata after saturation.

For the other direction, it is unclear to us if rerailing automata are "layered", or even simply "HD" "alternating". However, it may be the case that the minimal rerailing automaton of a language and its minimal "layered automaton" are essentially the same object.

\section{A canonical minimal layered automaton}\label{sec:concise}
In this section, we show that each "$\omega$-regular" language admits a unique minimal "consistent" "layered automaton".
This automaton is characterised by three properties: "normality", "centrality" and "safe minimality" (extending the properties of minimal "HD" coB\"uchi automata~\cite{Abu.Kup.Minimization2022}).

\begin{restatable}{theorem}{conciseAreUnique}
\label{thm:concise-are-unique}
    Let $\Aa$ and $\Bb$ be two equivalent "consistent" "layered automata" that are "normal", "central" and "safe minimal". 
    Then, $\Aa$ and $\Bb$ are "isomorphic@@lay" (as "layered automata").
\end{restatable}

Two "layered automata" $\Aa$, $\Bb$ are \emph{isomorphic} if there is a bijective function between their states that preserves both the structure of the "transition systems" and the "tree structure".
That is, a function $\f\colon \Aa \to \Bb$ such that the restriction to each layer $\T_x$ is an "isomorphism@@TS" of "transition systems" and preserving the child relation given by the morphisms $\mu_x$. 
(See Section~\ref{subsec:morphisms} for further definitions on "morphisms of layered automata".)

\subsection{Normality, centrality and safe minimality}

\AP We define the three properties that identify minimal "consistent" "layered automata".

\begin{definition}
  A "layered automaton" is in ""normal form"" if for all $x\geq 2$ and $q\in Q_x$ the following two properties hold.  
  \begin{description}
    \item[""N1""]\label{it:normal-return} If $q\runx{u}q'$ for some $u\in \S^*$ then $q'\runx{v}q$ for some $v\in \S^+$.
    \item[""N2""]\label{it:normal-small} For all $p\in \mu_x^{-1}(q)$, there is some $u$ for which $q \runxx{u}{x}$ is defined in $\T_{x}$ but $p\runxx{u}{x+1} \bot$ in $\T_{x+1}$.
  \end{description} 
\end{definition}

Being in "normal form" ensures that there are no evidently useless transitions,
and moreover the "priorities" of transitions in $\sem{\Aa}$ cannot be lowered without modifying the recognised language.
%The intuitive meaning of the two properties is easier to see by looking at their negations. 
The first property means that each $\T_x$ is a union of non-trivial SCCs.
Intuitively, this can be achieved just by removing transitions changing of SSC. 
A removal of such transition leads to its "priority" being lowered to $x-1$ in $\sem{\Aa}$. 
This does not change the accepted language, as such a transition cannot appear
infinitely often on a run without another transition of "priority" $\leq x-1$
appearing infinitely often too.

The second property says that $\mu_x$ does not completely cover all edges of
some SCC of $\T_{x}$, or in other words that the "safe languages" of children are always strictly included in the "safe language" of the parent. Otherwise, the SCC of $p$ does not "strongly accept"/reject anything. 
In the "semantics automaton" the property means that if $p\act{a:\geq x}$, then 
there is $u$ and a run $p\run{u:x}$ from $p$ producing $x$ as minimal priority.\\

\AP An ""$x$-SCC"" of a "layered automaton" is just an SCC of $\T_x$.
The next property requires that no "$x$-SCC" can be simulated by another one.
We define this formally below.

\begin{definition}
  For $x\geq 2$, and $p,q\in Q_x$ we define $p\intro*\fleqx q$ if $p\eqxx{x-1} q$ (that is, $\mu_{x-1}(q)=\mu_{x-1}(p)$)  and $\Lsafex(p)\incl \Lsafex(q)$.
\end{definition}

\begin{definition}
  \AP A "layered automaton" is ""centralised"" if for every $x\geq 2$ and for
  all states $p,q\in Q_{x}$,  if $p\fleq_{x}q$ then $p,q$ are in the same
  "$x$-SCC".
\end{definition}
Intuitively, if $p\fleq_{x}q$ are in different $x$-SCCs, then the $x$-SCC of $p$ and all its descendants can be removed. Because $p \fleq_x q$, this removal does not change anything about words being "strongly accepted"/rejected in the parent SCC of $p$ and $q$. Therefore, such states can be removed without changing the language of the automaton.
\vskip1em
%\subparagraph{Safe minimality.} 

Finally, on each layer we can perform a similar minimisation construction as for
finite automata by merging all states on layer $x$ that are language equivalent for all layers $y \le x$. 

\AP We define by induction on $x$ the relation $p\intro*\seq_x q$, for  $p,q\in \Qgeq{x}$:
\begin{align*}
  p \seq_1 q \quad  &\text{ iff } \;\;  L(p) = L(q), \;\; \tand \\[2mm]
  p \seq_x q \quad &\text{ iff } \;\; p \seq_{x-1} q  \;\; \tand \;\; \Lsafex(p) = \Lsafex(q) \;\; \text{ (for $x>1$)}.
\end{align*}

\begin{definition}
  A "layered automaton" is ""safe minimal"" when for all
  $x \ge 1$ and states $p,q \in Q_x$: if $p \seq_x q$ then $p=q$, 
\end{definition}

\begin{remark}\label{lem:sim-subset-approx}
  In a "consistent" "layered automaton" $\eqx \, \subseteq \,
  \seq_x$ for all $x$.
  This is because $p\eqx q$ means $\wmu_x(p)=\wmu_x(q)$. By the "consistency" assumption, $L(p)=L(q)$, and for $1<y\leq x$ the equality
  $\Lsafexx{y}(p)=\Lsafexx{y}(q)$ follows trivially.
  "Safe minimality" says that the "safe language" equivalence implies $\eqx$. 
  Namely, if $p,q\in Q_x$ with $L(p)=L(q)$ and $\Lsafexx{y}(p)=\Lsafexx{y}(q)$ for all $y\leq x$ then $p=q$.
\end{remark}

\subsection{Central sequences}\label{sub:central-sequences}
As an important technical tool for minimal "layered automata", we introduce the notion of a "central sequence" of a state $p$, which is a finite word that identifies the state $p$ uniquely (extending the synchronising property from~\cite[Sect.~4.2.2]{Cas.Ohl.Positional2024} and the notion of central sequence for coBüchi automata
from the full version of \cite{Lod.Wal.Congruences2025}).
The three properties "normal", "safe minimal", and "centralised" characterize minimal layered automata (as stated in Theorem~\ref{thm:concise-are-unique}), and they also guarantee that every state has a
"central sequence" (Lemma~\ref{lem:central-sequence}). 
This is an important step towards showing Theorem~\ref{thm:concise-are-unique}.

Let $\Aa$ be a "layered automaton" and $p\in Q_x$ an "$x$-state".
\AP We say that $z_p\in \S^+$ is a ""central sequence"" for $p$ if:
\begin{itemize}
  \item $p \runx{z_p} p$ in $\T_x$;
  \item for all $q\in Q_x$ with $q\eqxx{x-1} p$, either $q\runx{z_p} p$ or
  $q\runx{z_p} \bot$;
  \item for all $q'\in Q_{x+1}$ with $q'\eqxx{x-1} p$ we have $q'\runxx{z_p}{x+1} \bot$.
\end{itemize}

In $\sem{\Aa}$, this means that $p \run{z_p:x} \classx{p}$, and that from any
other state in $\classxx{p}{x-1}$, either $q \run{z_p:x}\classx{p}$ or $q
\run{z_p:<x}$, meaning $z_p$ produces a priority $<x$. 
Observe that the third condition is needed to ensure that the runs in question
do not have priority $>x$.

%\AP For $p\in Q_{\geq x}$, we say that $z_p$ is ""$x$-central for $p$"" if it is "central@@seq" for $\wmu_x(p)$.

\begin{lemma}\label{lem:central-sequence}
  Let $\Aa$ be a "normal", "centralised" and "safe minimal" "layered automaton".
  Then, for every $x\geq 2$, every state $p\in Q_x$ admits a "central sequence".
\end{lemma}
\begin{proof}
  Let $P_x = \set{p_1,\dots, p_k}$ be the "$x$-states" in $\classxx{p}{x-1}$.

  Assume first that $p$ is $\fleqx$-maximal in $P_x$. 
  By "safe minimality" of $\Aa$, there is no $p'\in P_x\setminus \set{p}$ with $\Lsafex(p')=\Lsafex(p)$.
  Since $p$ is $\fleqx$-maximal, for every $p'\in P_x\setminus
  \set{p}$ we can find $v_{p'}\in\Lsafex(p)\setminus\Lsafex(p')$.
  By normality, we can assume that $p\runx{v_{p'}}p$.

  Let $u_0\in \S^+$ such that $p\runx{u_0}p$ and
  $\mu^{-1}_x(p)\runxx{u_0}{x+1}\bot$ (which exists by "normality"). 
  Define $u_i$ by induction: 
  \begin{itemize}
    \item If $p_i\runx{u_{i-1}} \bot$ or $p_i\runx{u_{i-1}} p$ then $u_i = u_{i-1}$.
    \item If $p_i\runx{u_{i-1}} p' \neq p$, then we take the $v_{p'}$ chosen
    above and define $u_i = u_{i-1}v_{p'}$.
  \end{itemize}
  It is easy to check that $u_k$ is a "central sequence" for $p$. 

Assume now that $p$ is not $\fleqx$-maximal. Let $p'$ be a $\fleqx$-maximal state in $P$ with $p\fleqx p'$. 
As the automaton is  "centralised", $p$ and $p'$ are in the same SCC of $\T_x$. Let $u,v$ be words such that $p\runx{u} p' \runx{v} p$.
By the previous case, $p'$ admits a central sequence $z_{p'}$.
It is easy to check that $uz_{p'}v$ is a "central sequence" for $p$.
\end{proof}

\subsection{Uniqueness of minimal layered automata}
We are ready to prove Theorem~\ref{thm:concise-are-unique}.
\conciseAreUnique*

Let $\Aa$ and $\Bb$ be two "consistent" "layered automata" for $L$ that are "normal", "central" and "safe minimal". We should construct an "isomorphism@@lay" $\f:\Aa \rightarrow \Bb$. Such an "isomorphism" can be equivalently written as a sequence of "isomorphisms of transition systems" $\f_x:\T_x^\Aa \to \T_x^\Bb$ compatible with the $\mu_x$-functions.
We build this sequence inductively on $x$.

The first layers $\T_1^\Aa$ and $\T_1^\Bb$ both consist of the residual classes of $L$, and therefore there is an "isomorphism@@TS" $\f_1:\T_1^\Aa \to \T_1^\Bb$. Now assume that we have constructed "isomorphisms@@TS" $\f_1,\dots,\f_{x-1}$ up to layer $x-1$ and want to extend it to layer $x$. First, we show that for each "$x$-state" $p$ of $\Aa$ with parent $\mu^\Aa_{x-1}(p)=r$, there is an "$x$-state" $q$ of $\Bb$ with parent $\f_{x-1}(r)$, such that $q$ and $p$ have the same "safe language" (Claim~\ref{claim:states-with-same-safeL}).
By "safe minimality", such a state $q$ is unique.
We then define $\f_x(p) = q$.
The roles of the two automata are symmetric in the proof, so we get that both automata have states for the same "safe languages" on layer $x$. This implies that $\f_x$ is an "ismorphism@@TS" between $\T_x^\Aa$ and $\T_x^\Bb$ (Claim~\ref{claim:safe-induces-isomorphism}). 
Moreover, by definition, $\f_x$ is compatible with the $\mu_x$-functions: $\f_{x-1}(\mu^\Aa_{x-1}(p))= \mu^\Bb_{x-1}(\f_x(p))$.
This concludes the proof of Theorem~\ref{thm:concise-are-unique}.

\begin{claim}\label{claim:states-with-same-safeL}
 For each "$x$-state" $p$ of $\Aa$ with parent $\mu^\Aa_{x-1}(p)=r$, there is a unique "$x$-state" $q$ of $\Bb$ with parent $\f_{x-1}(r)$ such that $\Lsafex^{\Aa}(p) = \Lsafex^{\Bb}(q)$.
\end{claim}
\begin{claimproof}
  Unicity of $q$ directly follows from "safe minimality". We show existence via the following construction, which is illustrated in Figure~\ref{fig:minimality-proof}. 
  Let $p$ be a "$x$-state" of $\Aa$ with parent $r$, and let $\f_{x-1}(r) = s$ be the state on layer $x-1$ of $\Bb$ that $r$ is mapped to. Let $z_p$ be a "central sequence" for $p$. 
  Then $z_p^\omega$ generates minimal "priority" $x$ when read in $\sem{\Aa}$
  starting from any state $p'\sim_x p$. 
  Moreover, we have that $z_p$ loops on $r$ since it
  loops on $p$, and thus, $z_p$ also loops on $s$ in $\Bb$ because layers $x-1$
  of $\Aa$ and $\Bb$ are isomorphic. 
  Therefore, $z_p^\omega$ generates minimal  "priority" at least $x-1$ when read from $s$ in $\Bb$. 
  Since $r$ and $s$ are language equivalent, there needs to be some child $q'$ of $s$ such that
  $z_p^\omega$ is "$x$-safe" from $q'$. Let $q$ be such that $q' \runx{z_p} q$.
  Note that $q$ also has parent $s$ because $z_p$ loops on $s$.  

  Let $z_q$ be a "central sequence" for $q$, and $w$ be such that $q \runx{w} q'$ (which exists because $\Bb$ is "normal"). Note that $z_q$ and $w$ loop on $s$ and hence also on $r$. The $\omega$-word $(z_qwz_p)^\omega$ produces "priority" $x$ when read starting from $q$ in $\Bb$. Since $p$ and $q$ are language equivalent (because their parents are), it has to be possible in $\Aa$ to produce a "priority" $\ge x$ on $(z_qwz_p)^\omega$ from a child of $r$. Since $z_p$ is "central@@seq" for $p$, such a run is in $p$ after the first $z_qwz_p$, and hence $(z_qwz_p)^\omega$ must be safe from $p$. Let $p_1,p_2$ such that $p \xrightarrow{z_q} p_1 \xrightarrow{w} p_2$. Since $z_p$ is central for $p$, we get that  $p_2 \xrightarrow{z_p:x} p$.

We claim that $\LsafexxA{x}{\Aa}(p) = \LsafexxA{x}{\Bb}(q)$. Assume the contrary.
The first case is that there is $u \in \Lsafex(p) \setminus \Lsafex(q)$. Since $\Aa$ is "normal", we can choose $u$ such that it loops on $p$. Hence it also loops on $r$ and $s$. Then the $\omega$-word $(z_qwz_pu)^\omega$ generates "priority" $x$ from $p$ and "priority" $x-1$ from $q$. The latter holds because either "priority" $x-1$ is generated on $z_q$, and if not, then the state after $z_q$ is $q$ and hence also after $z_qwz_p$, and thus "priority" $x-1$ is produced on $u$. Since $p$ and $q$ are language equivalent, this is a contradiction.

The second case is that there is $v \in \Lsafex(q) \setminus \Lsafex(p)$. Then we get that $(z_qwz_pv)^\omega$ generates "priority" $x$ from $q$ and "priority" $x-1$ from $p$ using the same type of argument as in the first case.
\end{claimproof}

\begin{figure}[h]
  \includegraphics[width=1\textwidth]{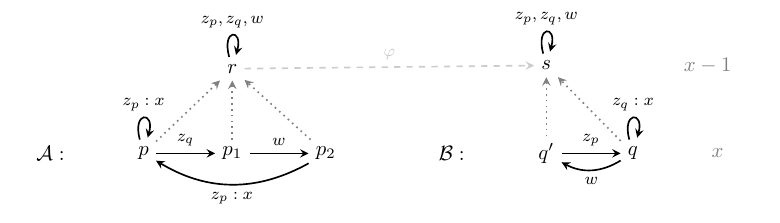}
  \caption{Illustration for the construction of the isomorphism in the proof of Claim~\ref{claim:states-with-same-safeL}.\label{fig:minimality-proof}}
\end{figure}

We let $\f_x(p) =q$, for $q$ the unique "$x$-state" in $\Bb$ satisfying $\f_{x-1}(\mu^\Aa_{x-1}(p))= \mu^\Bb_{x-1}(q)$ and $\Lsafex^{\Aa}(p) = \Lsafex^{\Bb}(q)$.

\begin{claim}\label{claim:safe-induces-isomorphism}
 The function $\f_x$ is an "isomorphism of transition systems". 
\end{claim}
\begin{claimproof}
  We first show that $\f_x$ is a "morphism@@TS".
  Let $p\actx{a}p'$ in $\Aa$ and $q = \f_x(p)$.
  Since $\Lsafex^{\Aa}(p) = \Lsafex^{\Bb}(q)$, there is a transition $q\actx{a}q'$ 
  in $\Bb$, and $\Lsafex^{\Aa}(p') = \Lsafex^{\Bb}(q')$. 
  If $r$ is the parent of $p$ we have $r\act{a}r'$ for some $r'$ in $\Aa$.
  By definition of $\f_x(p)$, the parent of $q=\f_x(p)$ is $\f_{x-1}(r)$.
  By isomorphism, $\f_{x-1}(r)\act{a}\f_{x-1}(r')$ thus $\f_{x-1}(r')$ is the
  parent of $q'$.
  So by definition, we have $\f_x(p')=q'$.
  Thus, $\f_x$ is a morphism, as desired.

  By Claim~\ref{claim:states-with-same-safeL}, $\f_x$ is a bijection, and the function $\f_x^{-1}$ is also a "morphism@@TS" by symmetry.
\end{claimproof}

\section{Minimisation in polynomial time}\label{sec:minimisation}

This section presents a minimisation procedure for "consistent" "layered automata". 
More precisely, the procedure constructs the unique minimal "layered automaton" appearing in Theorem~\ref{thm:concise-are-unique}.
The procedure establishes that every "consistent" "layered automaton" contains a "subautomaton" admiting a surjective morphism into this minimal automaton, thereby proving canonicity. 
We formalise these notions through "morphisms of layered automata", which we
introduce in Section~\ref{subsec:morphisms}.

\begin{theorem}\label{thm:minimisation}
    Let $\Aa$ be a "consistent" "layered automaton".
    We can build in polynomial time an equivalent "consistent" "layered
    automaton" $\Aa_L$ that is "normal", "central" and "safe minimal".
    
    Moreover, $\Aa$ contains a "subautomaton"\footnote{We note that this subautomaton may not be "consistent".} that admits a surjective "strong morphism" to $\Aa_L$.
    In particular, $\Aa_L$ has no more states nor "leaf states" than $\Aa$.
\end{theorem}

Together with Theorem~\ref{thm:concise-are-unique}, this gives us minimality and canonicity of "normal", "central" and "safe minimal" "consistent" "layered automata".
\begin{corollary}\label{cor:canonicity}
  Let $\Aa_L$ be a "consistent" "layered automaton" that is "normal", "central" and "safe minimal".
  Then, any equivalent "consistent" "layered automaton" $\Aa$ contains a "subautomaton" that admits a surjective "strong morphism" to $\Aa_L$.
  Therefore, $\Aa$ has at least as many states as $\Aa_L$, and $\sem{\Aa}$ has at least as many states as $\sem{\Aa_L}$.
\end{corollary}

% In the rest of this section, we will denote by $n$ the total number of states of a "layered automaton", that is, $n = |Q_{\geq 1}|$, and by $m$ the total number of transitions. Note that $m \leq n\cdot |\Sigma|$.
% Also, $n$ is at most $d$ times the number of "leaf states".

The rest of the section is devoted to the proof of Theorem~\ref{thm:minimisation}.
We start by introducing "morphisms@@lay" and proving some technical lemmas that are used to show correctness of the minimisation procedure.

\subsection{Morphisms of layered automata}\label{subsec:morphisms}

\AP A ""morphism@@layered"" between two layered automata $\Aa$ and $\Aa'$ is a function $\varphi\colon Q_{\geq 1} \to Q'_{\geq 1}$ that preserves the structure of the "transition systems" and the "ancestor" relation in the tree representation. That is:
\begin{itemize}\setlength\itemsep{2mm}
  \item $\varphi(q_\init) = q_\init'$,
  \item if $q\actx{a}q'$ in $\T_x$, then $\varphi(q)\actxx{a}{y}\varphi(q')$ in $\T_y'$, for $y = \llayer(\varphi(q))$, and
  \item if $p$ is an "ancestor" of $q$ in $\Aa$, then $\f(p)$ is an "ancestor" of $\f(q)$ in $\Aa'$.
\end{itemize}
\AP We say that it is a ""strong morphism"" if moreover it preserves non-transitions, that is:
\begin{itemize}\setlength\itemsep{2mm}
  \item if $q\actx{a}\bot$, then $\varphi(q)\actxx{a}{y} \bot$.
  \end{itemize}
Note that, while states may change layer via $\varphi$, some rigidity is imposed
by the conditions. Notably, if $q\in Q_x$ is sent to $\T_y'$, then all states
reachable from $q$ in $\T_x$ are also sent to the layer $y$.
In particular, due to our reachability assumption, all
states in $\T_1$ are sent
to $\T_1'$.

\AP A "morphism of layered automata" $\varphi$ is an ""isomorphism"" if it is bijective and $\varphi^{-1}$ is also a "morphism@@layered".
Note that in this case %as all states in $\T_1$ are reachable, 
it preserves layers, that is, $\restr{\varphi}{Q_x}$ is an "isomorphism of transition systems" from $\T_x$ to $\T_x'$.

% \AP A ""layer-preserving morphism"" between two "layered automata" $\Aa$ and $\Aa'$ is a function $\varphi\colon Q_{\geq 1} \to Q'_{\geq 1}$ such that:
% \begin{itemize}
%   \item $\varphi(q_\init) = q_\init'$,
%   \item $\restr{\varphi}{Q_x}$ is a "morphism of transition systems" from $Q_x$ to $Q_x'$, and
%   \item   for $q\in Q_{x+1}$, $\mu_x(\varphi(q)) = \varphi(\mu_x(q))$.
% \end{itemize}

%\AP It is a ""layer-preserving isomorphism"" if it is bijective.
%Note that in this case $\varphi^{-1}$ is also a "morphism@@layered" (by "determinism" of the $\T_x$).

\subparagraph{Subautomata.} \AP A "layered automaton" $\Aa'$ is a ""subautomaton"" of $\Aa$ if there exists an injective "morphism@@lay" $\f \colon \Aa' \to \Aa$. 

Note that $\Aa'$ may be obtained not just by removing states from some layers in
$\Aa$, but also by changing some SCCs from one layer to another.

%We first prove some lemmas that will come in handy to show the correctness of our constructions.
\subparagraph{Morphisms preserving strong acceptance.}
In the upcoming sections, we start with a "consistent" "layered automaton" $\Aa$ and tweak it to obtain a modified automaton $\Aa'$. 
In all except one modification, we obtain $\Aa'$ by \emph{removing} states or transitions of $\Aa$ in such a way that there is an injection $\Aa' \hookrightarrow \Aa$.
In two out of three cases the correctness of this operation is ensured by
Lemma~\ref{lem:consistent-morphism} below. 
In the "safe minimisation" step (Section~\ref{subsec:safe-minimisation}), $\Aa'$ is obtained by \emph{merging} states of $\Aa$, which leads to a surjection $\Aa
\twoheadrightarrow \Aa'$.

\AP We say that a "morphism of layered automata" $\f:\Aa'\to \Aa$ ""preserves strong acceptance"" if for every  $p'\in Q'_{\geq 1}$ and every $w\in\S^\w$ we have:
  \begin{itemize}
    \item if $w$ is "strongly accepted" by $p'$ then $w$ is "strongly accepted" by $\f(p')$,
    \item if $w$ is "strongly rejected" by $p'$ then $w$ is "strongly rejected" by $\f(p')$.
  \end{itemize}

%Note that such a morphism may not be layer-preserving. 

\begin{lemma}\label{lem:consistent-morphism}
  Let $\f:\Aa' \to \Aa$ be a "morphism that preserves strong acceptance". 
  If $\Aa$ is "consistent" then $\Aa'$ is "consistent" and $L(\Aa')=L(\Aa)$.
\end{lemma}
\begin{proof}
  Take $p_1'\eqxx{1} p_2'$ in $\Aa'$.
  Towards a contradiction suppose $w$ is "strongly accepted" by $p_1'$ and "strongly rejected" by $p_2'$.
  By the "morphism@@lay" property we have $\f(p_1)\eqxx{1}\f(p_2)$.
  Then, since $\f$ "preserves strong acceptance",  $w$ is "strongly accepted" by $\f(p_1')$ and "strongly rejected" by $\f(p_2')$.
  A contradiction with consistency of $\Aa$.

  If $w$ is "accepted@@layered" by $\Aa'$ then by
  Lemma~\ref{lem:layered-acceptance-strong} there is a state $p$ together with a
  decomposition $uw'=w$ such that $q'_\init\act{u}\wmu_1(p)$ and $w'$ strongly
  accepted from $p$.
  Then in $\Aa$ we have $q_\init\act{u}\wmu_1(\f(p))$ and $w'$ is strongly
  accepted from $\f(p)$. 
  By consistency of $\Aa'$ and Lemma~\ref{lem:layered-acceptance-strong}, $w$ is
  "accepted@@layered" by $\Aa'$.
  The argument when $w$ is "rejected@@layered" by $\Aa'$ is analogous.
\end{proof}

\subsection{Normalisation}
Let $\Aa$ be a "consistent" "layered automaton". We construct an equivalent
automaton satisfying both "N1" and "N2" in two steps.
First, we ensure "N1" by separating any SCCs that have transitions between them.
Then, we establish "N2" without affecting "N1".

\subsubsection{Separating SCCs}

In order to ensure property "N1" we separate every $\T_x$ into disconnected, non-trivial SCCs.
Let us fix $x>1$ and suppose that in $\T_x$ there is a transition $p\act{a}_x q$ of $\T_x$, such that the two states are not in the same SCC.
(Such a transition can be detected in linear time performing a SCC decomposition~\cite{Tarjan72DepthFirst}.)
The operation is simply to remove the transition $p\act{a}_x q$ from $\T_x$ and all
transitions mapping to it via $\wmu_x$ in $\T_{y}$ for $y>x$.
We also remove trivial SCCs that may have been generated (states in $\T_y$ with
no outgoing transition) until there are no such states in the result.
\AP Let us denote the resulting automaton $\intro*\sepA$.
\begin{lemma}
  Let $\Aa$ be a "consistent" "layered automaton".
  Then, $\sepA$ is "consistent". Moreover, $L(\sepA)=L(\Aa)$, and  $\sepA$ is a "subautomaton" of $\Aa$.
\end{lemma}
\begin{proof}
    We show that the mapping $\f\colon \sepA \to \Aa$ given by $\f(p)=p$ is an injective "morphism" that "preserves strong acceptance". The result then follows by Lemma~\ref{lem:consistent-morphism}.    
    The fact that $\f$ is a "morphism" and injective is straightforward.

    Let $w$ be "strongly accepted" by $p$ in $\sepA$ (the case "strongly rejected" is symmetric). In particular, $w\in \LsafexxA{x}{\sepA}(p)$, so  $w\in \LsafexxA{x}{\Aa}(p)$ by the "morphism" property.
    Assume by contradiction that $w$ is not "strongly accepted" by $p$ in $\Aa$, that is, there is a decomposition $w = uw'$  and state $q\in Q_{x+1}$ with $p\runx{u}\wmu_x(q)$ in $\Aa$ and $w'\in \LsafexxA{x+1}{\Aa}(q)$.
    Let $w' = u'w''$ such that the run $q\runxx{u'}{x+1} q' \runxx{w''}{x+1}$ does not change of SCC during the suffix $w''$.
    Then, $w''\in \LsafexxA{x+1}{\sepA}(q')$.
    Therefore, the decomposition $w = uu'\cdot w''$ witnesses that $w$ is not "strongly accepted" by $p$ in $\sepA$, a contradiction.
\end{proof}

After applying this operation we remove at least one transition. 
So, by repeating this operation at most $|Q_{\geq 1}|^2$ times we obtain an equivalent "consistent" "subautomaton" satisfying~"N1".

\subsubsection{Lowering}
We suppose that $\Aa$ satisfies the property "N1", and we make it satisfy "N2"
without affecting "N1".
Let us suppose that "N2" does not hold for a state $q$ at level $x$.
This means that $q$ has a child $p$ such that for every $u\in \S^*$ such that $q\runxx{u}{x}$ is defined, we also have $p\runxx{u}{x+1}$ defined.
(We can detect in $\LOGSPACE$ if there is such a pair of parent and child, as this simply corresponds to equivalence of deterministic safety automata.)
Take $S_{x+1}$, the SCC of $p$ in $\T_{x+1}$, as well as $S_x$, the SCC of $q=\mu_x(p)$ in $\T_x$.
We have that  $S_x$ is ""covered"" by $S_{x+1}$, meaning that
for every $p'\in S_{x+1}$ whenever there is a transition $\m_x(p')\act{a}_x$ then we also have $p'\act{a}_{x+1}$.
In $\sem{\Aa}$ this means that states mapping to $S_{x+1}$ do not have outgoing
transitions of "priority" $x$; if $\wmu_{x+1}(p')\in S_{x+1}$ then for no $a$ there is a
transition $p'\act{a:x}$ in $\sem{\Aa}$.
We show how to eliminate such a covered SCC $S_x$ at level $x$.

Here is a small lemma saying that SCCs at higher levels are included in SCCs of
lower levels.
\begin{lemma}\label{lem:SCC-is-in-SCC}
  Let $S\subseteq Q_x$ be an SCC in $\T_x$, and $S'\subseteq {Q_{x+1}}$ be an SCC in
  $\T_{x+1}$.
  If $\m_x(S')\cap S\not=\es$ then $\m_x(S')\incl S$.
\end{lemma}
\begin{proof}
  This is just because $\mu_x$ preserves paths.
\end{proof}

Suppose $\Aa$ satisfies "N1" and fix $S_x$, $S_{x+1}$ such that $S_{x}$ is
"covered" by $S_{x+1}$ as described above.
For every $y\geq x+2$ let $S_y$ be the part of $\T_y$ induced by states of
$\T_y$ mapping to $S_{x+1}$, namely by $\set{p\in \T_y: \wmu_{x+1}(p)\in S_{x+1}}$.
Similarly, for $y \ge x+1$, let $Z_y$ be the part of $\T_y$ mapping to $S_x$, namely induced by
states $\set{p\in Q_{\geq x}: \wmu_x(p)\in S_{x}}$.
We need to make a difference between states mapping to $S_x$ and to $S_{x+1}$,
as even though $S_x$ is covered by $S_{x+1}$ there may be states outside
$S_{x+1}$ mapping to $S_x$. 
Observe that by Lemma~\ref{lem:SCC-is-in-SCC}, every $S_y$ and $Z_y$ is a
set of SCCs in $\T_y$. 
\begin{example}\label{ex:N2}
Before giving the general definition of the lowering operation, consider the example automaton shown on the left-hand side in Figure~\ref{fig:N2}. There are several candidate SCCs at which the lowering operation can be applied. For this example, we pick $S_x = \{p_2\}$ and $S_{x+1} = \{p_3\}$, so $x=2$ in this case. Note that both states loop on $a,b$, so $S_x$ is covered by $S_{x+1}$. The collections of SCCs that are mapped to $S_{x+1}$ by $\wmu_3$ are $S_4 = \{p_4,q_4\}$ and $S_5 = \{q_5\}$, and the collections of SCCs that are mapped to $S_{x}$ by $\wmu_2$ are $Z_3 = \{p_3,r_3\}$, $Z_4 = \{p_4,q_4,r_4\}$, and $Z_5 = \{q_5,r_5\}$. The automaton $\lowA$ is obtained by removing $Z_y$ from layer $y$, and adding $S_{y+2}$ to layer $y$ for every $y \ge x$. The result is shown on the right-hand side of Figure~\ref{fig:N2}. The parent relation is preserved for all states except for the ones in $S_{x+2}$, which is $S_4$ in this case. The new parents of $p_4$ and $q_5$ are now $\wmu_{x-1}(p_4)$ and $\wmu_{x-1}(q_5)$, which are both $p_1$. This corresponds to the third case in the definition of $\mu^{\mathit{low}}_{y-1}$ below. The parent of $q_5$ is still $q_4$, but since the states changed the layers, the definition of $\mu^{\mathit{low}}_{y-1}$ needs to shift indices compared to $\mu$. So we get $\mu^{\mathit{low}}_2(q_5) = \mu_4(q_5) = q_4$. This is the second case in the definition of $\mu^{\mathit{low}}_{y-1}$ below for $y=3$. The parent of states that did not change layer and were not deleted (those that are in $Q^{\mathit{low}}_y$ and in $Q_y$) is just copied, as for $q_2$ in this example: $\mu^{\mathit{low}}_1(q_2) = \mu_1(q_2) = p_1$.  This is the first case in the definition of $\mu^{\mathit{low}}_{y-1}$ below.
\end{example}
\begin{figure}
\def\xd{.7cm}  % horizontal distance between nodes
\def\yd{-.9cm} % vertical distance between layers
\begin{center}
    \includegraphics[width=0.9\textwidth]{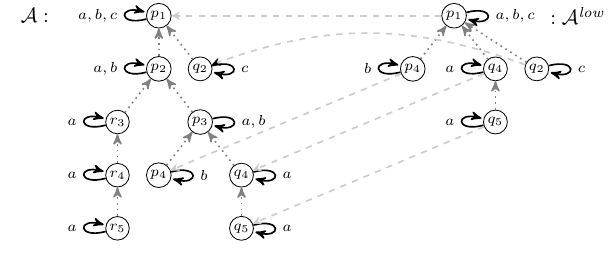}
    \caption{Illustration of the lowering operation for establishing "N2" from Example~\ref{ex:N2}. \label{fig:N2}}
    \end{center}
\end{figure}

\AP A new layered automaton $\intro*\lowA$ is obtained by lowering some
components.
\begin{itemize}
  \item $\Tlow_{y}=
  \begin{cases}
    (\T_{y}-Z_y)\cup S_{y+2} & y\geq x,\\
    \T_y & \text{otherwise.}
  \end{cases}$
  \item For $p\in\Tlow_{y}$, $y\geq 2$ we set
  $\mu^{\mathit{low}}_{y-1}(p)=\begin{cases}
    \mu_{y-1}(p) & \text{if $p\in Q_{y}$},\\
    \mu_{y+1}(p) & \text{if $p\in S_{y+2}$, $y>x$},\\
    \wmu_{x-1}(p) & \text{if $p\in S_{x+2}$}.
  \end{cases}$ 
\end{itemize}
That is, the SCCs in $\T_y$ that have $S_x$ as ancestor (that are mapped to
$S_x$ by $\wmu_x$) are overwritten by the ones in $\T_{y+2}$. 
So $S_x$ is replaced by $S_{x+2}$, while not only $S_{x+1}$ but all
the SCCs on layer $x+1$ mapping to $S_{x}$ are replaced with $S_{x+3}$.
In terms of $\sem{\lowA}$ this means that "priorities" of transitions coming
from states mapping to $S_{x}$ are lowered by $2$, and some states are removed.
 \begin{lemma}
  Let $\Aa$ be a "consistent" "layered automaton" satisfying "N1".
  Then, $\lowA$ is "consistent", and satisfies "N1". Moreover, $L(\lowA)=L(\Aa)$, 
  and $\lowA$ is a "subautomaton" of $\Aa$.
\end{lemma}
\begin{proof}
The automaton $\lowA$ satisfies "N1", as no transitions between SCCs are added.

To check that $\lowA$ is a "subautomaton" of $\Aa$ we verify that the identity mapping $\f\colon \lowA \to \Aa$ is an
injective "morphism@@lay". 
Clearly if $p\act{a}_yq$ in $\lowA$ then $\f(p)\act{a}_{z}\f(q)$ in $\Aa$ where
$z=y+2$ or $z=y$ depending on whether $p\in S_{y+2}$ or not. 
The property of preserving the "ancestor relation" follows from the definition of $\mlow_{y-1}$.
This is because $\Aa$ satisfies "N1", and $\lowA$ is obtained by removing whole SCCs.
Observe that if $p\in S_{x+2}$ then $p$ is on layer $x$ of $\Aal$, on layer
$x+2$ of $\Aa$, and $\f(\m_{x-1}(p))$ is on layer $x-1$ of $\Aa$.
So $\f$ preserves the "ancestor relation" but not the parent relation. 

We claim that if $w$ is "strongly accepted" from $p$ on layer $y$ in $\Aal$ then
there is a decomposition $uw'=w$ and a state $q$ such that
$p\act{u}_z\wmu_z(q)$ and $w'$ is strongly 
accepted from $q$ in $\Aa$; here $z=y+2$ or $z=y$ depending on whether $p\in
S_{y}$ or not.
We also have the analogous statement with the roles of accepted/rejected interchanged.

To show the claim take $w$ that is "strongly accepted" or "strongly rejected" from $p$ on
layer $y$ of $\lowA$.
If $p\in S_{y+2}$ then the whole run on $w$ from $p$ is in $S_{y+2}$ because $S_{y+2}$
is a union of SCCs by Lemma~\ref{lem:SCC-is-in-SCC}. 
But then all states of layer $y+1$ of $\Aal$ mapping to $S_{y+2}$ are in
$S_{y+3}$. 
So strong acceptance/rejection of $w$ from $p$ also holds in $\Aa$.
Thus, the claim is true for the trivial decomposition $u=\e$, $w'=w$ and $q=p$.
If $p\in \T_{y}$ for $y\not=x-1$ then the argument is the same because the run
on $w$ stays in $\T_y$.

The remaining case to complete the proof of the claim is when $p\in \Tlow_{x-1}$.
For the ease of writing assume $x-1$ is even, so $w$ is strongly accepted from
$p$ in $\Aal$. The other case is the same.
Suppose $w$ is not strongly accepted from $p$ in $\Aa$.
This means that there is a decomposition $uw'=w$ and a state $q$ on layer $x$
such that $p\act{u}_{x-1}\m_{x-1}(q)$ and $w'\in L_x(q)$ in $\Aa$.
Clearly, $q\in S_x$, as otherwise $w$ would not be strongly accepted from $p$ in
$\Aal$. 
Since $S_x$ is covered by $S_{x+1}$ there is $q'$ in $S_{x+1}$ with
$\m_x(q')=q$, $w'\in L_{x+1}(q')$.
If $w'$ is strongly accepted from $q'$ then we have proved our claim as
$p\act{u}_{x-1}\wmu_{x-1}(q')$.
Otherwise, we have a decomposition $u'w''=w'$ and a state $q''$ on layer $x+2$
such that $q'\act{u'}_{x+1}\m_{x+1}(q'')$ and $w''\in L_{x+2}(q'')$.
But then in $\lowA$ we have $p\act{u}_{x-1}\m_{x-1}(q)\act{u'}\wmu_{x-1}(q'')$,
because $x-1$ layers of $\Aa$ and $\lowA$ are the same; and $w''\in
L_x(\mlow_x(q''))$, because $q''\in S_{x+2}$.
Moreover, $\wmu_{x-1}(q'')=\mlow_{x-1}(q'')$ by definition.
This is a  contradiction with $w$ being strongly accepted from $p$ in $\Aal$.

Equipped with our claim we can show that $\lowA$ is "consistent".
For this suppose to the contrary that there are $p_1\sim p_2$ and $w$ such that
$w$ is strongly accepted from $p_1$ and strongly rejected from $p_2$ in $\lowA$.
By our claim there are states $q_1$, $q_2$ and decompositions $u_1w_1=u_2w_2=w$
such that $p_1\act{u_1}_1\wmu_{y_2}(q_1)$, $p_2\act{u_2}_1\wmu_{y_2}(q_2)$, and
$w_1$ strongly accepted from $q_1$ in $\Aa$ and $w_2$ strongly rejected from $q_2$ in
$\Aa$.
As $\Aa$ is consistent, by Corollary~\ref{cor:equivalences-acceptance} this
would mean that $u_1w_1=w$ is accepted by $\Aa$, and $u_2w_2=w$ is rejected by
$\Aa$. 
A contradiction.

Finally, we show that $L(\Aa)=L(\lowA)$. 
If $w\in L(\lowA)$ then by Corollary~\ref{cor:equivalences-acceptance} there
is a decomposition $uw'=w$ and a state $p$ such that $q_\init\act{u}_1\wmu_1(p)$
and $w'$ is "strongly accepted" from $p$ in $\lowA$.
From our claim it follows that there is a decomposition $u'w''=w'$ and a state
$q$ in $\Aa$ such that $q_\init\act{uu'}_1\wmu_1(q)$ and $w''$ is "strongly
accepted" from $q$ in $\Aa$.
Thus, by Corollary~\ref{cor:equivalences-acceptance}, $uu'w''=w$ is accepted
by $\Aa$.
The case when $w\not\in L(\lowA)$ is the same.
\end{proof}

After each lowering operation we remove entire SCCs, because of Lemma~\ref{lem:SCC-is-in-SCC},
so "N1" property still holds. 
When there are no "covered" SCCs we obtain an equivalent "subautomaton" that in
"normal form", namely, satisfying both "N1" and "N2".
Finding a covered SCC and removing it can clearly be done in PTIME.

\subsection{Safe minimisation}\label{subsec:safe-minimisation}
We show how we can "safe minimise" a "layered automaton" while maintaining "normalisation".
Recall that $p\seq_x q$ if $p$ and $q$ are language equivalent and $\Lsafexx{y}(p) = \Lsafexx{y}(q)$ for all $y\leq x$, and that an automaton is "safe minimal" if $\eqx$ and $\seq_x$ coincide.
First, we obtain directly from the definition that the $\seq_x$ relation is a bisimulation in $T_x$.
\begin{lemma}\label{lem:seq-is-bisimulaton}
  Let $p,p'\in Q_x$ such that $p\seq_x p'$. If $p\act{a}_x q$ then $p'\actx{a}
  q'$ and $q'\seq_x q$ .
\end{lemma}

The desired safe minimal automaton $\safemin{\Aa}$ is the automaton obtained by using $\seq_x$ equivalence classes as $x$-states.
In the following, when using the notation $\scl{q}{x}$, we assume that $q\in Q_x$.
\AP Formally, $\intro*\safeminA$ is given by:
\begin{itemize}\setlength\itemsep{2mm}
\item $\safemin{\T_x} = (\safemin{Q_x},\safemin{\delta_x})$ with 
 \begin{itemize}\setlength\itemsep{2mm}
 \item $\safemin{Q_x} = \{\scl{q}{x} \mid q \in Q_x\}$,
 \item$\safemin{\delta_x}([q]_{\seq_x},a) = 
   \begin{cases}
     [\delta_x(q,a)]_{\seq_x} & \text{if } \delta_x(q,a) \text{ defined},\\[1.5mm]
     \bot & \text{else}
   \end{cases}$ %for $q \in Q_x$.
 \end{itemize}
\item The initial state is $\scl{q_0}{1}$, for $q_0$ the initial state in $\T_1$.
\item $\safemin{\m_x}([q]_{\seq_{x+1}}) = [\m_x(q)]_{\seq_x}$.%for $q \in Q_{x+1}$. 
\end{itemize}

The transition function $\safemin{\delta}$ is well-defined by Lemma~\ref{lem:seq-is-bisimulaton}, and $\safemin{\m_x}$ is well defined because $q \seq_{x+1} q'$ implies that $q \seq_{x} q'$ by definition, and hence also $\mu_x(q) \seq_{x} \mu_x(q')$.
By a direct application of the definitions and the fact that $\mu_x$ is a morphism, one easily shows that $\safemin{\m_x}$ is a "morphism".
Lemma~\ref{lem:seq-is-bisimulaton} implies that
$\Lsafex^\Aa(p)=\Lsafex^{\safeminA}(\scl{p}{x})$.
Thus, $\Aar$ is safe minimal.
Moreover, the other properties we have already established are not invalidated
by this construction, as shown in the next lemma.

\begin{lemma}\label{lem:safe-minimisation-is-safe}
  If $\Aa$ is "normal", then so is $\Aar$.
\end{lemma}
\begin{proof}
  We check the properties "N1", and "N2", for $x\geq 2$.
  \begin{description}
    \item["N1".]  Let $\scl{q}{x}\runx{u} \scl{q'}{x}$ in $\safemin{\T_x}$. Then, we can assume that the representatives $q,q'$ are chosen such that $q\runx{u} q'$ in $\T_x$. Therefore $q'\runx{v} q$ for some $v \in \Sigma^+$, as $\Aa$ satisfies "N1", and hence $\scl{q'}{x}\runx{v} \scl{q}{x}$.
    \item["N2".] Let $\scl{q}{x}$ be a state in $\safeminA$, and $\scl{p}{x+1}$ a child of it.
    Then, we can assume that the representatives are chosen such that $\mu_x(p) = q$ in $\Aa$.    
    As $\Aa$ satisfies "N2", there is $u$ such that $q\runxx{u}{x}$ is defined in $\T_{x}$ but $p\runxx{u}{x+1}\bot$ in $\T_{x+1}$.
    Then, $\scl{q}{x} \runxx{u}{x}$ is defined in $\safemin{\T_{x}}$ but $\scl{p}{x}\runxx{u}{x+1}\bot$ in~$\safemin{\T_{x+1}}$.  
    \qedhere
  \end{description}
\end{proof}

 The following lemma shows the key property to obtain correctness of $\safeminA$. 

\begin{lemma}\label{lem:safeMin-strongAcceptance}
  Let $\Aa$ be a "consistent" "layered automaton" and $p\in Q_x$ an "$x$-state".
  If a word $w\in \S^\omega$ is "strongly accepted" (resp. "rejected@@strong") by $\scl{p}{x}$ in $\safeminA$, then it is "strongly accepted" (resp. "rejected@@strong") by $p$ in $\Aa$.
\end{lemma}
\begin{proof}
  Assume that $w$ is "strongly accepted" by $\scl{p}{x}$ in $\safeminA$(so $x$ is even).
  In particular, $w\in \LsafexxA{x}{\safeminA}(\scl{p}{x})$, and therefore $w\in \LsafexxA{x}{\Aa}(p)$.
  Assume by contradiction that $w$ is not "strongly accepted" by  $p$. Then, there is a decomposition $w=uw'$ and a state $q'\in Q_{x+1}$ such that
  $p\actx{u}q$, where $q=\wmu_x(q')$, and $w'\in \LsafexxA{x+1}{\Aa}(q')$.
  We claim that the decomposition $w=uw'$ witnesses that $w$ is not "strongly accepted" by $\scl{p}{x}$ in $\safeminA$, a contradiction.
  Indeed, by definition of $\safemin{\delta_x}$ we have that $\scl{p}{x} \actx{u} \scl{q}{x}$,  and also $w'\in \LsafexxA{x+1}{\safeminA}(\scl{q'}{x+1})$.
  As $\safemin{\m_x}([q']_{\seq_{x+1}}) = \scl{q}{x}$, this concludes the proof.
\end{proof}

We are ready to prove the correctness of the construction of $\Aar$.
\begin{lemma}
  Let $\Aa$ be a "normal" and "consistent" "layered automaton".
  Then, $\safeminA$ is "safe minimal", "normal", and "consistent" "layered automaton".
  Moreover, $L(\Aa)=L(\safemin{\Aa})$, and $\Aa$ admits a surjective "strong
  morphism" to $\safeminA$. 
\end{lemma}
\begin{proof}
  By Lemma~\ref{lem:seq-is-bisimulaton},  the mapping $\f\colon \Aa \to
  \safeminA$ with $\f(p) = [p]_{\seq_x}$ for $q \in Q_x$ is a "strong
  morphism@@lay", which is clearly surjective. 
  Automaton $\safeminA$ is "normal" thanks to Lemma~\ref{lem:safe-minimisation-is-safe}.

  We show "consistency" of $\safeminA$. Let $\scl{p}{x}$ and $\scl{p'}{x'}$ be $\eqxx{1}$-equivalent in $\safeminA$, that is, $L(\Aa,p) = L(\Aa,p')$.
  Assume by contradiction that a word $w$ is "strongly accepted" by $\scl{p}{x}$ and "strongly rejected" by $\scl{p'}{x'}$ in $\safeminA$.
  By Lemma~\ref{lem:safeMin-strongAcceptance}, $w$ is "strongly accepted" by $p$ and "strongly rejected" by $p'$ in $\Aa$.
  Therefore, by Lemma~\ref{lem:language-p-strong-acc}, $w\in L(\Aa,p)$ and $w\notin L(\Aa,p')$,
  contradicting language equivalence of these states.

  We prove that $L(\safeminA) = L(\Aa)$.
  Let $w$ be a word "accepted@@layered" by $\safeminA$ (similar proof for $w$ "rejected@@layered"), that is, %Lemma~\ref{lem:layered-acceptance-strong}
   there is $x$ even and a decomposition $w=uw'$ such that $\scl{q_\init}{1} \actxx{u}{1} \scl{q}{1}$ and $w'$ is "strongly accepted" by $\scl{q'}{x}$, with $\wmu_1(q') \seq_1 q$.
   By Lemma~\ref{lem:safeMin-strongAcceptance}, $w'$ is "strongly accepted" by $q'$ in $\Aa$, and therefore by Lemma~\ref{lem:language-p-strong-acc}, $w'\in  L(\Aa,q') = u^{-1}L(\Aa)$, where the last equality follows from "uniform semantic determinism" of $\Aa$.

  Finally, "safe minimality" of $\safeminA$ follows from language equivalence and the fact that $\Lsafex(p)=L_x^{\safeminA}(\scl{p}{x})$.
\end{proof}

\begin{lemma}
  Given a "consistent" "layered automaton" $\Aa$, we can compute $\safeminA$ in polynomial time. 
\end{lemma}
\begin{proof}
  For building $\safeminA$, we just need to compute the $\seq_x$ classes of $\Aa$, as the transitions and morphism functions are defined directly on them.
  For computing the classes of $\seq_1$, we need to check for language equivalence of states, which can be done in polynomial time by Proposition~\ref{prop:inclusion-PTIME}.
  The computation of $\seq_x$ for $x>1$ consists in checking the equivalence of
  deterministic safety automata, which can be done in $\LOGSPACE$.
\end{proof}

\subsection{Centralisation}
We now show how to make an automaton "centralised" %(more precisely, it is enough to assume that the automaton $\Aa$ satisfies "N1").
while preserving "normalisation" and "safe minimality".
The following lemma shows that we can extend the $\fleqx$ relation to compare "$x$-SCCs".

\begin{lemma}\label{lem:comparison-x-SCCs}
  Let $\Aa$ be a "consistent" "layered automaton" satisfying "N1".
  Let $S_1,S_2$ be two $x$-SCCs in $\T_x$. If there are $p_1\in S_1$ and $p_2\in S_2$ such
  that $p_1\fleq_x p_2$ then for every $q_1\in S_1$ there is $q_2\in S_2$ with $q_1\fleq_x q_2$.
\end{lemma}
\begin{proof}
  Take a path $p_1\act{u}_xq_1$ in $\T_x$. 
  Since $p_1\fleq_x p_2$, $u\in \Lsafex(p_2)$. Let $q_2$ be the state reached in $\T_x$ by the path $p_2\act{u}_xq_2$.
  As $\Aa$ satisfies "N1", $q_2\in S_2$.
  Clearly, $q_1\sim_{x-1} q_2$ because $p_1\sim_{x-1} p_2$.
  Finally, $\Lsafex(q_1)=u^{-1}\Lsafex(p_1) \subseteq u^{-1}\Lsafex(p_2) = \Lsafex(q_2)$, so $q_1\fleq_x q_2$.
\end{proof}

We write $S\fleq_x S'$ if there are $p\in S$ and $p'\in S'$ such that $p\fleq_x p'$.

We define a ""centralisation operation"" that removes "$x$-SCCs" that are not $\fleq_x$-maximal.
Fix $x\geq 2$ odd and two components $S\fleq_x S'$.
\AP Let $\intro*\centralA$ be the layered automaton with $S$ removed.
This means that for every $y\geq x$ we remove all $p\in Q_y$ such that $\wmu_x(p)\in S$.

%Let us see what this operation means in terms of the "semantics" of $\Aac$.

\begin{lemma}\label{lem:centralisation-is-safe}
  Let $\Aa$ be a "safe minimal", "normal", and "consistent" "layered automaton".
  Then, $\centralA$ is also  "safe minimal", "normal" and "consistent".
  Moreover, $L(\centralA)=L(\Aa)$, and $\centralA$ is a "subautomaton" of $\Aa$.
\end{lemma}
\begin{proof}
  Since $\centralA$ has been obtained by removing whole "$x$-SCCs", no transitions between existent SCCs are added, nor any "$x$-SCC" is "covered" in $\centralA$.
  Therefore, $\centralA$ is also in "normal form".

  For all relations $\seq_x$ with $x\geq 2$, the relation $\seq_x$ in $\centralA$ is the restriction of $\seq_x$ to the states of the new automaton, as for every state $p$ in $\centralA$, $\LsafexxA{x}{\centralA}(p) =  \LsafexxA{x}{\Aa}(p)$.
  Therefore, the construction cannot create two $\seq_x$ equivalent states.

  For the second part of the lemma, we show that the mapping $\f\colon \centralA \to \Aa$ given by $\f(p)=p$ is an injective "morphism" that "preserves strong acceptance". The result then follows by Lemma~\ref{lem:consistent-morphism}.
  The fact that $\f$ is a "morphism" and injective is straightforward.

%   Let $w$ be "strongly accepted" (resp. "rejected@@str") by $p$ in $\centralA$.
% As the $x$-SCCs of $p$ and all its children are the same in $\Aa$ and in $\centralA$, $w$ is also "strongly accepted" (resp. "rejected") by $p$ in~$\Aa$.

For showing that $\f$ "preserves strong acceptance", assume that $w$ is "strongly accepted" (resp. "rejected") from some state in $\centralA$. The only interesting case here is when this state is a parent of $S$, because for all other states, their SCCs and those of the children have not changed. By  Lemma~\ref{lem:SCC-is-in-SCC} we know that all parents of $S$ are in the same SCC, so let $R \subseteq Q_{x-1}$ be the parent SCC of $S$. Assume that $x-1$ is even, the other case being symmetric. Let $w$ be "strongly accepted" from $r \in R$ in $\centralA$. Assume that $w$ is not "strongly accepted" from $r$ in $\Aa$. Then there is a decomposition $w=uw'$, a run $r \runxx{u}{x-1} q$, and $w' \in \Lsafex(p)$ for some $p \in Q_x$ with $\mu_{x-1}(p) = q$. Since $w$ is strongly accepted from $r$ in $\centralA$, we get that the SCC of $p$ has been removed in $\centralA$, that is, $p \in S$. Since $S\fleq_x S'$, there is a state $p' \in S'$ with $\Lsafex(p) \subseteq \Lsafex(p')$ and $\mu_{x-1}(p) = \mu_{x-1}(p')$. But then $w' \in \Lsafex(p')$ and $\mu_{x-1}(p') = q$, which would mean that $w$ is not "strongly accepted" from $r$ in $\centralA$ because $S'$ is still a child SCC of $R$ in $\centralA$.
\end{proof}

\begin{lemma}
  Given a "consistent" "normalised" "layered automaton" $\Aa$, we can compute $\centralA$ in polynomial time. 
\end{lemma}
\begin{proof}
  To obtain $\centralA$ it suffices to identify $x\geq 2$ and two states $p,p'\in Q_x$ in different SCCs of $\T_x$ such that $p\fleq_x p'$.
  Once such states are found, we just remove from $\Aa$ the "$x$-SCC" of $p$ and the states that map to it via $\wmu_x$.

  For this, we compute the relations $\fleq_x$.
  For each $x\geq 2$, we check for each pair $p\eqxx{x-1}p'$ whether $\Lsafex(p)\subseteq \Lsafex(p')$, which is just language inclusion of deterministic safety automata.
\end{proof}

Every application of the "centralisation operation" removes some states. 
When the "operation@@centralisation" cannot be applied anymore, after at most $|Q_{\geq 1}|$ steps, we get an equivalent "automaton@@lay" that is "normal", "safe minimal" and "centralised".

\subsection{Subautomaton and surjective morphism}
Applying the polynomial-time transformations presented in the previous subsection, we obtain an automaton $\Aa_L$ that is "normal", "central" and "safe minimal", which proves the first part of Theorem~\ref{thm:minimisation}.
We show the second part of this theorem, that is, that the initial automaton $\Aa$ contains a "subautomaton" admiting a surjective "strong morphism" to $\Aa_L$.

% Indeed, the automaton $\Aa'$ obtained after "normalising" and "centralising" $\Aa$ is an equivalent "subautomaton".
% This automaton $\Aa'$ admits a surjective "morphism" to $\Aa_L$, obtained by "safe minimising" $\Aa'$. 

This construction took $3$ steps, producing the following sequence of automata:
\[\Aa  \;\; \leadsto \;\;  \Aa_{\mathsf{normal}} \;\;  \leadsto \;\;  \Aa_{\mathsf{safeMin}} \;\; \leadsto  \;\;   \Aa_{\mathsf{central}} = \Aa_L. \]
Moreover, we have proven that these automata satisfy the following relations:
\begin{itemize}
  \item $\Aa_{\mathsf{normal}}$ is a "subautomaton" of $\Aa$,
  \item $\Aa_{\mathsf{normal}}$ admits a surjective "strong morphism" to  $\Aa_{\mathsf{safeMin}}$. % given by $\f(q) = \scl{q}{x}$,
  \item $\Aa_{\mathsf{central}}$ is a "subautomaton" of $\Aa_{\mathsf{safeMin}}$ obtained by removing subtrees of at levels $\geq 2$ of the latter automaton.
\end{itemize}

We first note the following.
\begin{lemma}
  Let $\f:\Aa\to\Bb$ be a "morphism@@lay" between "consistent" "layered automata", and let $\Bb'$ be a "subautomaton" of $\Bb$ obtained by removing some subtrees
  at layers $\geq 2$. Then, $\f^{-1}(\Bb')$ is a "layered automaton" and a
  "subautomaton" of $\Aa$. 
\end{lemma}
\begin{proof}
  Clearly $\Aa'=\f^{-1}(\Bb')$ is a subautomaton of $\Aa$ because it is obtained
  by removing some nodes. 
  To check that it is a layered automaton first observe that layer $1$ of
  $\Bb'$ is the same as layer $1$ of $\Bb$, so the layer $1$ of $\Aa'$ is the
  same as layer $1$ of $\Aa$.
  It remains to verify that the $\m_x$ functions are still "morphisms" if restricted to $\Aa'$.
  Take $p\in\Tt^{\Aa'}_x$ for some $x\geq 2$. We have that $\m_{x-1}(p)\in
  \Tt^{\Aa'}$ because $\Bb'$ is obtained from $\Bb$ by removing subtrees.
  Clearly the transitions are preserved by $\m_x$ in $\Aa'$ as they are
  preserved in $\Aa$.
\end{proof}

Let $\Aa'$ be the subautomaton of $\Aa_{\mathsf{normal}}$ given by the previous lemma applied to the morphism $\f: \Aa_{\mathsf{normal}} \to \Aa_{\mathsf{safeMin}}$ and the "subautomaton"  $\Aa_{\mathsf{central}}$ of $\Aa_L$.
Then, $\Aa'=\f^{-1}(\Aa_{\mathsf{central}})$ is a "subautomaton" of
$\Aa_{\mathsf{normal}}$, hence also a "subautomaton" of $\Aa$.
Since $\f$ is a "strong morphism" its restriction to $\Aa'$ is also strong.
Moreover, $\f$ is surjective, as $\Aa'$ is the preimage of $\f$.
This concludes the proof of Theorem~\ref{thm:minimisation} as $\Aa_{\mathsf{central}}=\Aa_L$.

\begin{remark}[Order of the operations]
  We order the operations to ensure that each step preserves the properties
  obtained in previous steps. "Normalisation" is applied first, since the
  lowering operation may disrupt "centrality" and "safe minimality".
  Interestingly, safe minimisation does not always preserve
  "centralisation"—unlike the coBüchi case
  (see~\cite{Abu.Kup.Minimization2022}), where "centralisation" can precede safe
  minimisation. By contrast, Lemma~\ref{lem:centralisation-is-safe} establishes
  that "centralisation" preserves both "normalisation" and "safe minimality",
  making it the natural final step. 
\end{remark}

\section{Congruence-based characterisation}\label{sec:congruences}

In this section, we provide an alternative characterization of minimal "layered
automata" based on congruences over tuples of finite words. This approach yields
a direct, algebraic construction of the canonical minimal "layered automaton"
for any $\omega$-regular language. The congruence-based perspective complements
the syntactic minimisation procedure presented in Section~\ref{sec:minimisation} and offers
additional insight into the structure of "layered automata" through the lens of
equivalences on sequences of finite words. 
For example, the maximal length of such sequences reflects a parity index of the
language. 

\subsection{Definitions and basic properties of congruences}

Let $L \subseteq \Sigma^\omega$ (the definitions can be applied to every language of infinite words, but the results concerning "layered automata" only work for regular $\omega$-languages). In our congruence-based description of the minimal "layered automaton" for $L$, the states on layer $x$ are classes of an equivalence relation $\tc$ over $x$-tuples $\bu = (u_1, \ldots, u_x) \in (\Sigma^*)^x$ of finite words. The relation $\tc$ can be seen as one equivalence relation over the set of all finite tuples (of length at least one) of finite words over $\Sigma$, in which equivalent tuples must have the same length. 
For defining the transition relations on the individual layers and the morphisms between the layers, we consider two operations on such finite tuples:\AP
\begin{description}
\item[""Concatenation"":] We write $\bu \cdot u_{x+1}$ for concatenation with a word $u_{x+1} \in \Sigma^*$ in the last component, that is $\bu \cdot u_{x+1} = (u_1, \ldots, u_{x-1}, u_xu_{x+1})$ is the tuple that has the same first $x-1$ components as $\bu$, and $u_xu_{x+1}$ as last component (if $x=1$, then the resulting tuple is just $(u_1u_{2})$). We use this operation for defining the transition relation on the individual layers. \AP
\item[""Merging"":] We write $(\bu,u_{x+1})$ for the tuple $(u_1, \ldots, u_x,u_{x+1})$ obtained by adding $u_{x+1}$ as additional component. Merging the last two components of the tuple $(\bu,u_{x+1})$ of length at least $2$, is the operation of removing the last component and "concatenating" it to the second last component, resulting in the tuple $\bu \cdot u_{x+1}$. We use this operation for the morphism from layer $x+1$ to layer $x$. 
\end{description}
%Roughly, the equivalence relation on $x$-tuples $\bu$ captures information 
Before going into the formal definitions, let us look an example to get an intuition for the notions that we are going to define.
\begin{example}\label{ex:congruence-aba}
  Consider the language $L$ of all words over $\Sigma = \{a,b\}$ in which $aba$ occurs finitely often, and $a$ and $b$ both occur infinitely often. The congruence automaton for $L$ is shown in Figure~\ref{fig:congruence-aba}. The language is prefix independent, so there is only one class on layer $1$, and all the representatives of classes in the other layers can use $\e$ as first component. \textbf{Layer $2$} ensures that all words that do not contain $aba$ are safe from some state. In order to capture this, we introduce the notion of a tuple being (equivalent to) $\bot$. In the example, $(u_1,u_2) \tc \bot$ iff $u_2$ contains $aba$ as infix. This is captured by looking at possible extensions $u_2w$ of $u_2$ into words $u_1(u_2w)^\omega \in L$.
  %(in the general definition this is $u_1 \ldots u_x(u_{x+1}w)^ \omega \in L$ iff $x$ is even).
  Such an extension exists iff $u_2$ does not contain $aba$. And the tuples for which no such extension exists are declared $\bot$ (when we discuss layer $3$ of the example further below, we give some more details on the definition of $\bot$). 
  This notion of tuples being $\bot$ defines a safe language for each tuple: the set of finite words that can be appended to the last component without making the tuple $\bot$. On layer $x$ our equivalence merges $x$-tuples that are equivalent w.r.t.\ this safe language. This results in three classes on layer $2$ in the example that track the prefixes of $aba$, namely $\tccls{(\e,\e)},\tccls{(\e,a)},\tccls{(\e,ab)}$. In the picture, the representative $(\e,bb)$ is used for the class $\tccls{(\e,\e)}$. The intuitive reason is that $bb$ is a word that ensures the following: "concatenating" it to any tuple (on layer $2$) results in a tuple that is either $\bot$ or equivalent to $(\e,bb)$. So in that sense, $bb$ in the second component points to a unique class. This is captured by our definition of "pointed" tuples further below. We use this notion for selecting those classes that are used for the construction of the automaton: only classes that contain a "pointed" tuple are used. There can be classes that contain only "pointed" tuples or both, "pointed" and non-pointed tuples (as $(\e,\e)$ and $(\e,bb)$ in the example), but there can also be classes that do not contain any "pointed" tuple, and these are not used as states. The latter case does not happen in this example, we refer to Example~\ref{ex:abc} for this.
  \textbf{Layer 3} is responsible for checking that infinitely many $a$ and $b$ occur, so words that contain only one of the two letters have to be safe on layer $3$ (and are thus rejected because $3$ is odd). Note that there is no state on layer $3$ that maps to $\tccls{(\e,ab)}$ on layer $2$. The reason is that if $\tccls{(\e,ab)}$ occurs infinitely often in a safe run on layer $2$, then the word has to contain infinitely many $a$ and $b$, because all words looping on that class on layer $2$ contain $a$ and $b$. So there is nothing to check on layer $3$ in this case. 
  This is captured by our definition of $\bot$ as follows. We say that a tuple $(u_1, u_2, u_3)$ on layer $3$ is $\bot$ if either $(u_1,u_2u_3)$ is $\bot$ on layer $2$, or every extension $w \in \Sigma^+$ that closes a loop on the class of $(u_1,u_2)$, written as $(u_1,u_2u_3w) \tc (u_1,u_2)$, is such that the resulting word $u_1u_2(u_3w)^\omega$ is in $L$. Intuitively, this means that we do not have to track anything on layer $3$ because there is nothing to reject if the class on layer $2$ occurs infinitely often. Consider the tuple $(\e,ab,\e)$ in the example. Merging the last two components results in $(\e,ab)$. Our definition ensures that $(\e,ab,\e) \tc \bot$ because every nonempty $w$ such that $(\e,abw) \tc (\e,ab)$ must be such that $abw$ does not contain $aba$ (otherwise $(\e,abw)$ would be $\bot$ and thus not equivalent to $(\e,ab)$), and $w$ ends in $ab$. Then clearly $ab(w)^\omega \in L$ because $w$ contains both $a$ and $b$, and $ab(w)^\omega$ does not contain $aba$, which follows from the facts that $abw$ does not contain $aba$ and that $w$ ends in $ab$.
\end{example}
\begin{figure}
  \begin{center}
  \includegraphics[]{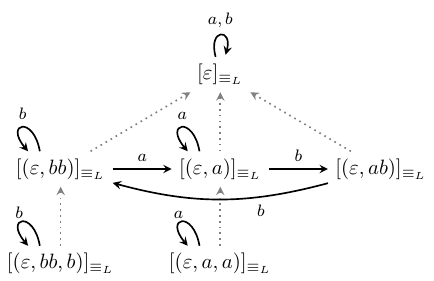}
  \end{center}
  \caption{Congruence automaton for the language of words with finitely many $aba$ and infinitely many $a$ and $b$ from Example~\ref{ex:congruence-aba}.\label{fig:congruence-aba}}
\end{figure}
We now proceed with the formal definitions, which generalise the definitions for the coBüchi case on pairs of words from \cite{Lod.Wal.Congruences2025}.\footnote{There is a small difference. In \cite{Lod.Wal.Congruences2025} for the coBüchi case all states of the congruence automaton are classes of tuples of length $2$. With our definitions here applied to the coB\"uchi case, there might also be states that are classes of tuples of length $1$ in some cases. This flexibility of tuple length actually makes the definition simpler because it avoids a special treatment of $\e$ as in \cite{Lod.Wal.Congruences2025}.}

\AP For tuples $\bu = (u_1, \ldots, u_x)$, $\bv = (v_1, \ldots, v_x)$ of finite words, we define inductively over the length of the tuples the notions $\bu \intro*\tc \bot$, $\bu \tc \bv$, and $\bu$ is "pointed". 

For the induction base $x=1$, define 
\begin{itemize}
\item $(u_1) \not\tc \bot$ for all $u_1$, 
\item $(u_1) \tc (v_1)$ iff $u_1 \sim_L v_1$, 
\item $(u_1)$ "pointed" for all $u_1$.
\end{itemize}

\textbf{Bottom and safe language.} 
For the induction step, let $(\bu, u_{x+1}) \tc \bot$ if
\[%\left\{
\begin{array}{l}
   \bu\cdot u_{x+1} \tc \bot \text{, or } \\[2mm] 
%   u_{x+1} \not= \e \text{ and } 
   \forall w \in \Sigma^+: \bu\cdot u_{x+1}w \tc \bu \Rightarrow [u_1 \cdots u_x(u_{x+1}w)^\omega \in L \; \Leftrightarrow \; x \text{ even}].
\end{array}
%\right.
\]
Note that we require $w$ to be non-empty in order to ensure that the looping part $u_{x+1}w$ is non-empty. 
\begin{remark}\label{rmk:bot}
  \begin{enumerate}
  \item Let $(\bu,u_{x+1})$ be such that there is no $w\in \Sigma^+$ so that $\bu\cdot u_{x+1}w \tc \bu$. Then, $(\bu,u_{x+1})\tc \bot$, as the quantification ``$\forall w$  [...]'' is vacuous.
  \item In proofs we often show that a tuple is $\not\tc \bot$ by finding a $w \in \Sigma^+$ with $\bu\cdot u_{x+1}w \tc \bu$ and $[u_1 \cdots u_x(u_{x+1}w)^\omega \in L \; \Leftrightarrow \; x+1 \text{ even}]$ (note the change to ``$x+1$ even'' here). We sometimes call such a $w$ a \emph{witness for the tuple not being $\bot$}.
  \end{enumerate}
\end{remark}
The next lemma states that $\bot$ is invariant under right "concatenation" in the last component, and thus we can use it for defining a safety language.
\begin{lemma}\label{lem:bot-properties}
  Let $\bu\in  (\Sigma^*)^x$ and $a \in \Sigma$. Then $\bu \tc \bot \Rightarrow \bu \cdot a \tc \bot$.
\end{lemma}
\begin{proof}
Induction on $x$. For $x=1$ the claim is true because no tuple of length $1$ is $\bot$. For $x+1$, assume that $(\bu, u_{x+1} a) \not\tc \bot$. Then $\bu \cdot u_{x+1} a \not\tc \bot$, and hence by induction  $\bu \cdot u_{x+1} \not\tc \bot$. Further, there is $w \in \Sigma^+$ such that $\bu \cdot u_{x+1}aw \tc \bu$ and $u_1 \cdots u_x(u_{x+1}aw)^\omega \in L$ iff $x+1$ even. Then $aw$ witnesses that $(\bu, u_{x+1}) \not\tc \bot$.
\end{proof}
For an infinite word $w \in \Sigma^\omega$ we define $(\bu,w) \not\tc \bot$ if $(\bu,v) \not\tc \bot$ for all prefixes $v$ of $w$.
\AP The ""safe language@@congr"" of a tuple is then the set of  finite and infinite words that can be "concatenated" to the tuple without making the tuple $\bot$:
\[
\intro*\Ls((\bu,u_{x+1})) := \{w \in \Sigma^* \cup \Sigma^\omega \mid (\bu,u_{x+1}w) \not\tc \bot\}.
\]
Note that the safe language is completely determined by the finite words it contains (the infinite words in the safe language are those for which all finite prefixes are in the safe language).

\textbf{Equivalence.} Based on that, let
\[
  (\bu, u_{x+1}) \tc (\bv, v_{x+1}) \; \text{ if } \hspace{2mm}
  \begin{array}{l}
    \bu\cdot u_{x+1} \tc \bv\cdot v_{x+1} \text{, and } \Ls((\bu,u_{x+1})) = \Ls((\bv,v_{x+1}))
%    \forall w \in \Sigma^*: (\bu,u_{x+1}w) \tc \bot \Leftrightarrow (\bv,v_{x+1}w) \tc \bot.
\end{array}
\]
This is residual equivalence w.r.t.\ the safe language together with the condition that the tuples map to the same class on the previous layer. 
\AP The class of a tuple is written $\intro*\tccls{\bu}$.

\begin{example}
  Let $L = \mathsf{parity}(1,2,\dots, d)$. Then, the tuples that are $\tc \bot$ are those that have length $> d$ or that contain $x$ in some component of index $>x$.  
  Therefore, for every length $x\in\{1, \ldots, d\}$, there are exactly two equivalence classes: $\tccls{(\e,1,2,\dots,x-1)}\tc \bot$, and $\tccls{(\e,\dots,\e)} \not \tc \bot$.  
\end{example}

The following lemma shows that $\tc$ is a congruence w.r.t.\ to the operations that we consider on tuples of words, namely right "concatenation" in the last component, "merging" the last two components, and extension with a new last component. Note that this implies that $\tc$ is preserved under arbitrary sequences of such operations, for example first extending the tuple by a new last component, and then filling this new component with a word by repeated right "concatenation".
\begin{lemma}\label{lem:tc-properties}
  Let $\bu, \bv \in  (\Sigma^*)^x$, $a \in \Sigma$, $u_{x+1},v_{x+1} \in \Sigma^*$.
  The following properties hold.
  \begin{enumerate}
  \item $\bu \tc \bv \Rightarrow \bu\cdot a \tc \bv \cdot a$. \label{item:tc-rc-concat}
  \item $(\bu,u_{x+1}) \tc (\bv,v_{x+1}) \Rightarrow \bu\cdot u_{x+1} \tc \bv \cdot v_{x+1}$.\label{item:tc-rc-merge}
  \item $\bu \tc \bv \Rightarrow (\bu,\e) \tc (\bv,\e)$.\label{item:tc-rc-extend}
  \end{enumerate}
\end{lemma}
\begin{proof}
 The proof is a simple application of the definitions.
    \begin{enumerate}
  \item Clearly, if $\Ls(\bu) = \Ls(\bv)$, then also $\Ls(\bu \cdot a) = \Ls(\bv \cdot a)$. Based on this observation, the proof is a straightforward induction on $x$.
  \item This property is part of the inductive definition of $\tc$.
  \item %It is sufficient to show $\bu \tc \bv \Rightarrow (\bu,\e) \tc (\bv,\e)$, then one can repeatedly apply property \ref{item:tc-rc-concat} in order to fill the last component with $w$.
Assume that $\bu \tc \bv$. Then clearly $\bu \cdot \e \tc \bv \cdot \e$. For showing that $\Ls(\bu,\e) = \Ls(\bv,\e)$, let $u_{x+1} \in \Ls(\bu,\e)$. Then there is $w$ such that $\bu \cdot u_{x+1}w \tc \bu$ and $u_1 \cdots u_x(u_{x+1}w)^\omega \in L$ iff $x+1$ even. By assumption we have $\bu \tc \bv$, which by \ref{item:tc-rc-concat} implies that $\bu \cdot u_{x+1}w \tc \bv \cdot u_{x+1}w$, and thus $\bv \cdot u_{x+1}w \tc \bv$. Further, $v_1 \cdots v_x (u_{x+1}w)^\omega \in L$ iff  $u_1 \cdots u_x(u_{x+1}w)^\omega \in L$ because $v_1\cdots v_x \sim_L u_1 \cdots u_x$. Therefore, $u_{x+1} \in \Ls(\bv,v_{x+1})$. By symmetry we obtain that $\Ls(\bu,\e) = \Ls(\bv,\e)$. This shows that $(\bu,\e) \tc (\bv,\e)$.
  \end{enumerate}
\end{proof}

\textbf{"Pointed".}
We can use all classes of $\tc$ (except $\bot$) for the construction of a "layered automaton" that accepts $L$. However, this automaton is not necessarily minimal (see Example~\ref{ex:abc} below). We now proceed with the definition of "pointed" tuples, which are used to select those classes of $\tc$ that are actually used in the definition of the minimal automaton $\Aatc$. The definition of "pointed" is made such that it allows us to construct "central sequences" for the classes that contain "pointed" tuples. Consider a tuple $(\bu,u_{x+1}) \not\tc \bot$. It can be seen as the tuple that is reached from $(\bu,\e)$ by reading $u_{x+1}$. The parent of $(\bu,\e)$ on level $x$ is $\bu$. We now say that $(\bu,u_{x+1})$ is pointed, if reading $u_{x+1}$ from any $x+1$-tuple $(\bv,v_{x+1})$ that has $\bu$ as parent on layer $x$, results in a tuple $(\bv,v_{x+1}u_{x+1})$ that is either $\bot$ or equivalent to $(\bu,u_{x+1})$. In other words, $u_{x+1}$ uniquely points to the class of $(\bu,u_{x+1})$ when starting from a child of $\bu$.

\AP
Formally, $(\bu, u_{x+1}) \not\tc \bot$ is ""pointed"" if $\bu$ is "pointed" and for all "pointed" $\bv$  and all $v_{x+1}$ with $\bv\cdot v_{x+1} \tc \bu$:
\[
(\bv, v_{x+1}u_{x+1}) \not\tc \bot \Rightarrow (\bv, v_{x+1}u_{x+1}) \tc (\bu, u_{x+1}).
\]
This is illustrated in Figure~\ref{fig:pointed}, where the gray triangle represents all the children $(\bv, v_{x+1})$ of $\tccls{\bu}$ that are used in the definition.
\begin{figure}
  \begin{center}
  \begin{tikzpicture}[
      parent/.style={
        dotted, ->,thick, >=stealth, draw=black!50, font=\footnotesize
      },
      layer/.style={
        ->,thick, >=stealth, draw=black, font=\footnotesize
      },
      iso/.style={
        dashed, ->,thick, >=stealth, draw=black!20, text=black!20, font=\footnotesize
      }
    ]
    % nodes for classes
    \node (bu) {$\tccls{\bu}$};
    \node[right=2 of bu] (buuxp1) {$\tccls{\bu\cdot u_{x+1}}$};
    \node[below=of buuxp1](bu-uxp1) {$\tccls{(\bu,u_{x+1})}$};
    % nodes for triangle
    \node[below=of bu](bbu) {};
    \node[left=of bbu](bbul) {};
    \node[right=of bbu](bbur) {};
    % node bot
    \node[left=of bbul](bot) {$\bot$};
    % startpoints of edges from triangle
    \node (a1) at ($(bbul)+(5mm,0)$) {};
    \node (a2) at ($(bbur)+(-5mm,0)$){};
    % nodes for layer numbers
    \node[right=3 of buuxp1,text=black!30] (x) {$x$};
    \node[below=of x,text=black!30] (xp1) {$x+1$};

    % triangle
    \fill[gray!40] (bu.south) --  (bbul.south) --  (bbur.south) -- cycle;

    % edges
    \path[layer] (bu) edge node[above]{$u_{x+1}$} (buuxp1);
    \path[layer] (a2) edge[bend right]  node[above]{$u_{x+1}$} (bu-uxp1);
    \path[layer] (a1) edge[bend left] node[below]{$u_{x+1}$} (bot);
    \path[parent] (bu-uxp1) -- (buuxp1);
    \path[parent] (bbul.south) -- (bu.south);
    \path[parent] (bbur.south) -- (bu.south);
  \end{tikzpicture}

  \end{center}
  \caption{Illustration of the definition of $(\bu,u_{x+1})$ being "pointed". Basically, from every child of $\tccls{\bu}$, the word $u_{x+1}$ leads to $\bot$ or to $(\bu,u_{x+1})$.  \label{fig:pointed}}
\end{figure}
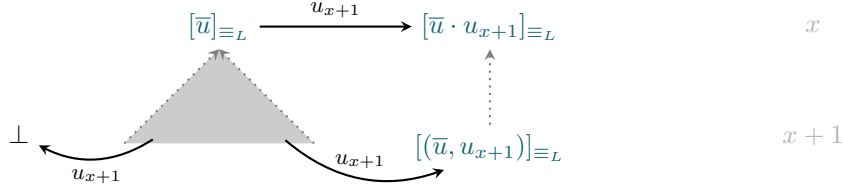

\begin{example} \label{ex:abc}
Consider the alphabet $\S := \{a,b,c\}$ and the language of all
$\omega$-words that contain at least one of the letters only finitely often. This language is prefix independent, so the first layer has only one class, and we only need to consider $\e$ in the first component. Further, $(\e,u_2) \tc \bot$ iff $u_2$
contains all letters, and $(\e,u_2) \tc (\e,v_2)$ if the same letters occur in $u_2$ and $v_2$. All tuples of length more than $2$ are $\bot$. Otherwise, if $(\e,u_2,u_3) \not\tc \bot$, then $(\e,u_2u_3) \not\tc \bot$ and there is $w \in \Sigma^+$ with $(\e,u_2u_3w) \tc (\e,u_2)$ and $u_2(u_3w)^\omega \not\in L$. From $(\e,u_2u_3w) \tc (\e,u_2)$ we know that $u_2$ and $u_2u_3w$ contain the same letters. From $(\e,u_2u_3) \not\tc \bot$ we conclude that $u_2u_3$ does not contain all the letters. Hence $u_3w$ does not contain all the letters, so $u_2(u_3w)^\omega \in L$, a contradiction.
The full diagram of $\tc$ on layer $2$ is shown in Figure~\ref{fig:abc} (except the class $\bot$, where the missing transitions are going). The parent of all these classes is the unique class on level $1$. For the automaton, we only need the classes $\tccls{(\e,ab)}$, $\tccls{(\e,ac)}$, and $\tccls{(\e,bc)}$, which respectively check that $c$, $b$, or $a$ occur finitely often. And these are precisely the classes that contain "pointed" tuples. Consider, for example, $(\e,ab)$. Reading $ab$ starting from any class in the diagram either goes to $\bot$ or ends in $\tccls{(\e,ab)}$, so $(\e,ab)$ is "pointed". But for example, the class of $(\e,a)$, consisting of tuples $(\e,a^i)$, $i \ge 1$, does not contain any "pointed" tuple. If $(\e,a^i)$ were pointed, then with $(\bv,v_{x+1}) = (\e,ab)$ in the definition of pointed, we would have $(\e,aba^i) \tc \{\bot, (\e,a)\}$ but we have $(\e,aba^i) \tc (\e,ab)$. In other words, reading $a^i$ starting from $(\e,ab)$ leads back to $(\e,ab)$ and not to $(\e,a)$ or $\bot$.
\end{example}

\begin{figure}
  \begin{center}
  \includegraphics[]{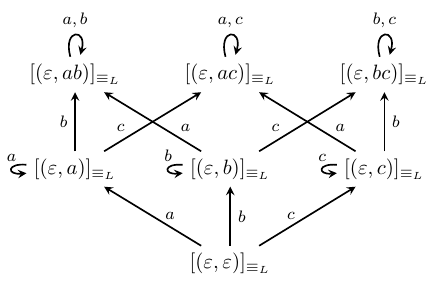}
  \end{center}
  \caption{Diagram of all the $\tc$-classes on layer $2$ for the language $L$ from Example~\ref{ex:abc} (at least one letter finitely often).\label{fig:abc}}
\end{figure}

The following lemma asserts that being "pointed" is a property that is preserved under right "concatenation" in the last component and under "merging" the last two components.

\begin{lemma}\label{lem:pointed-properties}
  Let $\bu \in (\Sigma^*)^x$, $a \in \Sigma$, $u_{x+1} \in \Sigma^*$.
  The following properties hold.
  \begin{enumerate}
  \item If $\bu$ is "pointed" and $\bu \cdot a \not \tc \bot$, then $\bu \cdot a$ is "pointed". \label{item:pointed-concat}
  \item If $(\bu,u_{x+1})$ is "pointed", then $\bu \cdot u_{x+1}$ is "pointed". \label{item:pointed-merge}
  \end{enumerate}
\end{lemma}
\begin{proof}
  Note that the property~\ref{item:pointed-merge} follows from property~\ref{item:pointed-concat}: If $(\bu,u_{x+1})$ is "pointed", then by definition $\bu$ is "pointed". By repeated application of property~\ref{item:pointed-concat} we get that $\bu \cdot u_{x+1}$ is "pointed".
  
  Property~\ref{item:pointed-concat} clearly holds for $x=1$ because all tuples of length $1$ are "pointed". We now show that claim for tuples of length at least two, so assume that $(\bu,u_{x+1})$ is "pointed". For showing that  $(\bu,u_{x+1}a)$ is pointed, let  $\bv \in (\Sigma^*)^x$, $v_{x+1} \in \Sigma^*$ such that $\bv \cdot v_{x+1} \tc \bu$. Then  $(\bv,v_{x+1}u_{x+1}) \tc \{\bot,(\bu,u_{x+1})\}$ because $(\bu,u_{x+1})$ is "pointed".

  If $(\bv,v_{x+1}u_{x+1}) \tc \bot$, then $(\bv,v_{x+1}u_{x+1}a) \tc \bot$ by Lemma~\ref{lem:bot-properties}. If  $(\bv,v_{x+1}u_{x+1}) \tc (\bu,u_{x+1})$, then $(\bv,v_{x+1}u_{x+1}a) \tc (\bu,u_{x+1}a)$ by Lemma~\ref{lem:tc-properties}(\ref{item:tc-rc-concat}).
\end{proof}

%%%%%%%%%%%%%%%%%%%%%%%%%%%%%%%%%%%%%%%%%%%%%%%%%%%%%%%%%%%%
\subsubsection*{Regularity and Finiteness.}
We show that for $\omega$-regular languages, the number of $\tc$-classes is finite. The proof consists of two parts, the first part (Lemma~\ref{lem:tc-finite-index}) shows that on each layer (for each fixed length of the tuples), there are finitely many $\tc$-classes. The second part (Lemma~\ref{lem:tc-finite-depth}) shows that there are only finitely many layers that contain classes that are different from $\bot$.
\begin{lemma}\label{lem:tc-finite-index}
If $L$ is $\omega$-regular, then $\tc$ has finite index over $(\Sigma^*)^x$ for all $x \ge 1$. 
\end{lemma}
\begin{proof}
  Let $L = L(\Aa)$ for a nondeterministic Büchi automaton (NBA) $\Aa$. For a finite word $u \in \Sigma^*$, let the transition profile of $u$ contain for each pair $p,q$ of states of $\Aa$ the information, whether there is a run from $p$ to $q$ on $u$, and whether there is such a run passing through an accepting state. It is not difficult to show by induction on $x$ that if $\bu,\bv \in (\Sigma^*)^x$ are such that $u_i,v_i$ have the same transition profile for all $i$, then $\bu \tc \bv$. 
  Since there are only finitely many possible transition profiles, the claim follows.  So let $\bu,\bv \in (\Sigma^*)^x$ be such that $u_i,v_i$ have the same transition profile for each $i$. If $u_1,v_1$ have the same transition profile, then they also are in the same residual class, because the same states are reachable from the initial state of $\Aa$ via $u_i$ and via $v_i$. So the claim follows for $x=1$. For the step, consider $u_{x+1},v_{x+1} \in \Sigma^*$ with the same transition profile, and show that $(\bu,u_{x+1}) \tc (\bv,v_{x+1})$. Since $u_x,v_x$ and  $u_{x+1},v_{x+1}$ have, respectively, the same transition profile, also $u_xu_{x+1},v_xv_{x+1}$ have the same transition profile, so $\bu \cdot u_{x+1} \tc \bv \cdot v_{x+1}$ by induction. This implies that  $\bu \cdot u_{x+1}w' \tc \bv \cdot v_{x+1}w'$ for all $w' \in \Sigma^*$ by Lemma~\ref{lem:tc-properties}(\ref{item:tc-rc-concat}). Thus, $(\bu,u_{x+1}w') \tc \bu$ iff $(\bv,v_{x+1}w') \tc \bv$ because also $\bu \tc \bv$ by induction. 
  Further, since all pairs $u_i,v_i$ have the same transition profiles, we get that $u_1\cdots u_x(u_{x+1}w')^\omega \in L$ iff $v_1\cdots v_x(v_{x+1}w')^\omega \in L$. This means that for every $w \in \Ls(\bu,u_{x+1})$, the same words witnessing that $(\bu,u_{x+1}w) \not\tc \bot$ also witness that $(\bv,v_{x+1}w) \not\tc \bot$, and vice versa. Hence $\Ls(\bu,u_{x+1}) = \Ls(\bv,v_{x+1})$.
\end{proof}

The following lemma shows that $\tc$ has ``finite depth'' if $L$ is $\omega$-regular, in the sense that the length of tuples that are not $\bot$ is bounded by the parity index of the language.
\begin{lemma}\label{lem:tc-finite-depth}
If $L$ is $\omega$-regular and accepted by a "DPA" with priorities in $[1,d]$, then $\bu \tc \bot$ for all $\bu \in (\Sigma^*)^{d+1}$ (and hence also for all longer tuples). 
\end{lemma}
\begin{proof}
  Let $\Aa$ be a "DPA" for $L$ and let $n$ be bigger than the number of states of $\Aa$. Assume by contradiction that there is $\bu = (u_1, \ldots, u_{d+1}) \in (\Sigma^*)^{d+1}$ with $\bu \not\tc \bot$. 
  We inductively define words $v_x \in \Sigma^*$ for $x \in \{2, \ldots, d+2\}$ and show that the loop that the run of $\Aa$ on $u_1v_2^n$ enters, must visit $d$ different priorities with the smallest one being even. This is not possible if $\Aa$ only uses priorities $1,\ldots,d$.

  We start with the definition of $v_{d+2} := \e$, which is only needed for a simple induction base.

  Assume that  $v_{x+1}$ has already been defined for $x \in \{2, \ldots, d+1\}$, and that $(u_1, \ldots, u_x) \cdot v_{x+1} \tc (u_1, \ldots, u_x)$. Since $v_{d+2} := \e$, this condition is clearly satisfied for the start of the sequence with $x = d+1$. From $(u_1, \ldots, u_x) \cdot v_{x+1} \tc (u_1, \ldots, u_x)$ we also get that $(u_1, \ldots, u_x) \cdot v_{x+1}^i \tc (u_1, \ldots, u_x)$ for all $i$ by Lemma~\ref{lem:tc-properties}(\ref{item:tc-rc-concat}).
  In particular $(u_1, \ldots, u_x) \cdot v_{x+1}^n \not\tc \bot$. Hence, there is a word $w_x \in \Sigma^+$, such that $(u_1, \ldots, u_{x-1}) \cdot u_xv_{x+1}^nw_x \tc (u_1, \ldots, u_{x-1})$ and $[u_1 \cdots u_{x-1}(u_xv_{x+1}^nw_x)^ \omega \in L \Leftrightarrow x \text{ is even}]$.
  Let $v_x := u_xv_{x+1}^nw_x$.
  Note that
  \[
  (u_1, \ldots, u_{x-1}) \cdot v_x = (u_1, \ldots, u_{x-1}) \cdot u_xv_{x+1}^nw_x \tc (u_1, \ldots, u_{x-1}),
  \]
  so $v_x$ satisfies the property that was assumed for the inductive definition.

  Our goal is to show that $v_x$ induces $(d+1)-x$ many nested loops of different parity in $\A$ whose acceptance status alternates. For this we need to read $v_x$ starting from the right type of state. Let $q_0$ be the initial state of $\Aa$, and $q_i$ be the state reached in $\Aa$ after reading $u_1 \cdots u_i$:
  \[
q_0 \xrightarrow{u_1} q_1  \xrightarrow{u_2} q_2  \xrightarrow{u_3} \cdots  \xrightarrow{u_{d+1}} q_{d+1}
\]
We now show that for each $q \sim q_{x-1}$ such that a non-empty power of $v_x$ loops on $q$, this loop contains $d+2-x$ many different priorities with the least one having the parity of $x$. So assume that $q \xrightarrow{v_x^+} q$ where the plus represents some non-empty power.
By the choice of $v_x$, we have that $u_1 \cdots u_{x-1}v_x^\omega \in L$ iff $x$ is even. Hence, $v_x^\omega$ is accepted from $q$ iff $x$ is even, which means that the least priority on the loop has the same parity as $x$.

We show that the loop contains at least $d+2-x$ many priorities by a downward induction starting with $x = d+1$. The induction base is trivial because the loop must contain at least 1 priority.

If $x < d+1$, then consider the states reached after the factors of the first $v_x = u_xv_{x+1}^nw_x$:
\[
q \xrightarrow{u_x} p_0  \xrightarrow{v_{x+1}} p_1  \xrightarrow{v_{x+1}} p_2 \cdots  \xrightarrow{v_{x+1}} p_n \xrightarrow{w_x} q'
\]
Since $q \sim q_{x-1}$ and $p_0$ is reached from $q$ via $u_x$, we get that $p_0 \sim q_{x}$. Furthermore, $(u_1, \cdots, u_x) \cdot v_{x+1}^i \tc  (u_1, \cdots, u_x)$ as already noted above. So we get that $p_i \sim q_x$ for all $i \in \{0, \ldots, n\}$. By the choice of $n$, there must be $i \not= j$ with $p_i = p_j =:p$, so we have found a state $p \sim q_x$ such that a non-empty power of $v_{x+1}$ loops on $p$. By induction, we know that on this loop at least $d+2-(x+1)$ many priorities are visited with the least one having the parity of $x+1$. So the least priority visited on the loop $q \xrightarrow{v_x^+} q$ must be strictly smaller than the minimal one on the loop $p \xrightarrow{v_{x+1}^+} q$, which gives the desired claim.

From this we immediately get that the loop that is entered on $u_1v_2^n$ contains at least $d = d+2-2$ many priorities with the least one being even.
\end{proof}

%%%%%%%%%%%%%%%%%%%%%%%%%%%%%%%%%%%%%%%%%%%%%%%%%%%%%%%%%%%%
\subsection{The automaton $\Aatc$}

\AP For the definition of $\intro*\Aatc$ we assume that $L$ is $\omega$-regular. So in all statements that involve $\Aatc$, we assume that $L$ is an $\omega$-regular language (in the main results stated as theorems we mention this again explicitly, but not in all auxiliary lemmas).
Let $d$ be maximal such that there exists $\bu \in (\Sigma^*)^d$ with $\bu \not\tc \bot$. By Lemma~\ref{lem:tc-finite-depth} this is the minimal $d$ such that $L$ is accepted by a $[1,d]$-DPA. 

The $x$-th layer of $\Aatc$ consist of the equivalence classes of the "pointed" tuples of dimension $x$, the transitions corresponding to "concatenation" in the last component. The morphism from layer $x+1$ to layer $x$ is obtained by "merging" the last two components. And the initial state is as usual the class of the empty word. Formally, we define the "layered" automaton $\Aatc = (\qintc, \Ttc_1, \phitc_1, \ldots, \Ttc_{d-1}, \phitc_{d-1}, \Ttc_d)$ by
\begin{itemize}\setlength\itemsep{2mm}
\item $\Ttc_x = (\Qtc_x,\deltc_x)$ with 
\begin{itemize}\setlength\itemsep{2mm}
\item $\Qtc_x = \{\tccls{\bar{u}} \mid \bar{u} \in (\Sigma^*)^x \text{ "pointed"}\}$
\item $\deltc_x(\tccls{\bar{u}},a) =
  \begin{cases}
    \tccls{\bar{u}\cdot a} & \text{if } \bar{u}\cdot a \not\tc \bot\\
    \text{undefined} & \text{else}.
  \end{cases}$
\end{itemize}
\item $\qintc = \tccls{\e}$
\item $\phitc_x: \Qtc_{x+1} \rightarrow \Qtc_x$ with $\phitc_x(\tccls{(\bu,u_{x+1})}) = \tccls{\bu\cdot u_{x+1}}$.
\end{itemize}
Note that in the definition of $\deltc_x$ and $\phitc_x$ the result is independent of the chosen representative by Lemma~\ref{lem:tc-properties}(\ref{item:tc-rc-concat}) and (\ref{item:tc-rc-merge}). Further, $\phitc_x$ is a morphism because it does not matter whether we first "concatenate" and then "merge", or first "merge" and then "concatenate":
\[
\begin{array}{lcl}
  \phitc_x(\delta_{x+1,\tc}(\tccls{(\bar{u},u_{x+1}}),a)) &=& \phitc_x(\tccls{(\bar{u}, u_{x+1}a)}) \\
  &=& \tccls{\bar{u} \cdot u_{x+1}a} \\
  &=& \deltc_x(\tccls{\bar{u} \cdot u_{x+1}},a) \\
  &=& \deltc_x(\phitc_x(\tccls{(\bar{u}, u_{x+1})},a))
\end{array}
\]
So  $\Aatc$ is a "layered automaton".

%%%%%%%%%%%%%%%%%%%%%%%%%%%%%%%%%%%%%%%%%%%%%%%%%%%%%%%%%%%%
\subsection{Correctness of $\Aatc$}
We show that $\Aatc$ is a "consistent" "layered automaton" that accepts $L$ (Theorem~\ref{thm:Aatc-correct}). We first need some auxiliary definitions and lemmas.

Our first goal is to show that the "safe language"s of tuples are covered by the "pointed" tuples. This is expressed for finite and infinite words respectively in Lemmas~\ref{lem:safe-pointed-finite} and \ref{lem:safe-pointed-infinite}. For constructing corresponding tuples, we start with the auxiliary Lemma~\ref{lem:tuple-safe-decrease} that is applied in the proof of Lemma~\ref{lem:safe-pointed-finite} (so the reader may skip Lemma~\ref{lem:tuple-safe-decrease} now and come back to it when reading the proof of Lemma~\ref{lem:safe-pointed-finite}).
\begin{lemma}\label{lem:tuple-safe-decrease}
Let $\bu,\bv \in (\Sigma^*)^x$ and $u_{x+1},v_{x+1} \in \Sigma^*$. If $\bu \tc \bv\cdot v_{x+1}$, then $\Ls(\bu,u_{x+1}) \supseteq \Ls(\bv,v_{x+1}u_{x+1})$.
\end{lemma}
\begin{proof}
  Let $w \in \Ls(\bv,v_{x+1}u_{x+1})$, that is, $(\bv,v_{x+1}u_{x+1}w) \not\tc \bot$. We show that $w \in \Ls(\bu,u_{x+1})$. It is sufficient to show this for finite words $w$ because the infinite words in the "safe language"s are determined by the finite ones.
  
%  First consider the case that $u_{x+1}'vw = \e$. Then $u_{x+1} = w = v = \e$ and hence $\bu \tc \bu'\cdot u_{x+1}' = \bu'$. Since $\bu' \not\tc \bot$ because $(\bu',u_{x+1}'vw) \not\tc \bot$, we get that $\bu \not\tc \bot$ and thus $(\bu,vw) = (\bu,\e) \not\tc \bot$.

  If $(\bv,v_{x+1}u_{x+1}w) \not\tc \bot$, then there is $w' \in \Sigma^+$ with
  $\bv \cdot v_{x+1}u_{x+1}ww' \tc \bv$ and $v_1 \cdots
  v_x(v_{x+1}u_{x+1}ww')^\omega \in L$ iff $x+1$ is even. 
  Observe that $\bu \tc \bv\cdot v_{x+1}$ in
  particular implies $u_1 \cdots u_x \tc v_1\cdots v_xv_{x+1}$. 
  Since $v_1 \cdots v_x(v_{x+1}u_{x+1}ww')^\omega = v_1 \cdots
  v_xv_{x+1}(u_{x+1}ww'v_{x+1})^\omega$, we get 
  $u_1 \cdots u_x(u_{x+1}ww'v_{x+1})^\omega \in L$ iff $x+1$ is even. Further 
  \[
  \bu \cdot u_{x+1}ww'v_{x+1} \tc \underbrace{\bv \cdot v_{x+1}u_{x+1}ww'}_{{}\tc \bv} v_{x+1} \tc \bv \cdot v_{x+1} \tc \bu
  \]
  So the suffix $w'v_{x+1}$ witnesses that $(\bu,u_{x+1}w) \not\tc \bot$ and hence $w \in \Ls(\bu,u_{x+1})$.
\end{proof}
The following lemma states that if we start from a "pointed" tuple $\bu$ of
length $x$, then the "safe language" of the $(x+1)$-length tuple $(\bu,\e)$ with
empty $(x+1)$-st component is covered by the "pointed" children $(\bv,v_{x+1})$ of $\tccls{\bu}$, that is, those "pointed" $(\bv,v_{x+1})$ of
length $x+1$ with $\phitc_x(\tccls{(\bv,v_{x+1})}) = \tccls{\bu}$.
Recall that $\phitc_x(\tccls{(\bv,v_{x+1})}) = \tccls{\bu}$ is equivalent to
$\bv \cdot v_{x+1} \tc \bu$. 

\begin{lemma}\label{lem:safe-pointed-finite}
 Consider $(\bu,u_{x+1})\not\tc\bot$ such that $\bu$ is "pointed".
  Then there is "pointed" $(\bv,v_{x+1})$ such that $\bv\cdot v_{x+1}=\bu$ and $(\bv,v_{x+1}u_{x+1})\not\tc\bot$.
\end{lemma}

\begin{proof}
  Assume that we can show that there is are $\bv \in (\Sigma^*)^x$, $v_{x+1} \in \Sigma^*$ with $\bv \cdot v_{x+1} \tc \bu$ and $(\bv,v_{x+1}u_{x+1})$ is "pointed" (note that the lemma claims that the tuple without $u_{x+1}$ at the end has to be "pointed"). From that we can construct the desired tuple as follows. Since $(\bv,v_{x+1}u_{x+1}) \not\tc \bot$ (because it is "pointed"), there is $w \in \Sigma^+$ such that $\bv\cdot v_{x+1}u_{x+1}w \tc \bv$ and $v_1 \cdots v_x(v_{x+1}u_{x+1}w)^\omega \in L$ iff $x+1$ is even. Then also $(\bv,v_{x+1}u_{x+1}wv_{x+1}u_{x+1}) \not\tc \bot$. Further $\underbrace{\bv\cdot v_{x+1}u_{x+1}w}_{\tc \bv}v_{x+1} \tc \bv\cdot v_{x+1} \tc \bu$  by Lemma~\ref{lem:tc-properties} and $(\bv,v_{x+1}u_{x+1}wv_{x+1})$ is "pointed" as an extension of a "pointed" element by Lemma~\ref{lem:pointed-properties}. So with $v_{x+1}' = v_{x+1}u_{x+1}wv_{x+1}$ we obtain a tuple $(\bv,v_{x+1}')$ as claimed in the lemma.

  It remains to show that we can find $\bv \in (\Sigma^*)^x$, $v_{x+1} \in \Sigma^*$ with $\bv \cdot v_{x+1} \tc \bu$ and $(\bv,v_{x+1}u_{x+1})$ "pointed". The idea is quite simple. We start with $(\bu,u_{x+1})$ as current tuple. As long as the current tuple is not "pointed", there is a witness for that: a new tuple whose last component ends in $u_{x+1}$, and whose "safe language" is different from the "safe language" of the current tuple. By Lemma~\ref{lem:tuple-safe-decrease}, the "safe language" of the new tuple must be strictly included in the "safe language" of the previous tuple. This can happen only finitely often, since the "safe language" of a tuple only depends on its $\tc$-class, and there are only finitely many $\tc$-classes. We formalise this idea below.

  Starting with $i= 0$ we define a sequence $(\bv^{(i)},v_{x+1}^{(i)})$ of tuples with $\bv^{(i)} \in (\Sigma^*)^x$ and $v_{x+1}^{(i)} \in \Sigma^*$ such that
  \begin{enumerate}[label=(\arabic*)]
  \item $\bv^{(i)}$ is "pointed",
  \item $\bv^{(i)}\cdot v_{x+1}^{(i)} \tc \bu$, 
  \item $(\bv^{(i)}, v_{x+1}^{(i)}u_{x+1}) \not\tc \bot$, and 
  \item if $i > 0$, then $\Ls(\bv^{(i)}, v_{x+1}^{(i)}u_{x+1}) \subsetneq \Ls(\bv^{(i-1)}, v_{x+1}^{(i-1)}u_{x+1})$.
  \end{enumerate}
  We show that we can extend the sequence by one more tuple if $(\bv^{(i)}, v_{x+1}^{(i)}u_{x+1})$ is not "pointed". 
  Because of the last property, each such sequence must be finite because the "safe language" of a tuple only depends on its $\tc$-class, and by Lemma~\ref{lem:tc-finite-index} there are only finitely many $\tc$ classes for each tuple length (we are assuming that $L$ is regular).
  
  We start with $\bv^{(0)} = \bu$ and $v_{x+1}^{(0)} = \e$, which is easily seen to have all the properties. If $(\bv^{(i)}, v_{x+1}^{(i)}u_{x+1})$ is "pointed", then we have found the desired tuple with $\bv := \bv^{(i)}$ and $v_{x+1} := v_{x+1}^{(i)}$. Otherwise, we show how to obtain the next tuple in the sequence.

  So assume that $(\bv^{(i)}, v_{x+1}^{(i)}u_{x+1})$ is not "pointed". Since $\bv^{(i)}$ is "pointed", there must be $\bv' \in (\Sigma^*)^x$ and $v_{x+1}' \in \Sigma^*$ such that
  \begin{enumerate}[label=(\roman*)]
  \item  $\bv'$ is "pointed", 
  \item  $\bv' \cdot v_{x+1}'\tc \bv^{(i)}$,
  \item $(\bv', v_{x+1}'v_{x+1}^{(i)}u_{x+1}) \not\tc \bot$, and
  \item $(\bv', v_{x+1}'v_{x+1}^{(i)}u_{x+1}) \not\tc (\bv^{(i)}, v_{x+1}^{(i)}u_{x+1})$.
  \end{enumerate}
  If we define $\bv^{(i+1)} = \bv'$ and $v_{x+1}^{(i+1)} :=  v_{x+1}'v_{x+1}^{(i)}$, then properties (1) and (3) for the new tuple are immediate from (i) and (iii). Property (2) follows from (ii) by $\bv^{(i+1)} \cdot v_{x+1}^{(i+1)} = \bv' \cdot v_{x+1}'v_{x+1}^{(i)} \tc \bv^{(i)}v_{x+1}^{(i)} \tc \bu$, where the second last equivalence is from (ii), and the last equivalence from (2) for $i$. Property (4) follows from (iv) and Lemma~\ref{lem:tuple-safe-decrease} because $(\bv', v_{x+1}'v_{x+1}^{(i)}u_{x+1}) \not\tc (\bv^{(i)}, v_{x+1}^{(i)}u_{x+1})$ is the same as $\Ls(\bv', v_{x+1}'v_{x+1}^{(i)}u_{x+1}) \not= \Ls(\bv^{(i)}, v_{x+1}^{(i)}u_{x+1})$, and thus by Lemma~\ref{lem:tuple-safe-decrease} $\Ls(\bv', v_{x+1}'v_{x+1}^{(i)}u_{x+1}) \subsetneq \Ls(\bv^{(i)}, v_{x+1}^{(i)}u_{x+1})$. That we can apply  Lemma~\ref{lem:tuple-safe-decrease} is guaranteed by (ii).
\end{proof}

%Instantiating Lemma~\ref{lem:safe-pointed-finite} with $u_{x+1}=\e$ yields the
%following characterisation of those states $\tccls{\bu}$ in $\Ttc_x$ that are
%not leaves. 
%\begin{lemma}\label{lem:leaf-characterisation}
%  Let $\bu \in (\Sigma^*)^x$ be "pointed". Then $(\phitc_x)^{-1}(\tccls{\bu}) \not= \emptyset$ iff $(\bu,\e) \not\tc \bot$.
%\end{lemma}
%\begin{proof}
%  Let $\bu \in (\Sigma^*)^x$ be "pointed". If $(\bu,\e) \not\tc \bot$, then we apply Lemma~\ref{lem:safe-pointed-finite} with $u_{x+1}=\e$. We obtain $\bv \in (\Sigma^*)^x$, $v_{x+1} \in \Sigma^*$ such that $(\bv,v_{x+1})$ is "pointed" and  $\bv \cdot v_{x+1} \tc \bu$. So $\tccls{(\bv,v_{x+1})}$ is a state in $\T_{x+1}$ with $\phitc_x(\tccls{(\bv,v_{x+1})}) = \tccls{\bv \cdot v_{x+1}} = \tccls{\bu}$.
%
%  For the other direction, assume that $(\phitc_x)^{-1}(\tccls{\bu}) \not= \emptyset$ and let $\bv \in (\Sigma^*)^x$, $v_{x+1} \in \Sigma^*$ such that $(\bv,v_{x+1})$ is "pointed" and  $\bv \cdot v_{x+1} \tc \bu$, that is, $\phitc_x(\tccls{(\bv,v_{x+1})}) = \tccls{\bv \cdot v_{x+1}} = \tccls{\bu}$. Let $w \in \Sigma^+$ be such that $\bv \cdot v_{x+1}w \tc \bv$, and $v_1\cdots v_x(v_{x+1}w)^\omega \in L$ iff $x+1$ is even. Since $\bv \cdot v_{x+1} \tc \bu$, we get $u_1\cdots u_x(wv_{x+1})^\omega \in L$ iff $x+1$ is even. So $wv_{x+1}$ is a witness for $(\bu,\e) \not\tc \bot$.
%\end{proof}
The following is an analogue of Lemma~\ref{lem:safe-pointed-finite} for infinite words.
\begin{lemma}\label{lem:safe-pointed-infinite}
  Let $\bu \in (\Sigma^*)^x$ be "pointed" and $w \in \Sigma^\omega$ such that $(\bu,w) \not\tc \bot$. Then there are $\bv \in (\Sigma^*)^x$, $v_{x+1} \in \Sigma^*$ with $\bv \cdot v_{x+1} \tc \bu$, $(\bv,v_{x+1})$ is "pointed" and $(\bv,v_{x+1}w) \not\tc \bot$. 
\end{lemma}
\begin{proof}
We can apply Lemma~\ref{lem:safe-pointed-finite} to each finite prefix
$u^{(i)}=u_{x+1}$ of $w$, obtaining $\bv^{(i)} \in (\Sigma^*)^x$, $v_{x+1}^{(i)}
\in \Sigma^*$ with $\bv^{(i)} \cdot v_{x+1}^{(i)} \tc \bu$,
$(\bv^{(i)},v_{x+1}^{(i)})$ is "pointed" and $(\bv^{(i)},v_{x+1}^{(i)}u^{(i)})
\not\tc \bot$. We can assume that each $(\bv^{(i)},v_{x+1}^{(i)})$ is the
shortest "pointed" tuple from its $\tc$-class, because the required properties hold for all
"pointed" tuples of a $\tc$-class or for none. Hence, there are only finitely
many candidates for the $(\bv^{(i)},v_{x+1}^{(i)})$, and thus there must be one
tuple $(\bv, v_{x+1})$ that equals $(\bv^{(i)},v_{x+1}^{(i)})$ for infinitely
many prefixes $u^{(i)}$. This tuple then satisfies all the required properties.
\end{proof}
The next lemma is used in the proof of Theorem~\ref{thm:Aatc-correct} for showing that $\Aatc$ is "consistent" and accepts the language $L$.
\begin{lemma}\label{lem:Aatc-strong-acceptance}
Let $x$ be even and $\bu \in (\Sigma^*)^x$ be "pointed" such that $u_{x+1}^\omega$ is "strongly accepted" from $\tccls{\bu}$ in $\Aatc$. Then $u_1\cdots u_xu_{x+1}^\omega \in L$. The corresponding statement holds for $x$ odd, $u_{x+1}^\omega$ "strongly rejected" and $u_1\cdots u_xu_{x+1}^\omega \not\in L$.
\end{lemma}
%\igwin{The above statement is too weak for existence of centralised sequences proof}
%\chin{Not sure. I think the argument in the existence of centralised sequences proof can be adapted to the above statement. I put a comment where it is used (in the proof of Lemma~\ref{lem:Aatc-central-sequences}).}
\begin{proof}
  We show the case $x$ even. Since we never refer to layer $x-1$, the proof for $x$ odd is symmetric.

  Since $u_{x+1}^\omega$ is "strongly accepted" from $\tccls{\bu}$, we have $\bu \cdot u_{x+1}^\omega \not\tc \bot$. Let $k \ge 0,\ell \ge 1$ be such that $\bu \cdot u_{x+1}^k \tc \bu \cdot u_{x+1}^{k+\ell}$. Assume that $(\bu \cdot u_{x+1}^k,u_{x+1}^\omega) \not\tc \bot$. Then, by Lemma~\ref{lem:safe-pointed-infinite}, there would be $\bv \in (\Sigma^*)^x$, $v_{x+1} \in \Sigma^*$ with $\bv \cdot v_{x+1} \tc \bu\cdot u_{x+1}^k$, $(\bv,v_{x+1})$ "pointed" and $(\bv,v_{x+1}u_{x+1}^\omega) \not\tc \bot$. This would witness that $u_{x+1}^\omega$ is not "strongly accepted" from $\tccls{\bu}$, as illustrated in Figure~\ref{fig:Aatc-strong-acc-proof}. Hence, we must have $(\bu \cdot u_{x+1}^k,u_{x+1}^\omega) \tc \bot$, which means that $(\bu \cdot u_{x+1}^k,u_{x+1}^m) \tc \bot$ for some $m$.   By iterated application of $\bu \cdot u_{x+1}^k \tc \bu \cdot u_{x+1}^{k+\ell}$, we obtain $\bu \cdot u_{x+1}^{k+\ell \cdot m} \tc \bu \cdot u_{x+1}^k$. By definition of $(\bu \cdot u_{x+1}^k,u_{x+1}^m) \tc \bot$, this implies that $u_1 \cdots u_xu_{x+1}^\omega = u_1 \cdots u_xu_{x+1}^k(u_{x+1}^{\ell \cdot m})^\omega \in L$ because $x$ is even. 
\end{proof}
\begin{figure}
  \begin{center}
    \includegraphics[]{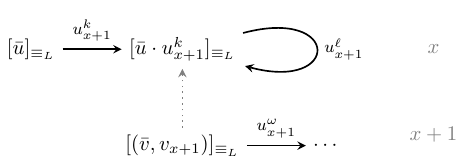}
  \end{center}
  \caption{Illustration of the contradiction in the proof of Lemma~\ref{lem:Aatc-strong-acceptance} under the assumption that $(\bu \cdot v^k,v^\omega) \not\tc \bot$.\label{fig:Aatc-strong-acc-proof}}
\end{figure}

\begin{theorem}\label{thm:Aatc-correct}
  Let $L$ be an $\omega$-regular language. Then $\Aatc$ is a "consistent" %, and hence "semantically deterministic", 
  "layered automaton" with $L(\Aatc) = L$.
\end{theorem}
\begin{proof}
  Assume that $\Aatc$ is not "uniformly semantically deterministic" and thus not "consistent" by Corollary~\ref{cor:consistent-iff-unifSD}. Then there is an even $x$ and an odd $y$ and "pointed" tuples $\bu \in (\Sigma^*)^x$, $\bv \in (\Sigma^*)^y$ with $u_1 \cdots u_x \sim_L v_1 \cdots v_y$ and such that some $w \in  \Sigma^\omega$ is "strongly accepted" from $\tccls{\bu}$ and "strongly rejected" from $\tccls{\bv}$. Clearly, the set of all these $w$ is $\omega$-regular, so there is also an ultimately periodic $w = uv^\omega$ with this property. Then $v^\omega$ is "strongly accepted" from $\tccls{\bu \cdot u}$ and "strongly rejected" from  $\tccls{\bv \cdot u}$. By Lemma~\ref{lem:Aatc-strong-acceptance} we get that $u_1 \cdots u_xuv^\omega \in L$ and $v_1 \cdots v_yuv^\omega \notin L$, which contradicts $u_1 \cdots u_x \sim_L v_1 \cdots v_y$. We conclude that $\Aatc$ is "consistent". %and thus "semantically deterministic".

  For showing that $L(\Aatc) = L$, let $uv^\omega$ be an ultimately periodic word. By Corollary~\ref{cor:equivalences-acceptance}, $uv^\omega$ contains a suffix that is "strongly accepted" or "strongly rejected" by some state in  $\Aatc$, and not both because we have already shown that $\Aatc$ is "consistent". Then $v^\omega$ is "strongly accepted" or "strongly rejected" from some $\tccls{\bu}$ for a "pointed" tuple $\bu = (u_1, \ldots, u_x) \in \Sigma^x$ such that $(u_1 \cdots u_x)$ is reachable on layer~1 from $(\e)$ by a prefix $uv^k$ of $uv^\omega$. Since all transitions in $\Aatc$ respect the $\sim_L$ class, we get that $u_1 \cdots u_x \sim_L uv^k$. By Lemma~\ref{lem:Aatc-strong-acceptance} we get that $uv^kv^\omega \in L$ iff $x$ is even. Hence, $L(\Aatc)$ and $L$ contain the same ultimately periodic words and are thus equal (this follows from the closure properties of $\omega$-regular languages, see e.g.\ \cite[Fact~1]{CalbrixNP93}).
\end{proof}

%%%%%%%%%%%%%%%%%%%%%%%%%%%%%%%%%%%%%%%%%%%%%%%%%%%%%%%%%%%%
\subsection{Minimality of $\Aatc$}
We show that $\Aatc$ is "safe minimal", "normal", and "centralised" (Lemma~\ref{lem:Aatc-minimal-properties}) and then use Theorem~\ref{thm:concise-are-unique} to conclude that $\Aatc$ is the minimal "layered automaton" for $L$.

We start by showing that $\Aatc$ has "central sequences". The proof is split into three lemmas: Lemma~\ref{lem:central-from-bot-child} shows how to obtain a "central sequence" for each "pointed" class that has a child that is $\bot$. Lemma~\ref{lem:leaf-characterisation} shows that all children of leafs of $\Aatc$ are $\bot$, and in the proof of Lemma~\ref{lem:Aatc-central-sequences} we show that Lemma~\ref{lem:central-from-bot-child} can also be applied to the inner states of $\Aatc$.

\begin{lemma}\label{lem:central-from-bot-child}
Let $x\ge 2$, $\bu \in (\Sigma^*)^{x-1}$ and $u_x,u_{x+1} \in \Sigma^*$ such that $(\bu,u_x)$ is pointed, $(\bu,u_x,u_{x+1}) \tc \bot$ and $(\bu,u_xu_{x+1}) \not\tc \bot$. Then $\tccls{(\bu,u_xu_{x+1})}$ has a "central sequence".
\end{lemma}
\begin{proof}
  Since $(\bu,u_xu_{x+1}) \not\tc \bot$, there is $w$ with $\bu \cdot u_xu_{x+1}w \tc \bu$ and $u_1 \cdots u_{x-1}(u_xu_{x+1}w)^\omega \in L$ iff $x$ even. We show that $z := wu_xu_{x+1}$ is a "central sequence" for $\tccls{(\bu,u_xu_{x+1})}$.
  
  First, every child of $\tccls{\bu \cdot u_xu_{x+1}}$ is taken to $\tccls{(\bu,u_xu_{x+1})}$ or $\bot$ by $z$, which follows from $(\bu,u_xu_{x+1})$ being pointed: Let $(\bv,v_x)$ be "pointed" such that $\bv \cdot v_x \tc \bu\cdot u_xu_{x+1}$. Then by Lemma~\ref{lem:tc-properties}, $\bv  \cdot v_xw \tc \bu\cdot u_xu_{x+1}w \tc \bu$ and thus
  \[
  (\bv,v_xwu_xu_{x+1}) \tc \{(\bu,u_xu_{x+1}),\bot\}.
  \]
  Second, by choice of $w$, $(\bu,u_xu_{x+1}wu_xu_{x+1}) \not \tc \bot$, and hence $z$ loops on $\tccls{(\bu,u_xu_{x+1})}$.

  Third, every grand child of $\tccls{\bu \cdot u_xu_{x+1}}$ is taken to $\bot$ by $z$: Let  $(\bv,v_x,v_{x+1})$ be "pointed" such that $\bv \cdot v_xv_{x+1} \tc \bu\cdot u_xu_{x+1}$. Then again by Lemma~\ref{lem:tc-properties}, $\bv  \cdot v_xv_{x+1}w \tc \bu\cdot u_xu_{x+1}w \tc \bu$. Since $(\bu,u_x)$ is pointed, we get $(\bv,v_xv_{x+1}wu_x) \tc \{(\bu,u_x),\bot\}$. If $(\bv,v_xv_{x+1}wu_x) \tc \bot$, then clearly $(\bv,v_xv_{x+1}z) = (\bv,v_xv_{x+1}wu_xu_{x+1}) \tc \bot$ by Lemma~\ref{lem:bot-properties}. If  $(\bv,v_xv_{x+1}wu_x) \tc (\bu,u_x)$, then by Lemma~\ref{lem:tuple-safe-decrease}
  \[
  \Ls((\bv,v_x,v_{x+1}z)) = \Ls((\bv,v_x,v_{x+1}wu_xu_{x+1})) \subseteq \Ls((\bu,u_x,u_{x+1})) = \emptyset.
  \]
This means that $(\bv,v_x,v_{x+1}z) \tc \bot$, and thus $z$ satisfies all properties of a "central sequence" for $\tccls{(\bu,u_xu_{x+1})}$.
\end{proof}

We now give that characterisation of leaf states by showing that all tuples that extend a "pointed" tuple in a leaf class by another component are $\bot$. The proof is a simple instantiation of Lemma~\ref{lem:safe-pointed-finite}.
\begin{lemma}\label{lem:leaf-characterisation}
  Let $x \ge 1$ and $\bu \in (\Sigma^*)^x$ be "pointed". Then $\tccls{\bu}$ is a leaf of $\Aatc$ iff $(\bu,\e) \tc \bot$.
\end{lemma}
\begin{proof}
  Let $\bu \in (\Sigma^*)^x$ be "pointed". If $(\bu,\e) \not\tc \bot$, then we apply Lemma~\ref{lem:safe-pointed-finite} with $u_{x+1}=\e$. We obtain $\bv \in (\Sigma^*)^x$, $v_{x+1} \in \Sigma^*$ such that $(\bv,v_{x+1})$ is "pointed" and  $\bv \cdot v_{x+1} \tc \bu$. So $\tccls{(\bv,v_{x+1})}$ is a state in $\Ttc_{x+1}$ with $\phitc_x(\tccls{(\bv,v_{x+1})}) = \tccls{\bv \cdot v_{x+1}} = \tccls{\bu}$. Hence $\tccls{\bu}$ is not a leaf.

  For the other direction, assume that $\tccls{\bu}$ is not a leaf and let $\bv \in (\Sigma^*)^x$, $v_{x+1} \in \Sigma^*$ such that $(\bv,v_{x+1})$ is "pointed" and  $\bv \cdot v_{x+1} \tc \bu$, that is, $\phitc_x(\tccls{(\bv,v_{x+1})}) = \tccls{\bv \cdot v_{x+1}} = \tccls{\bu}$. By Lemma~\ref{lem:tuple-safe-decrease} with $u_{x+1} = \e$ we get that $\Ls(\bv,v_{x+1}) \subseteq \Ls(\bu,\e)$. Since $\Ls(\bv,v_{x+1}) \not= \emptyset$, we also have $\Ls(\bu,\e) \not= \emptyset$, which means that $(\bu,\e) \not\tc \bot$.
%Let $w \in \Sigma^+$ be such that $\bv \cdot v_{x+1}w \tc \bv$, and $v_1\cdots v_x(v_{x+1}w)^\omega \in L$ iff $x+1$ is even. Since $\bv \cdot v_{x+1} \tc \bu$, we get $u_1\cdots u_x(wv_{x+1})^\omega \in L$ iff $x+1$ is even. In combination with
%  \[
%  \bu \cdot wv_{x+1} \tc  \overbrace{\underbrace{\bv \cdot v_{x+1}}_{{}\tc \bu} w}^{{}\tc\bv} v_{x+1} \tc \bv \cdot v_{x+1} \tc \bu
%  \]
%we see that $wv_{x+1}$ is a witness for $(\bu,\e) \not\tc \bot$.
\end{proof}

\begin{lemma}\label{lem:Aatc-central-sequences}
In $\Aatc$ all states on layers $x \ge 2$ have a "central sequence".
\end{lemma}
\begin{proof}
  If $\tccls{(\bu,u_x)}$ is a leaf of $\Aatc$, then $(\bu,u_x,\e) \tc \bot$ by Lemma~\ref{lem:leaf-characterisation}, and thus $\tccls{(\bu,u_x)}$ has a "central sequence" by Lemma~\ref{lem:central-from-bot-child}.

  For showing that internal states have a central sequence, let $(\bv,v_x,v_{x+1})$ be "pointed" and show that its parent $\tccls{(\bv,v_xv_{x+1})}$ has a central sequence, assuming that $\tccls{(\bv,v_x,v_{x+1})}$ has a "central sequence" $z$. We want to show that we can apply Lemma~\ref{lem:central-from-bot-child}, so we construct a word $z'$ that loops on $\tccls{(\bv,v_xv_{x+1})}$ and takes $(\bv,v_x,v_{x+1})$ to $\bot$, so it has the following properties:
  \begin{enumerate}[label=(\roman*)]
  \item\label{item:zp-x-loop} $(\bv,v_xv_{x+1}z') \tc (\bv,v_xv_{x+1})$ and
  \item\label{item:zp-xp1-bot} $(\bv,v_x,v_{x+1}z') \tc \bot$.
  \end{enumerate}
  We can apply Lemma~\ref{lem:central-from-bot-child} to $(\bv,v_x,v_{x+1}z')$ and obtain a "central sequence" for $\tccls{(\bv,v_xv_{x+1}z')} = \tccls{(\bv,v_xv_{x+1})}$ as desired. 

  The following arguments for establishing \ref{item:zp-x-loop} and \ref{item:zp-xp1-bot} are not very difficult but might be confusing because they involve several tuples. Figure~\ref{fig:Aatc-central-sequences} gives an overview of the situation that is constructed below. The tuple on the left-hand side on layer $x$ is the one for which we want to construct the word $z'$. 
  
  We start by explaining the $z$-edges in Figure~\ref{fig:Aatc-central-sequences}. Since $z$ is central for $\tccls{(\bv,v_x,v_{x+1})}$, every child of $\tccls{(\bv,v_x,v_{x+1})}$ is taken to $\bot$ by $z$ (this is indicated by $\tccls{(\cdots)}$ on layer $x+2$ in Figure~\ref{fig:Aatc-central-sequences}). Further, we have $(\bv,v_x,v_{x+1}z) \tc (\bv,v_x,v_{x+1})$ and thus also $(\bv,v_xv_{x+1}z) \tc (\bv,v_xv_{x+1})$ and $\bv\cdot v_xv_{x+1}z \tc \bv \cdot v_xv_{x+1}$ by Lemma~\ref{lem:tc-properties} (these are the $z$-loops in Figure~\ref{fig:Aatc-central-sequences}). Hence, $(\bv,v_xv_{x+1}z) \not\tc \bot$, and there is $w$ such that $(\bv,v_xv_{x+1}zw) \tc \bv$ and $v_1 \cdots v_{x-1}(v_xv_{x+1}zw)^\omega \in L$ iff $x$ is even. This is the $w$-edge on layer $x-1$ in Figure~\ref{fig:Aatc-central-sequences}. Further, the choice of $w$ ensures that  $(\bv,v_xv_{x+1}(zwv_xv_{x+1})^\omega) \not\tc \bot$.

  We let $z' := zwv_xv_{x+1}$. The paths induced by $z'$ on the layers $x-1$, $x$, and $x+1$ are shown in Figure~\ref{fig:Aatc-central-sequences}. Most edges shown on the layers are just obtained by appending the corresponding word to the last component of the corresponding tuple. The parent edges are obtained by the morphism property. There are two thick edges on the layers for which the target class is determined by the definition of "pointed". First, consider the thick $v_x$-edge on layer $x$.  
  Since $(\bv,v_x)$ is "pointed" and $\bv \cdot v_xv_{x+1}zw \tc \bv$ (as can be seen in Figure~\ref{fig:Aatc-central-sequences}), we get that $(\bv,v_xv_{x+1}zwv_x) \tc \{(\bv,v_x),\bot\}$. And as noted earlier, the choice of $w$ ensures that $(\bv,v_xv_{x+1}(zwv_xv_{x+1})^\omega) \not\tc \bot$, so in particular $(\bv,v_xv_{x+1}zwv_x) \not\tc \bot$. Hence, we obtain that $(\bv,v_xv_{x+1}zwv_x) \tc (\bv,v_x)$ and thus also
  \[
  (\bv,v_xv_{x+1}z') = (\bv,v_xv_{x+1}zwv_xv_{x+1}) \tc (\bv,v_xv_{x+1}).
  \]
This explains the edges on layer $x$ in Figure~\ref{fig:Aatc-central-sequences} and shows \ref{item:zp-x-loop}.

We now show \ref{item:zp-xp1-bot}. Actually, the thick $v_{x+1}$-edge on layer $x+1$ is drawn under the assumption that \ref{item:zp-xp1-bot} does not hold, and is then used for a contradiction.
Since $(\bv,v_x,v_{x+1})$ is "pointed" and $(\bv,v_xv_{x+1}zwv_x) \tc (\bv,v_x)$ (as established before), we get that $(\bv,v_x,v_{x+1}zwv_xv_{x+1}) \tc \{(\bv,v_x,v_{x+1}),\bot\}$. If $(\bv,v_x,v_{x+1}zwv_xv_{x+1}) \tc \bot$, then \ref{item:zp-xp1-bot} holds. Otherwise, we get $(\bv,v_x,v_{x+1}zwv_xv_{x+1}) \tc (\bv,v_x,v_{x+1})$, which is indicated by the thick $v_{x+1}$-edge on layer $x+1$. Now it can be seen from Figure~\ref{fig:Aatc-central-sequences}, that in this case, $(zwv_xv_{x+1})^\omega$ is strongly accepted/rejected from $\tccls{(\bv,v_x,v_{x+1})}$ iff $x+1$ is even/odd. By Lemma~\ref{lem:Aatc-strong-acceptance}, this would mean that $v_1 \cdots v_{x+1}(zwv_xv_{x+1})^\omega \in L$ iff $x+1$ is even. But we have chosen $w$ such that $v_1 \cdots v_{x-1}(v_xv_{x+1}zw)^ \omega = v_1 \cdots v_{x+1}(zwv_xv_{x+1})^\omega \in L$ iff $x$ is even, a contradiction.
\end{proof}

\begin{figure}
  \begin{center}
  \begin{tikzpicture}[
      parent/.style={
        dotted, ->,thick, >=stealth, draw=black!50, font=\footnotesize
      },
      layer/.style={
        ->, >=stealth, draw=black, font=\footnotesize
      },
      layerp/.style={
        ->,very thick, >=stealth, draw=black, font=\footnotesize
      },
      iso/.style={
        dashed, ->,thick, >=stealth, draw=black!20, text=black!20, font=\footnotesize
      },
      mloop/.style={
        out=175, in=185, looseness=6
      }
    ]\footnotesize
    % nodes for classes
    \node                        (bvvxvxp1) {$\tccls{\bv \cdot v_xv_{x+1}}$};
    \node [right=2 of bvvxvxp1]   (bv) {$\tccls{\bv}$};
    \node [right=2 of bv]         (bvvx) {$\tccls{\bv \cdot v_x}$};
    \node [below= of bvvxvxp1]   (bv-vxvxp1) {$\tccls{(\bv, v_xv_{x+1})}$};
    \node [below= of bv]         (bv-vxvxp1zw) {$\tccls{(\bv, v_xv_{x+1}zw)}$};
    \node [below= of bvvx]       (bv-vx) {$\tccls{(\bv, v_x)}$};
    \node [below= of bv-vxvxp1]  (bv-vx-vxp1) {$\tccls{(\bv, v_x,v_{x+1})}$};
    \node [below= of bv-vxvxp1zw](bv-vx-vxp1zw) {$\tccls{(\bv, v_x,v_{x+1}zw)}$};
    \node [below= of bv-vx]      (bv-vx-vxp1zwvx) {$\tccls{(\bv, v_x,v_{x+1}zwv_x))}$};
    \node [below= of bv-vx-vxp1] (xp2) {$\tccls{(\cdots)}$};
    \node [left= of xp2.center] (bot) {$\bot$};
    % nodes for layer numbers
    \node[left=3 of bvvxvxp1.center,text=black!50] () {$x-1$};
    \node[left=3 of bv-vxvxp1.center,text=black!50] () {$x$};
    \node[left=3 of bv-vx-vxp1.center,text=black!50] () {$x+1$};
    \node[left=3 of xp2.center,text=black!50] () {$x+2$};

    % layer edges
    \path[layer] (bvvxvxp1) edge[mloop] node[above]{$z$} (bvvxvxp1);
    \path[layer] (bv-vxvxp1) edge[mloop] node[above]{$z$} (bv-vxvxp1);
    \path[layer] (bv-vx-vxp1) edge[mloop] node[above]{$z$} (bv-vx-vxp1);
    \path[layer] (bvvxvxp1) edge node[above]{$w$} (bv);
    \path[layer] (bv-vxvxp1) edge node[above]{$w$} (bv-vxvxp1zw);
    \path[layer] (bv-vx-vxp1) edge node[above]{$w$} (bv-vx-vxp1zw);
    \path[layer] (bv) edge node[above]{$v_x$} (bvvx);
    \path[layerp] (bv-vxvxp1zw) edge node[above]{$v_x$} (bv-vx);
    \path[layer] (bv-vx-vxp1zw) edge node[above]{$v_x$} (bv-vx-vxp1zwvx);
    \path[layer] (bvvx) edge[bend right=17] node[above,pos=0.3]{$v_{x+1}$} (bvvxvxp1);
    \path[layer] (bv-vx) edge[bend right=17] node[above,pos=0.3]{$v_{x+1}$} (bv-vxvxp1);
    \path[layerp] (bv-vx-vxp1zwvx) edge[bend right=17] node[above,pos=0.3]{$v_{x+1}$} (bv-vx-vxp1);
    \path[layer] (xp2) edge node[above]{$z$} (bot);
%    \path[layer] (bu) edge node[above]{$v_{x+1}$} (buuxp1);
%    \path[layer] (a2) edge[bend right]  node[above]{$v_{x+1}$} (bu-uxp1);
    %    \path[layer] (a1) edge[bend left] node[above]{$v_{x+1}$} (bot);

    % parent edges
    \path[parent] (bv-vx-vxp1) -- (bv-vxvxp1);
    \path[parent] (bv-vxvxp1) -- (bvvxvxp1);
    \path[parent] (bv-vx-vxp1zw) -- (bv-vxvxp1zw);
    \path[parent] (bv-vxvxp1zw) -- (bv);
    \path[parent] (bv-vx-vxp1zwvx) -- (bv-vx);
    \path[parent] (bv-vx) -- (bvvx);
    \path[parent] (xp2) -- (bv-vx-vxp1);
  \end{tikzpicture}

  \end{center}
  \caption{Illustration supporting the reasoning in the proof of Lemma~\ref{lem:Aatc-central-sequences}. Dotted edges represent the morphism to the parent layer. The $v_{x+1}$-edge on layer $x+1$ is drawn under an assumption that leads to a condtradiction, so the $wv_xv_{x+1}$-loop on layer $x+1$ does not exist.  \label{fig:Aatc-central-sequences}}
\end{figure}
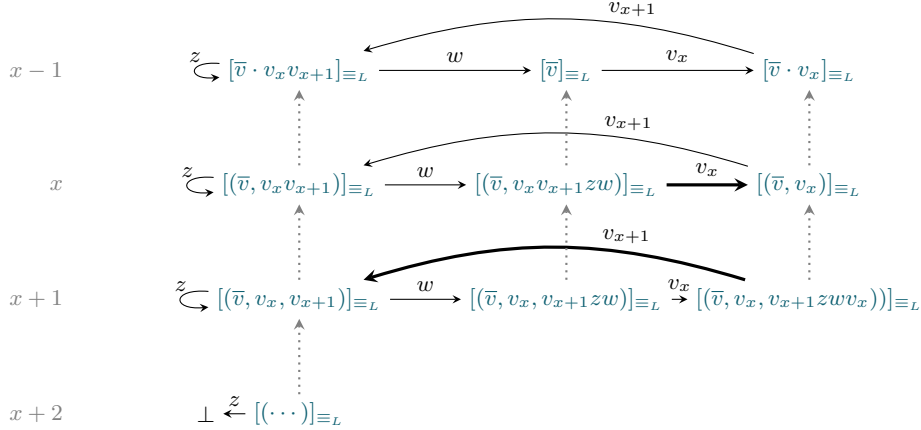

\begin{lemma}\label{lem:Aatc-minimal-properties}
Let $L$ be an $\omega$-regular language. Then $\Aatc$ is "safe minimal", "normal", and "centralised".
\end{lemma}
\begin{proof}
  A layered automaton is "safe minimal" if its first layer corresponds to residual equivalence of $L$, and for all $x \ge 2$ and for all states $p,q$ on layer $x$ that have the same parent, the "safe language" of $p$ and $q$ are different. So let $(\bu,u_x)$ and $(\bv,v_x)$ be two pointed tuples with $\bu \cdot u_x \tc \bv \cdot v_x$. Then, by definition of $\tc$, we get $(\bu,u_x) \tc (\bv,v_x)$ iff $\Ls(\bu,u_x) = \Ls(\bv,v_x)$. So two pointed tuples on layer $x$ with the same parent class correspond to different states iff they have a different "safe language".

  For showing that $\Aatc$ is normal, we first need to show property N1 from Definition \ref{it:normal-return}, namely that each layer $x \ge 2$ is a union of non-trivial SCCs. So let $(\bu,u_x)$ be some pointed tuple on level $x$, and let $v \in \Sigma^*$ be such that $(\bu,u_xv) \not\tc \bot$, that is, in $\Aatc$ there is the path $\tccls{(\bu,u_x)} \xrightarrow{v} \tccls{(\bu,u_xv)}$. Since $(\bu,u_xv) \not\tc \bot$, there is $w \in \Sigma^+$ such that $\bu \cdot u_xvw \tc \bu$ and $u_1 \cdots u_{x-1}(u_xvw)^\omega \in L$ iff $x$ is even. Then $(\bu,u_xvwu_x) \not\tc \bot$ and since $(\bu,u_x)$ is "pointed", we get $(\bu,u_xvwu_x) \tc (\bu,u_x)$, so in $\Aatc$ we get the following path $\tccls{(\bu,u_xv)} \xrightarrow{wu_x} \tccls{(\bu,u_x)}$. 

  The property "N2" required for "normal" follows from the existence of central sequences (the central sequence of a state has the properties required for "N2"). That $\Aatc$ is "centralised" also follows from the existence of "central sequences": If $\tccls{\bu} \fleq_{x} \tccls{\bv}$ for two tuples of length $x \ge 2$, then a central sequence $z$ for $\tccls{\bu}$ cannot send $\tccls{\bv}$ to $\bot$, and hence $\tccls{\bv \cdot z} = \tccls{\bu}$, which means $\tccls{\bv} \run{z}_x \tccls{\bu}$. By property N1 of "normal", $\tccls{\bu}$ and $\tccls{\bv}$ are in the same $x$-SCC.
\end{proof}

Combining Lemma~\ref{lem:Aatc-minimal-properties}, Theorem~\ref{thm:Aatc-correct},  and Theorem~\ref{thm:concise-are-unique} yields the main result of this section.
\begin{theorem}\label{thm:Aatc-minimal}
Let $L$ be an $\omega$-regular language. Then $\Aatc$ is the minimal "layered automaton" for $L$.
\end{theorem}

\section{Conclusions}\label{sec:conclusions}
We have introduced "layered automata", a model merging and generalising minimal
"HD" coBüchi automata and Zielonka trees. This model provides canonical
representations for "$\omega$-regular languages". We have characterised these
representations via %language-dependent
machine-independent congruences over tuples of finite words.
Beyond this foundational characterisation, "consistent" "layered automata"
exhibit strong algorithmic properties. 
%"Canonical" "layered automata" are "consistent". %%Antonio: I do not see this as a strenght, but as a restriction
%Consistency of a "layered automaton" can be checked in PTIME, moreover 
Emptiness and inclusion testing are in PTIME for "consistent" "layered automata".
Furthermore, "consistent" "layered automata" can be minimised in PTIME, and the
result is never bigger than the smallest equivalent deterministic automaton.
These properties make "consistent" "layered automata" an attractive model of
acceptors for $\w$-regular languages as they have good
algorithmic properties and at the same time are both "history deterministic" and "0-1
probabilistic".

While we establish that canonical "layered
automata" are minimal within the class of "consistent" "layered automata"
(Corollary~\ref{cor:canonicity}), we conjecture they are also minimal among
all "alternating" "history deterministic" automata. 

\begin{conjecture}\label{conj:minimality-HD}
Let $\Bb$ be a "history deterministic" "alternating parity automaton".
Then, there is an equivalent "layered automaton" $\Aa$ such that $\sem{\Aa}$ has no more states than $\Bb$.

In particular, the minimal "layered automaton" of a language is minimal among all "history deterministic" "alternating parity automata".\footnote{We have a draft of a proof for Conjecture~\ref{conj:minimality-HD} for the case of 2 priorities (B\"uchi and coB\"uchi automata). That is, the minimal "HD" "coB\"uchi automaton" from Abu Radi and Kupferman~\cite{Abu.Kup.Minimization2022} is minimal among all "HD" "alternating" "coB\"uchi automata".}
\end{conjecture}
If true, this conjecture would establish that "layered automata" provide a 
canonical and algorithmically efficient representation for all "history deterministic" automata.

\subparagraph{Other future directions.}
Naturally, there are many aspects of the theory of "layered automata" itself
that remain to be developed. %Beyond this foundational work, however, 
We believe
that "layered automata" offer a promising framework for broader practical and
theoretical applications. A natural first direction is to provide efficient
constructions for Boolean operations on languages and direct translations from
logical formalisms such as Linear Temporal Logic to "layered automata". A second
avenue is to build the minimal "layered automaton" directly from COCOA
representations using techniques from~\cite{Ehl.Kha.Fully2024}, thereby
circumventing the full product construction discussed in
Section~\ref{subsec:comparison}. Finally, we envision applying "layered
automata" to passive and active learning of omega-regular languages, extending
the passive learning results for coBüchi
automata~\cite{Lod.Wal.Congruences2025}.

\newpage
\bibliography{Bib-history-determinism}

\end{document}